\documentclass[11pt,a4paper]{article}
\usepackage[margin=2.5cm]{geometry}
\usepackage{amsmath,amssymb,amsthm}
\usepackage{microtype}
\PassOptionsToPackage{hyphens}{url}
\usepackage[hidelinks]{hyperref}
\usepackage{enumitem}
\usepackage{booktabs}
\usepackage{multirow}
\usepackage{algorithm}
\usepackage{algpseudocode}
\usepackage{placeins}
\usepackage{graphicx}
\usepackage{url}

\newcommand{\PX}{P_X}
\newcommand{\QX}{Q_X}
\newcommand{\chisq}{\chi^2(\QX\,\|\,\PX)}
\newcommand{\chisqZ}{\chi^2(Q_Z\,\|\,P_Z)}
\newcommand{\Lp}[1]{L^p(#1)}
\newcommand{\norm}[1]{\left\|#1\right\|}
\newcommand{\abs}[1]{\left|#1\right|}
\newcommand{\Cov}{\operatorname{Cov}}
\newcommand{\E}{\mathbb{E}}
\newcommand{\Prob}{\mathbb{P}}
\newcommand{\R}{\mathbb{R}}
\newcommand{\Rbar}{\overline{\mathbb{R}}}
\newcommand{\TV}{\operatorname{TV}}
\newcommand{\KL}{\operatorname{KL}}
\newcommand{\Cset}{\widehat{\mathcal C}}
\newcommand{\Zc}{\mathcal Z}
\newcommand{\Qt}{\widetilde Q}
\newcommand{\ind}{\mathbf 1}
\DeclareMathOperator{\Qf}{Q}

\theoremstyle{plain}
\newtheorem{theorem}{Theorem}[section]
\newtheorem{proposition}[theorem]{Proposition}
\newtheorem{lemma}[theorem]{Lemma}
\newtheorem{corollary}[theorem]{Corollary}
\theoremstyle{definition}
\newtheorem{definition}[theorem]{Definition}
\newtheorem{example}[theorem]{Example}
\newtheorem{assumption}{Assumption}
\theoremstyle{definition}
\newtheorem{remark}[theorem]{Remark}

\title{Sketched Calibration for Conformal Prediction\\ under Covariate Shift}
\author{Mehrdad Pournaderi\\ Mofid Securities\\ \texttt{m.pournaderi@emofid.com}}
\date{}

\begin{document}
\maketitle

\begin{abstract}
Weighted conformal prediction corrects for covariate shift by reweighting calibration scores with the likelihood ratio between target and source covariates. Its cost grows with the chi-square divergence between the two covariate laws, typically exponentially in the size of the shift, and much of it can be paid for shift in directions that do not affect the response. We propose \emph{sketched calibration}: weighted conformal prediction with the ratio of compressed covariates $Z=T(X)$, used in the weights and optionally in the score. Compression never increases the shift-dependent calibration cost, and target coverage is at least $1-\alpha-\Delta_T$, where the \emph{leakage} $\Delta_T$ measures how much of the discarded shift reappears as a change in the law of the response, or of the score, given $Z$. The leakage is a covariance between the discarded shift and the response on the fibres of the sketch; it vanishes when the sketch is sufficient or retains the shift, and a minimax construction shows that no threshold rule for a fixed score can avoid it. In a heteroscedastic simulation where unweighted calibration fails, a one-dimensional sketch cuts the fraction of infinite prediction sets from $37.8\%$ to $8.0\%$ at mean coverage $92.3\%$ for a $90\%$ target; nonlinear and real-data experiments show when learned sketches succeed and when they fail.
\end{abstract}

\section{Introduction}

Split conformal prediction \cite{papadopoulos2002,vovk2005} turns any fitted score into a prediction set with finite-sample marginal coverage, provided calibration and test data are exchangeable. Under \emph{covariate shift}, where training and calibration pairs follow $\PX\otimes P_{Y|X}$ but the test pair follows $\QX\otimes P_{Y|X}$, exchangeability fails, and Tibshirani et al.\ \cite{tibshirani2019} showed that coverage is restored by reweighting the calibration scores with the likelihood ratio $w=\mathrm d\QX/\mathrm d\PX$. This is one of the most used tools in distribution-free inference, and also one of the most fragile in practice. The weight is a $d$-dimensional density ratio, hard to estimate when $d$ is moderate \cite{sugiyama2012}; and even when $w$ is known exactly, the calibration error of the weighted rule is of order $\sqrt{(1+\chisq)/m}$, where $\chisq=\E_{\PX}[w^2]-1$ is the chi-square divergence and $m$ the calibration size, and no threshold rule for a fixed score can do better on every problem (Theorem~\ref{thm:price} below; sharper constants are in \cite{isaev2026,pournaderi2026}). The weighted rule therefore behaves as if it had $m/(1+\chisq)$ calibration points, the population \emph{effective sample size} of importance weighting \cite{polo2022,gretton2009}. For a Gaussian mean shift of size $\mu$ the divergence is $e^{\mu^2}-1$, and for a product shift across $d$ coordinates it is $\prod_j(1+\chi^2_j)-1$, exponential in $d$ \cite[Sec.\ 4.1]{polo2022}. The symptom is familiar: the weighted threshold places mass $w(x)/(\sum_iw(X_i)+w(x))$ on $+\infty$, and whenever this exceeds $\alpha$ the prediction set is the whole real line. In the experiments below, with $\chisq\approx500$ and $m=500$, this happens for $38\%$ of test points.

\paragraph{A simple example.}
Let $X=(X_1,X_2)$ have independent standard Gaussian coordinates under the source, let $Y=X_1+\varepsilon$ with $\varepsilon$ independent noise, and let the target shift only the second coordinate, $X_2\sim N(\mu,1)$. The likelihood ratio is $w(x)=\exp(\mu x_2-\mu^2/2)$, its chi-square divergence is $e^{\mu^2}-1$, and full weighting reacts to the shift in full: at $\mu=2.5$ the divergence exceeds $500$ and the effective calibration sample is a few points. Yet the shift is irrelevant for prediction. Retaining $Z=X_1$ and discarding $X_2$ leaves the joint law of $(Z,Y)$ unchanged between source and target, so conformal prediction on $(Z,Y)$ is exactly valid with no weighting at all and uses the full calibration sample. The chi-square price was being paid for a direction that does not matter. The natural remedy is to calibrate on a compressed covariate $Z=T(X)$ that retains what matters; we call this \emph{sketched calibration}, with the linear case $Z=\Pi X$, $\Pi\in\R^{k\times d}$, as the leading example.

\paragraph{The danger.}
Discarding covariates is not free, and the subtlety is easy to miss: \emph{a covariate shift in $X$ is not a covariate shift in $Z$}.

\begin{example}[Two points]\label{ex:toy}
Let $\mathcal X=\{x_1,x_2\}$ with $T(x_1)=T(x_2)=z$, $P_{Y|X=x_1}=N(0,1)$, $P_{Y|X=x_2}=N(3,1)$, $\PX=\tfrac12\delta_{x_1}+\tfrac12\delta_{x_2}$, $\QX=\delta_{x_1}$. Then $w(x_1)=2$, $w(x_2)=0$, $\chisq=1$. The sketch collapses both points to $z$, so $P_Z=Q_Z=\delta_z$: in the sketched space there is \emph{no} covariate shift. But the law of $Y$ given $Z=z$ is $\tfrac12N(0,1)+\tfrac12N(3,1)$ under the source and $N(0,1)$ under the target. The sketched rule calibrates to the mixture, while the target response is $N(0,1)$; the discrepancy, $\TV\bigl(N(0,1),\tfrac12N(0,1)+\tfrac12N(3,1)\bigr)=\Phi(1.5)-\tfrac12\approx0.43$ in total variation, cannot be removed by any reweighting in $z$, because $z$ carries no information. The covariate shift has become a shift in $Y\mid Z$.
\end{example}

In general the conditional law of $Y$ given $Z=z$ is a mixture of $P_{Y|X=x}$ over the fibre $\{T(x)=z\}$, mixed according to the conditional law of $X$ given $Z$, and that fibre law differs between source and target. The part of the covariate shift that reappears as a shift of $Y\mid Z$ is what we call the \emph{leakage} of the sketch. The two examples together set the terms: compression can remove the whole chi-square price (first example) or convert it into an uncorrectable bias (second), and which happens depends on how the discarded shift interacts with the response.

\paragraph{This paper.}
Reduced-dimension importance weighting is an established idea in covariate-shift adaptation: Stojanov et al.\ \cite{stojanov2019} weight by the density ratio of a learned response-sufficient projection, and Polo and Vicente \cite{polo2022} prove that a linear reduction can only increase the effective sample size. What is new here is the analysis inside conformal prediction, where the question is whether the finite-sample coverage guarantee survives and exactly what is lost when the reduction is not sufficient. Three contributions answer it.
\begin{enumerate}[leftmargin=1.8em,itemsep=2pt]
\item \textbf{A coverage--calibration trade-off (Theorem~\ref{prop:main}).} Weighted conformal prediction run on $(T(X_i),Y_i)$ with the sketched ratio $w_Z=\mathrm dQ_Z/\mathrm dP_Z$ has target coverage at least $1-\alpha-\Delta_T$, where the leakage $\Delta_T$ is the average total-variation shift of $Y\mid Z$. The sketched ratio is the $L^2$-projection of $w$ onto sketch-measurable functions, so the sketched divergence, the supremum of a bounded ratio and, for $\alpha\le\tfrac12$, the leading variance proxy of the sharpest known calibration bound do not increase. The leakage is zero when the sketch is sufficient for the response or retains the shift. The same one-line argument gives bias bounds for estimated sketched ratios and for weight clipping, and a variant of the rule that keeps a predictor fitted on all covariates and sketches only the weights, whose leakage is governed by the conditional law of the \emph{score} rather than of the response (Theorem~\ref{thm:fullscore}). On each fibre of the sketch the leakage is exactly a covariance between the discarded shift and the response kernel, which yields an interaction bound through the maximal correlation of the two (Proposition~\ref{prop:cov}).
\item \textbf{The trade-off cannot be avoided (Theorem~\ref{thm:lb}).} For every leakage level $\Delta_T$ up to $\tfrac14\min\{\alpha,1-\alpha\}$ there is a shift problem on which any rule that chooses a threshold for a prescribed score from sketched data, even with unlimited calibration data and randomisation, has worst-case coverage error at least $\alpha\Delta_T/(\alpha+\Delta_T)$, and this value is attained. The restriction to threshold rules is essential, and we show why.
\item \textbf{The Gaussian single-index model, and evidence (Sections~\ref{sec:gauss} and \ref{sec:exp}).} Under a Gaussian mean shift and a linear single-index response the leakage is explicit: it is the interaction between the discarded part of the shift and the discarded part of the signal, so a sketch containing the index has zero leakage for every shift, the one-dimensional sketch along the index pays only for the component of the shift along the index, and a random sketch retains, on average, a $k/d$ fraction of the squared norm of both shift and signal. A clipped-weight version of the finite-sample calibration bound covers the lognormal ratios of this model (Theorem~\ref{thm:lognormal}). We also bound the error of the natural source-data certificate of a sketch, give an oracle inequality for selecting among candidate sketches with a few target labels, and propose a default residual-based sketch (Section~\ref{sec:choose}). Simulations in the Gaussian model and in a nonlinear multiple-index model, and two real-data studies with controlled covariate-dependent sampling, show the trade-off between weight variability, coverage and interval width, including the cases where a learned sketch fails.
\end{enumerate}

\paragraph{Organisation.} Section~\ref{sec:background} fixes notation and states the two background facts the paper rests on. Section~\ref{sec:sketch} defines sketched calibration, gives the procedure and the main theorem, draws out what compression changes, expresses the leakage as a covariance, introduces the full-score variant, and discusses the choice of the sketch. Section~\ref{sec:gauss} makes everything explicit in the Gaussian model, Section~\ref{sec:lb} proves the lower bound, Section~\ref{sec:exp} reports the experiments, and Section~\ref{sec:disc} discusses open problems. Appendices contain the proofs of the background results (\ref{app:background}), sharper calibration bounds and their behaviour under sketching (\ref{app:sharper}), the lower-bound construction (\ref{app:lb}), the Lipschitz-kernel calculations (\ref{app:lip}), the proof of the clipped-weight bound (\ref{app:lognormal}), auxiliary lemmas (\ref{app:aux}) and experimental details (\ref{app:exp}).

\section{Related work}

Weighted conformal prediction under known covariate shift is due to Tibshirani et al.\ \cite{tibshirani2019}; Lei and Cand\`es \cite{lei2021} applied it to counterfactual inference and quantified the effect of estimated weights, and Yang, Kuchibhotla and Tchetgen Tchetgen \cite{yang2024} obtained doubly robust versions. Isaev et al.\ \cite{isaev2026} derived finite-sample $L^p$ bounds on the realised conditional-coverage profile of weighted conformalised quantile regression \cite{romano2019} and proved that the $\sqrt{(1+\chi^2)/m}$ calibration rate is minimax for fixed-score threshold calibration; our Theorem~\ref{thm:price} gives a short self-contained version of both halves, and our work takes the lower bound as the motivation for reducing $\chi^2$. The companion paper \cite{pournaderi2026} gives sharp training-conditional bounds for the weighted quantile from a single concentration inequality at a fixed population quantile, with explicit constants and a variance proxy that can be much smaller than $1+\chi^2$; Appendix~\ref{app:sharper} relates our cruder bounds to those and shows that the variance proxy there does not increase under sketching. Wang and Goel \cite{wang2026} analyse weight clipping under unbounded shift and give an empirical correction for dataset-conditional coverage; Proposition~\ref{prop:clip} is a one-line clipping bound in our surrogate-target framework, included for comparison with the leakage bound. The impossibility of distribution-free conditional coverage \cite{vovk2012,barber2021} is the reason approximate and localised notions \cite{guan2023,gibbs2023} are studied.

Two recent methods also avoid estimating the full density ratio, by different routes. Likelihood-ratio regularised quantile regression (LR-QR) \cite{joshi2025} fits a quantile function by the pinball loss plus a data-dependent regulariser that involves only one-dimensional projections of the likelihood ratio, estimable from unlabelled source and target covariates, and obtains target coverage up to error terms that vanish with the sample sizes through a stability analysis; it changes the quantile fit rather than the calibration step, and its guarantee is approximate. Wasserstein-regularised conformal prediction (WR-CP) \cite{xu2025} addresses joint covariate and concept shift, bounds the coverage gap by a Wasserstein distance rather than by total variation, and combines importance weighting for the covariate part with representation learning that shrinks the concept-shift part during training, with an adaptive coverage bound in a multi-source setting. The setting here is pure covariate shift with a shared response kernel: the representation changes only the calibration space, the bias it creates is a total-variation quantity produced by the compression itself rather than a concept shift present in the data, and the guarantee is finite-sample and exact up to that explicit bias. Neither method addresses the question studied here, namely what is lost when a reduction is not sufficient and whether the loss can be avoided.
Importance weighting on a reduced representation is not new, and two predecessors are close. Stojanov et al.\ \cite{stojanov2019} prove (their Theorem 1) that if $X\perp Y\mid h(X)$ and the loss factors through $h(X)$, then weighting the source risk by the ratio of the laws of $h(X)$ gives the same target risk as weighting by the ratio of the laws of $X$; they learn a linear $h(x)=W^{\!\top}x$ by minimising the trace of a conditional covariance operator (a kernel sufficient-dimension-reduction criterion), estimate the reduced ratio by kernel mean matching, and show that the resulting generalisation bound avoids the exponential dependence on $d$ of the full-dimensional one. Polo and Vicente \cite{polo2022} formalise the effective sample size of importance weighting as $\mathrm{ESS}^\ast=\exp(-D_2)=1/(1+\chi^2)$, show that it cannot increase when features are added (their Theorem 4.1), and that linear dimensionality reduction or feature selection, even data-dependent, cannot decrease it (their Lemma 5.1 and Theorems 5.2--5.3, by data processing for R\'enyi divergence); they also note that such a reduction may discard response-relevant information without quantifying the loss. Our Theorem~\ref{prop:main}(iv)(b) is the conformal counterpart of Stojanov et al.'s Theorem 1, and Theorem~\ref{prop:main}(ii) sharpens Polo and Vicente's monotonicity to an exact Pythagorean identity through the projection structure of the sketched ratio. What neither paper addresses, and what this paper is about, is the price of a reduction that is \emph{not} sufficient: the leakage $\Delta_T$ in the finite-sample coverage guarantee, the lower bound showing it is unavoidable for fixed-score calibration (sharp to first order as $\Delta_T/\alpha\to0$), and the interaction structure between discarded shift and discarded signal that governs it. The zero-leakage condition ``$Y\perp X\mid T(X)$'' is the defining property of a sufficient dimension-reduction subspace \cite{li1991,cook1998}, which is why the estimators of \cite{stojanov2019} are natural sketches for Corollary~\ref{cor:informed}(b).

Sketching and random projection are classical in numerical linear algebra and statistics: Johnson and Lindenstrauss \cite{jl1984,dasgupta2003}, random projections for regression \cite{maillard2012}, sketched least squares \cite{raskutti2016}, and random-projection ensembles for classification \cite{cannings2017}. The combination with conformal inference is recent: Sesia and Favaro's conformal intervals for frequencies stored in a count-min sketch \cite{sesia2022}, and random-projection ensembles for conformal classification \cite{rpconformal2024}.
\section{Background: weighted conformal prediction and the price of shift}\label{sec:background}

\subsection{Notation and standing conventions}

All spaces are standard Borel, so regular conditional distributions exist and are unique up to null sets. $X$ takes values in $\mathcal X$ and $Y\in\R$. For a probability measure $\nu$ and a measurable map $T$, $T_\#\nu$ is the push-forward. For a nondecreasing right-continuous $H:\R\to[0,\infty)$ with $H(-\infty)=0$ and $u>0$, the generalised inverse is $\Qf(u;H)=\inf\{t\in\R:H(t)\ge u\}$ with $\inf\varnothing=+\infty$. $\TV(\mu,\nu)=\sup_A|\mu(A)-\nu(A)|$; for $\mu,\nu$ with densities $f,g$ with respect to a common measure, $\TV=\tfrac12\int|f-g|$ (Scheff\'e). $\KL(\mu\|\nu)=\int\log(\mathrm d\mu/\mathrm d\nu)\,\mathrm d\mu$, $\chi^2(\mu\|\nu)=\int(\mathrm d\mu/\mathrm d\nu-1)^2\mathrm d\nu$. $\Phi,\varphi$ are the standard normal CDF and density, $z_u=\Phi^{-1}(u)$.

\subsection{Setting}

Training data $\mathcal D^{\mathrm{tr}}_n$ and calibration data $\mathcal D^{\mathrm{cal}}_m=\{(X_i,Y_i)\}_{i=1}^m$ are independent i.i.d.\ samples from the \emph{source} $P=\PX\otimes P_{Y|X}$; the test pair $(X,Y)$ is independent of both and follows the \emph{target} $Q=\QX\otimes P_{Y|X}$, with the same response kernel. A score $S:\mathcal X\times\R\to\R$ is fitted on the training fold and held fixed throughout; all probabilities below are conditional on the training fold, and we write $S_i=S(X_i,Y_i)$. We use three nested assumptions on the shift.

\begin{assumption}[Known shift]\label{ass:L}
$\QX\ll\PX$ and the likelihood ratio $w=\mathrm d\QX/\mathrm d\PX$ is known.
\end{assumption}
\begin{assumption}[Finite chi-square]\label{ass:chi}
Assumption~\ref{ass:L} holds and $w\in L^2(\PX)$, so that $\chisq$, which equals $\E_{\PX}[w^2]-1$, is finite.
\end{assumption}
\begin{assumption}[Bounded ratio]\label{ass:bdd}
Assumption~\ref{ass:L} holds and $w\le w_{\max}<\infty$ $\PX$-a.s.; then $\chisq\le w_{\max}-1$.
\end{assumption}

Validity of weighted conformal prediction needs only Assumption~\ref{ass:L}; the projection identities of Section~\ref{sec:sketch} need Assumption~\ref{ass:chi}; the Bernstein-type concentration bound of Theorem~\ref{thm:price}(a) needs Assumption~\ref{ass:bdd}. The Gaussian examples of Section~\ref{sec:gauss} satisfy Assumption~\ref{ass:chi} but not \ref{ass:bdd}.

\subsection{Weighted conformal prediction}

\begin{algorithm}[!htb]
\caption{Weighted split conformal prediction \cite{tibshirani2019}}\label{alg:wcp}
\begin{algorithmic}[1]
\Require fitted score $S$; calibration pairs $(X_i,Y_i)_{i=1}^m$; weight function $w$; test covariate $x$; level $\alpha$
\State $S_i\gets S(X_i,Y_i)$ for $i=1,\dots,m$
\State $p_i(x)\gets w(X_i)\big/\bigl(\sum_{j=1}^m w(X_j)+w(x)\bigr)$ for $i=1,\dots,m$;\quad $p_{m+1}(x)\gets w(x)\big/\bigl(\sum_{j=1}^m w(X_j)+w(x)\bigr)$
\State $\widehat Q^w(x)\gets\Qf\bigl(1-\alpha;\ \sum_{i=1}^m p_i(x)\,\ind\{S_i\le\cdot\}\bigr)$\quad(the $(1-\alpha)$-quantile of $\sum_i p_i(x)\delta_{S_i}+p_{m+1}(x)\delta_{+\infty}$; it equals $+\infty$ iff $\sum_{i\le m}p_i(x)<1-\alpha$)
\State \Return $\Cset^w(x)=\{y\in\R:S(x,y)\le\widehat Q^w(x)\}$\quad(convention: if $\sum_{j\le m}w(X_j)+w(x)=0$, return $\Cset^w(x)=\R$)
\end{algorithmic}
\end{algorithm}

Algorithm~\ref{alg:wcp} is invariant to multiplying the weight function by a positive constant, since only the normalised $p_i(x)$ enter. This scale invariance is used repeatedly below.

\begin{theorem}[Validity under covariate shift; \cite{tibshirani2019}]\label{thm:wcp}
Under Assumption~\ref{ass:L}, Algorithm~\ref{alg:wcp} satisfies $\Prob_Q\{Y\in\Cset^w(X)\}\ge1-\alpha$, where the probability is over the calibration sample and the independent test pair $(X,Y)\sim Q$.
\end{theorem}

The proof (Appendix~\ref{app:wcp}) uses only that the joint law of calibration and test pairs has density $w(x_{m+1})$ with respect to $P^{\otimes(m+1)}$; no boundedness is needed. It is included because the sketched version in Section~\ref{sec:sketch} relies on exactly which properties of the joint law are used.

\subsection{The price of shift}

The following theorem is the motivation for everything that follows. Its upper bound shows that the calibration error of the weighted rule scales as $\sqrt{(1+\chisq)/m}$; its lower bound shows that no threshold rule for a fixed score can do better on every problem. The statements are the simplest ones that exhibit the rate, and they are proved in full in Appendix~\ref{app:price}; sharper versions from \cite{isaev2026,pournaderi2026} are discussed in Appendix~\ref{app:sharper}. For $\delta\in(0,1)$ define
\begin{equation}\label{eq:eta}
\eta_m(\delta):=\sqrt{\frac{2(1+\chisq)\log(16\sqrt m/\delta)}{m}}+\frac{2w_{\max}\log(16\sqrt m/\delta)}{3m}+\frac1{\sqrt m}.
\end{equation}

\begin{theorem}[The $\chi^2$ price of shift]\label{thm:price}
\begin{enumerate}[leftmargin=1.8em,itemsep=2pt]
\item[(a)] (Upper bound.) Under Assumption~\ref{ass:bdd}, suppose the score has a continuous distribution under $Q$. Then with probability at least $1-\delta$ over the calibration sample, the coverage of Algorithm~\ref{alg:wcp} conditional on that sample satisfies
\[
\abs{\Prob_Q\{Y\in\Cset^w(X)\mid\mathcal D^{\mathrm{cal}}_m\}-(1-\alpha)}\ \le\ 2\eta_m(\delta)+\frac{w_{\max}}{m}.
\]
\item[(b)] (Lower bound.) Fix $\alpha\in(0,1)$ and $\kappa>0$. There exist $\PX,\QX$ on a two-point space with $\chisq=\kappa$, a score $S^\star$, and a family of response kernels satisfying Assumption~\ref{ass:A1} below, such that for every $m\ge m_\star:=2(1+\kappa)\max\{\alpha,1-\alpha\}/\min\{\alpha,1-\alpha\}$ and every rule $\widehat q=\widehat q(\mathcal D^{\mathrm{cal}}_m,U)$ mapping $m$ source calibration pairs and an independent randomisation variable $U$ to a threshold in $\Rbar$,
\begin{gather*}
\sup_{P_{Y|X}}\ \E_{\mathcal D^{\mathrm{cal}}_m,U}\,\abs{\Cov(\widehat q)-(1-\alpha)}\ \ge\ \frac18\sqrt{\frac{(1+\kappa)\,\alpha(1-\alpha)}{2m}},\\
\text{where }\Cov(q):=\Prob_Q\{S^\star(X,Y)\le q\}.
\end{gather*}
\end{enumerate}
\end{theorem}

The loss in (b) is the error of the \emph{realised} threshold: the probability defining $\Cov(q)$ is over the test pair only. Averaging over the randomisation inside the absolute value would make the loss trivial (a rule that outputs $-\infty$ with probability $\alpha$ and $+\infty$ otherwise has average coverage exactly $1-\alpha$), so the conditioning on $U$ matters.

The lower bound uses the following regularity condition on response kernels, which we also need in Section~\ref{sec:lb}. It is the local-mass condition of \cite{isaev2026}: without an upper density bound the coverage of any set is uncontrolled, and without lower growth near the quantiles the threshold is not identifiable.

\begin{assumption}[Local mass]\label{ass:A1}
There are $0<\mu_{\mathrm{low}}\le\mu_{\mathrm{up}}<\infty$ and $r_0>0$ such that for $\PX$-a.e.\ $x$: $P_{Y|X=x}([u,v])\le\mu_{\mathrm{up}}(v-u)$ for all finite $u\le v$, and for $\tau\in\{\alpha/2,1-\alpha/2\}$ and $0\le h\le r_0$, with $f_\tau(x)$ the conditional $\tau$-quantile,
$P_{Y|X=x}([f_\tau(x)-h,f_\tau(x)])\ge\mu_{\mathrm{low}}h$ and $P_{Y|X=x}([f_\tau(x),f_\tau(x)+h])\ge\mu_{\mathrm{low}}h$.
\end{assumption}

Theorem~\ref{thm:price}(b) is a statement about the calibration step only: the score is held fixed and the adversary varies the response kernel within a regular class. It says that no threshold rule can beat the scale $\sqrt{(1+\chisq)/m}$ on every problem, so that an effective calibration size of $m\cdot\mathrm{ESS}^\ast=m/(1+\chisq)$ is a property of the problem, not of the method.

Sharper versions of both halves exist. In particular, \cite{pournaderi2026} shows that for a given score the variance of the realised coverage is governed not by $1+\chisq$ but by a proxy $\sigma_\alpha^2$ that depends on how the ratio couples to the score tail, and that the effective calibration size can be as small as $m/\mathrm{CVaR}_\alpha(w)$ rather than $m/(1+\chisq)$; Appendix~\ref{app:sharper} records the statements. Neither refinement changes the conclusion of this paper, because for $\alpha\le\tfrac12$ the leading variance proxy does not increase under sketching (Proposition~\ref{prop:proxy}). The only remaining lever is to change what is calibrated, and Polo and Vicente's monotonicity result \cite{polo2022} says that a linear reduction can only raise the effective sample size; the question is what it costs in coverage.

\section{Sketched calibration}\label{sec:sketch}

\subsection{Definition}

Let $T:\mathcal X\to\Zc$ be a measurable map into a standard Borel space $\Zc$, chosen independently of the calibration data (it may depend on the training fold or on auxiliary unlabeled data), and write $Z=T(X)$.

\begin{definition}[Sketched calibration]\label{def:sketch}
The sketched calibration rule is Algorithm~\ref{alg:wcp} applied to the pairs $(Z_i,Y_i)=(T(X_i),Y_i)$, a score $S_Z:\Zc\times\R\to\R$ fitted on the sketched training fold, and the \emph{sketched likelihood ratio}
\[
w_Z:=\frac{\mathrm dQ_Z}{\mathrm dP_Z},\qquad P_Z:=T_\#\PX,\quad Q_Z:=T_\#\QX ,
\]
evaluated at the sketched test point $T(x)$. Its output is $\Cset(T(x))=\{y:S_Z(T(x),y)\le\widehat Q^{w_Z}(T(x))\}$.
\end{definition}

$Q_Z\ll P_Z$ follows from $\QX\ll\PX$, so $w_Z$ exists. It is a $\dim\Zc$-dimensional density ratio; when $\Zc=\R^k$ with $k\ll d$ it is correspondingly easier to estimate (Proposition~\ref{prop:est} and experiment E2 quantify this).
\paragraph{Procedure.} In practice sketched calibration is four steps.
\begin{enumerate}[leftmargin=1.8em,itemsep=1pt]
\item Using the training fold (and unlabeled covariates, if available), choose the sketch $T$ and fit the score $S_Z$ on the sketched training pairs $(T(X_i),Y_i)$. Any measurable map is allowed, provided it is fixed before the calibration data are seen: a coordinate selection (variable selection), a learned representation such as a sufficient-dimension-reduction direction \cite{li1991,cook1998,stojanov2019} or the penultimate layer of a fitted network, a quantiser, or a random projection \cite{jl1984}.
\item Obtain the ratio $w_Z$ in the sketched space, in closed form when the sketched laws are known or by estimation from unlabeled source and target covariates; this is a $\dim\Zc$-dimensional density-ratio problem.
\item Compute the held-out calibration scores $S_Z(T(X_i),Y_i)$, $i=1,\dots,m$.
\item For a test covariate $x$, run Algorithm~\ref{alg:wcp} on the sketched scores with weights $w_Z(T(X_i))$ and test weight $w_Z(T(x))$, including the test-point mass at $+\infty$.
\end{enumerate}
Random projection is one choice of $T$, not the recommended default: Section~\ref{sec:gauss} shows that it retains, on average, only a $k/d$ fraction of the squared norm of any fixed direction, and the sketches that work in Section~\ref{sec:exp} are learned from the training fold. The requirement in step~1 is the only restriction on $T$.

\begin{remark}[Which sketches are allowed]\label{rem:allowed}
Nothing in what follows uses linearity. The one requirement is that the data used to choose $T$ be independent of the calibration fold and of the test pair, so that, conditional on $T$, the $(Z_i,Y_i)$ remain i.i.d.\ and the test pair retains its law; all statements are then conditional on $T$. Estimating $T$ from the calibration covariates, or from an unlabeled target sample that contains the test covariate, would break this and is excluded.
\end{remark}
\subsection{Why a covariate shift in $X$ is not a covariate shift in $Z$}

Push the full laws forward along $T$. Disintegrating $\PX$ and $\QX$ with respect to $Z$ gives fibre laws $P_{X|Z=z}$ and $Q_{X|Z=z}$ on $\{x:T(x)=z\}$, and
\begin{gather}
P^T:=T_\#P=P_Z\otimes P_{Y|Z},\qquad Q^T:=T_\#Q=Q_Z\otimes Q_{Y|Z},\notag\\
P_{Y|Z=z}=\int P_{Y|X=x}\,\mathrm dP_{X|Z=z}(x),\qquad
Q_{Y|Z=z}=\int P_{Y|X=x}\,\mathrm dQ_{X|Z=z}(x).\label{eq:fibre}
\end{gather}
The response kernel $P_{Y|X}$ is shared by source and target, but the fibre laws are not, so in general $Q_{Y|Z}\ne P_{Y|Z}$.
Example~\ref{ex:toy} in the introduction is the extreme case, in which the sketch is constant and the whole covariate shift becomes a shift of $Y\mid Z$.

The sketched rule is designed for the \emph{surrogate} target
\[
\Qt:=Q_Z\otimes P_{Y|Z},
\]
which is a genuine covariate shift in $Z$ with ratio $w_Z$. Its discrepancy from the true target $Q^T$ is the quantity that governs everything.

\begin{definition}[Leakage]
The pointwise and average leakage of the sketch $T$ are
\[
\delta_T(z):=\TV\bigl(Q_{Y|Z=z},P_{Y|Z=z}\bigr),\qquad
\Delta_T:=\E_{Q_Z}[\delta_T(Z)]=\TV\bigl(Q^T,\Qt\bigr),
\]
the last equality holding because $Q^T$ and $\Qt$ share the $Z$-marginal (Lemma~\ref{lem:tvmix}).
\end{definition}
\subsection{Main result}

\begin{theorem}[Sketched calibration is valid up to leakage, at no larger calibration cost]\label{prop:main}
Let $\Cset$ be the sketched calibration set. Under Assumption~\ref{ass:L}, for the true target $(X,Y)\sim Q$,
\begin{equation}\label{eq:val}
\Prob_Q\bigl\{Y\in\Cset(T(X))\bigr\}\ \ge\ 1-\alpha-\Delta_T .
\end{equation}
Moreover:
\begin{enumerate}[leftmargin=1.8em,itemsep=1pt]
\item[(i)] (Projection) Under Assumption~\ref{ass:L}, $w_Z\circ T=\E_P[\,w\mid T(X)\,]$ $\PX$-a.s.\ and $\E_P[w_Z\circ T]=1$. Under Assumption~\ref{ass:chi}, $w_Z\circ T$ is the $L^2(\PX)$-projection of $w$ onto $\sigma(T)$-measurable functions.
\item[(ii)] (Pythagoras) Under Assumption~\ref{ass:chi}, $\chisq=\chisqZ+\E_{\PX}\bigl[(w-w_Z\circ T)^2\bigr]$; in particular $\chisqZ\le\chisq$. Under Assumption~\ref{ass:bdd}, $w_Z\le w_{\max}$.
\item[(iii)] (Leakage bounds) $\delta_T(z)\le\TV(Q_{X|Z=z},P_{X|Z=z})$ for $Q_Z$-a.e.\ $z$, and
\[
\Delta_T\ \le\ \E_{Q_Z}\sqrt{\tfrac12\KL(Q_{X|Z}\|P_{X|Z})}\ \le\ \sqrt{\tfrac12\bigl(\KL(\QX\|\PX)-\KL(Q_Z\|P_Z)\bigr)} ,
\]
where the second inequality is asserted when $\KL(\QX\|\PX)<\infty$, which holds under Assumption~\ref{ass:chi} since $\KL\le\chi^2$.
\item[(iv)] (Zero leakage) $\Delta_T=0$ if either (a) $w$ is $\sigma(T)$-measurable, i.e.\ the shift lives in the sketch, or (b) $P_{Y|X=x}$ depends on $x$ only through $T(x)$, i.e.\ the sketch is sufficient for the response.
\end{enumerate}
\end{theorem}

\begin{proof}
\emph{Validity.} The pairs $(Z_i,Y_i)$ are i.i.d.\ from $P^T=P_Z\otimes P_{Y|Z}$, and $\Qt=Q_Z\otimes P_{Y|Z}$ differs from $P^T$ only in the $Z$-marginal, with ratio
\[
\frac{\mathrm d\Qt}{\mathrm dP^T}(z,y)=w_Z(z).
\]
This is exactly the setting of Theorem~\ref{thm:wcp} with $(\mathcal X,P,Q,w)$ replaced by $(\Zc,P^T,\Qt,w_Z)$, so
\[
\Prob_{\Qt}\bigl\{Y\in\Cset(Z)\bigr\}\ \ge\ 1-\alpha,
\]
the probability being over the calibration data and $(Z,Y)\sim\Qt$ independent of them. Conditional on the calibration data, $E:=\{(z,y):y\in\Cset(z)\}$ is a fixed measurable subset of $\Zc\times\R$, so
\[
|Q^T(E)-\Qt(E)|\le\TV(Q^T,\Qt)=\Delta_T.
\]
Since $(T(X),Y)\sim Q^T$ when $(X,Y)\sim Q$, integrating over the calibration data gives \eqref{eq:val}.

(i) For bounded measurable $g$ on $\Zc$, the tower property (valid since $w\ge0$ is $P$-integrable) gives
\[
\E_Q[g(Z)]=\E_P[w(X)\,g(T(X))]=\E_P\bigl[\E_P[w\mid T(X)]\,g(T(X))\bigr];
\]
hence $\mathrm dQ_Z/\mathrm dP_Z=\E_P[w\mid T(X)=\cdot]$, and the mean is preserved. If $w\in L^2(P)$, conditional expectation given $\sigma(T)$ is the orthogonal projection of $L^2(P)$ onto the closed subspace of $\sigma(T)$-measurable square-integrable functions.

(ii) By definition,
\[
\chisq=\E_P[(w-1)^2],\qquad \chisqZ=\E_{P_Z}[(w_Z-1)^2]=\E_P[(w_Z\circ T-1)^2].
\]
Write
\[
w-1=(w_Z\circ T-1)+(w-w_Z\circ T);
\]
the first summand is $\sigma(T)$-measurable and the second is orthogonal to every $\sigma(T)$-measurable function in $L^2(P)$ by (i), so the cross term vanishes. The sup bound is immediate from (i).

(iii) By \eqref{eq:fibre}, $Q_{Y|Z=z}$ and $P_{Y|Z=z}$ are the images of $Q_{X|Z=z}$ and $P_{X|Z=z}$ under the same Markov kernel $P_{Y|X}$; Lemma~\ref{lem:dp} gives the first claim. Pinsker's inequality $\TV\le\sqrt{\KL/2}$ applied on each fibre and integrated against $Q_Z$ gives the first inequality of the display; Lemma~\ref{lem:chain} (chain rule for the deterministic map $T$) gives
\[
\E_{Q_Z}\KL(Q_{X|Z}\|P_{X|Z})=\KL(\QX\|\PX)-\KL(Q_Z\|P_Z)
\]
when $\KL(\QX\|\PX)<\infty$, and Jensen's inequality for the concave map $u\mapsto\sqrt u$ gives the second.

(iv) (a) If $w=g\circ T$ then $w_Z=g$ by (i), and for bounded measurable $h$ on $\mathcal X$ and $B\subseteq\Zc$,
\begin{align*}
\E_Q[h(X)\ind_B(Z)]&=\E_P[g(Z)h(X)\ind_B(Z)]=\int_B g(z)\,\E_P[h(X)\mid Z=z]\,\mathrm dP_Z(z)\\
&=\int_B\E_P[h(X)\mid Z=z]\,\mathrm dQ_Z(z),
\end{align*}
so
\[
\E_Q[h(X)\mid Z]=\E_P[h(X)\mid Z]\qquad Q_Z\text{-a.s.},
\]
i.e.\ $Q_{X|Z=z}=P_{X|Z=z}$ for $Q_Z$-a.e.\ $z$, and the two fibre integrals in \eqref{eq:fibre} coincide. (b) If $P_{Y|X=x}=K(T(x),\cdot)$ for a kernel $K$ then both fibre integrals equal $K(z,\cdot)$.
\end{proof}
\subsection{What compression changes}\label{sec:three}

Theorem~\ref{prop:main} separates three effects of replacing $X$ by $T(X)$, and the experiments of Section~\ref{sec:exp} are designed to show each one.

\emph{Weight variability.} By part (ii) the sketched divergence $\chisqZ$, and under bounded ratios the supremum of the ratio, are no larger than the originals. These are the quantities that enter the calibration-error bounds (Theorem~\ref{thm:price}, Corollary~\ref{cor:profile}), so compression does not increase the shift-dependent calibration terms in these bounds. The realised variability and the frequency with which the threshold is $+\infty$ are problem-dependent and need not decrease: with $\alpha=0.1$, $m=10$, $\PX=(0.3,0.1,0.6)$ and $w=(1,5,\tfrac13)$, merging the first two atoms gives $w_Z=(2,\tfrac13)$ and lowers $\chi^2$ from $\tfrac{28}{15}$ to $\tfrac23$, yet raises the probability of an infinite threshold from $62.1\%$ to $79.9\%$ (an exact computation; with $Y$ independent of $X$ the leakage is zero). In the experiments the sketched rules do return fewer infinite sets, but that is an observation, not a consequence of the theorem.

\emph{Coverage loss.} The discarded shift may reappear as a shift of $Y\mid Z$. This is the leakage; it is zero in the two cases of part (iv), it is the only term by which the sketched guarantee \eqref{eq:val} falls short of $1-\alpha$, and Section~\ref{sec:lb} shows it cannot be avoided.

\emph{Interval width.} The score is fitted on $T(X)$, so if the sketch discards predictive information the oracle set in the sketched space can be wider than in the full space, even when the leakage is zero and $\chisqZ$ is small. In the Gaussian model of Section~\ref{sec:gauss} the oracle length grows from $2z_{1-\alpha/2}\sigma$ to $2z_{1-\alpha/2}s$ with $s^2=\sigma^2+\norm{\beta_\perp}^2$. This third effect is not visible in \eqref{eq:val}, which is a coverage statement, but it is the reason a sketch must be chosen with the response in mind: a random sketch (experiment E3) makes the first effect visible and the third severe, reducing the coverage SD by about $40\%$ while tripling the median length, whereas an informed sketch that retains the signal avoids both the second and the third effect and pays only for the shift along the retained directions.

The three effects do not move together under refinement of the sketch: the weight dispersion $\chisqZ$ is monotone (it can only grow when directions are added), whereas the leakage need not be, as the next remark shows; this is why the choice of sketch is not a one-dimensional tuning problem.
\begin{remark}[The family of sketches, and non-monotonicity of leakage]\label{rem:family}
Varying $T$ interpolates between unweighted conformal ($T$ constant: $w_Z\equiv1$, $\chisqZ=0$, leakage $\TV(Q_Y,P_Y)$) and exact weighting ($T$ injective: zero leakage, full $\chi^2$). Parts (i)--(ii) of Theorem~\ref{prop:main} say what is traded: the sketch keeps the component of $w-1$ in a chosen subspace of $L^2(\PX)$ and discards the rest, and $\chisqZ$ is monotone under refinement of the sketch, because projecting onto a larger subspace cannot decrease the norm (this is the effective-sample-size monotonicity of \cite[Thm.\ 5.2]{polo2022} in the $\chi^2$ scale). The leakage is \emph{not} monotone under refinement. In the Gaussian model of Section~\ref{sec:gauss} with $d=2$, $a=(1,1)/\sqrt2$ and $\beta=(1,-1)/\sqrt2$, the constant sketch has zero leakage (because $\langle a,\beta\rangle=0$), the coordinate sketch $T(x)=x_1$ has leakage $2\Phi\bigl(\mu/(4\sqrt{\sigma^2+1/2})\bigr)-1>0$, and the identity sketch has zero leakage again; see \eqref{eq:gleak}. Refining a sketch can create an interaction between the discarded shift and the discarded signal that a coarser sketch did not have.
\end{remark}

\subsection{Leakage is a covariance}\label{sec:cov}

The bound of Theorem~\ref{prop:main}(iii) charges the leakage to the whole discarded shift. The following identity shows that only the part of the discarded shift that is correlated with the response kernel on the fibre matters, and turns the slogan ``leakage is an interaction'' into a general inequality.

\begin{proposition}[Leakage as a fibre covariance]\label{prop:cov}
Under Assumption~\ref{ass:L}, for $Q_Z$-a.e.\ $z$,
\begin{equation}\label{eq:cov}
\delta_T(z)=\sup_{A}\ \Cov_{P_{X|Z=z}}\Bigl(\frac{w(X)}{w_Z(z)},\ P_{Y|X=X}(A)\Bigr),
\end{equation}
the supremum over measurable $A\subseteq\R$, and the same holds for $\delta_{T,S}(z)$ with $P_{Y|X=x}(A)$ replaced by $P_{Y|X=x}\{S(x,Y)\in A\}$. Consequently, if $\chi^2(Q_{X|Z=z}\|P_{X|Z=z})<\infty$ and on the fibre over $z$ the ratio depends on $x$ only through a statistic $U=U_z(x)$ and the response kernel only through a statistic $V=V_z(x)$, then
\begin{equation}\label{eq:maxcorr}
\delta_T(z)\ \le\ \tfrac12\,\rho^\ast_z(U,V)\,\sqrt{\chi^2\bigl(Q_{X|Z=z}\,\|\,P_{X|Z=z}\bigr)},
\end{equation}
where $\rho^\ast_z(U,V)=\sup\{\Cov(f(U),g(V))/\sqrt{\operatorname{Var}f(U)\operatorname{Var}g(V)}\}$ is the maximal correlation of $(U,V)$ under $P_{X|Z=z}$ \cite{gebelein1941,renyi1959}. In particular $\delta_T(z)=0$ whenever $U$ and $V$ are independent under $P_{X|Z=z}$ and the fibre divergence is finite, whatever its size.
\end{proposition}
\begin{proof}
By the proof of Lemma~\ref{lem:chain}, for $Q_Z$-a.e.\ $z$ the fibre ratio is
\[
\frac{\mathrm dQ_{X|Z=z}}{\mathrm dP_{X|Z=z}}=\frac{w}{w_Z(z)}\qquad\text{on the fibre},
\]
so by \eqref{eq:fibre}, for measurable $A$,
\[
Q_{Y|Z=z}(A)-P_{Y|Z=z}(A)=\E_{P_{X|Z=z}}\Bigl[\Bigl(\frac{w(X)}{w_Z(z)}-1\Bigr)P_{Y|X=X}(A)\Bigr],
\qquad
\E_{P_{X|Z=z}}\Bigl[\frac{w}{w_Z(z)}\Bigr]=1,
\]
which gives \eqref{eq:cov}; the TV supremum over $A$ may be taken one-sided because complements change the sign. For \eqref{eq:maxcorr}, write $w/w_Z(z)=\phi(U)$ and $P_{Y|X=\cdot}(A)=\psi_A(V)$; by definition of the maximal correlation,
\[
\Cov\bigl(\phi(U),\psi_A(V)\bigr)\le\rho^\ast_z\,\mathrm{sd}(\phi(U))\,\mathrm{sd}(\psi_A(V)).
\]
Here
\[
\mathrm{sd}(\phi(U))^2=\E_{P_{X|Z=z}}\Bigl[\Bigl(\frac{w}{w_Z(z)}\Bigr)^2\Bigr]-1=\chi^2(Q_{X|Z=z}\|P_{X|Z=z}),
\]
and $\psi_A\in[0,1]$ gives $\operatorname{Var}\psi_A(V)\le\tfrac14$. The score version is identical with the kernel $x\mapsto P_{Y|X=x}\{S(x,Y)\in A\}$.
\end{proof}

Inequality \eqref{eq:maxcorr} isolates two factors: the size of the discarded shift on the fibre, measured by its chi-square divergence, and the \emph{dependence} between the discarded-shift statistic and the response-relevant statistic. The second factor is what the Kullback--Leibler bound of Theorem~\ref{prop:main}(iii) ignores, and it is zero in the situation that motivated the paper, a shift in directions that carry no information about the response. For jointly Gaussian $(U,V)$ the maximal correlation is the ordinary correlation in absolute value \cite{lancaster1957}, which is how \eqref{eq:maxcorr} specialises in Section~\ref{sec:gauss}.

\subsection{Sketched weights with a full-covariate score}\label{sec:fullscore}

The third effect of Section~\ref{sec:three} is avoidable. Nothing in the surrogate-target argument requires the score to be computed from the sketch; it requires only that the weights be. This suggests keeping a predictor fitted on all covariates and compressing the covariate only where the compression pays, in the weights.

\begin{definition}[Sketched weights, full score]\label{def:fullscore}
Let $S:\mathcal X\times\R\to\R$ be a score fitted on the training fold, $S_i=S(X_i,Y_i)$, and $Z_i=T(X_i)$. The rule with sketched weights and full score is
\begin{gather*}
\Cset_T(x):=\bigl\{y\in\R:\ S(x,y)\le\widehat Q_T(T(x))\bigr\},\\
\widehat Q_T(z):=\Qf\Bigl(1-\alpha;\ \sum_{i=1}^m\frac{w_Z(Z_i)}{\sum_jw_Z(Z_j)+w_Z(z)}\,\ind\{S_i\le\cdot\}\Bigr),
\end{gather*}
that is, Algorithm~\ref{alg:wcp} applied to the pairs $(Z_i,S_i)$ with the identity score $(z,s)\mapsto s$ and the weight function $w_Z$.
\end{definition}

Write $P_{S|Z=z}$ and $Q_{S|Z=z}$ for the conditional laws of $S(X,Y)$ given $T(X)=z$ under $P$ and $Q$ (both built from the common kernel $P_{Y|X}$ and the respective fibre laws), and define the \emph{score leakage}
\[
\delta_{T,S}(z):=\TV\bigl(Q_{S|Z=z},P_{S|Z=z}\bigr),\qquad
\Delta_{T,S}:=\E_{Q_Z}[\delta_{T,S}(Z)]=\TV\bigl(Q^{T,S},\,Q_Z\otimes P_{S|Z}\bigr),
\]
where $Q^{T,S}$ is the law of $(T(X),S(X,Y))$ under $Q$ and the last equality is Lemma~\ref{lem:tvmix}.

\begin{theorem}[Validity of the full-score rule]\label{thm:fullscore}
Under Assumption~\ref{ass:L}, $\Prob_Q\{Y\in\Cset_T(X)\}\ge1-\alpha-\Delta_{T,S}$. Moreover:
\begin{enumerate}[leftmargin=1.8em,itemsep=1pt]
\item[(i)] $\delta_{T,S}(z)\le\TV(Q_{X|Z=z},P_{X|Z=z})$ for $Q_Z$-a.e.\ $z$, so $\Delta_{T,S}$ obeys the bounds of Theorem~\ref{prop:main}(iii).
\item[(ii)] $\Delta_{T,S}=0$ if $w$ is $\sigma(T)$-measurable, or if the law of $S(x,Y)$ under $P_{Y|X=x}$ depends on $x$ only through $T(x)$ (the sketch is \emph{sufficient for the score}).
\item[(iii)] If $S(x,y)=s(T(x),y)$ for some $s$, then $\Delta_{T,S}\le\Delta_T$ and $\Cset_T$ is the sketched calibration rule of Definition~\ref{def:sketch}. In general $\Delta_{T,S}$ and $\Delta_T$ are not ordered.
\end{enumerate}
\end{theorem}
\begin{proof}
The pairs $(Z_i,S_i)$ are i.i.d.\ from $P^{T,S}:=P_Z\otimes P_{S|Z}$, and the surrogate target $\Qt_S:=Q_Z\otimes P_{S|Z}$ has
\[
\frac{\mathrm d\Qt_S}{\mathrm dP^{T,S}}(z,s)=w_Z(z).
\]
Theorem~\ref{thm:wcp} applied in $\Zc\times\R$ with the identity score gives
\[
\Prob_{\Qt_S}\bigl\{S\le\widehat Q_T(Z)\bigr\}\ge1-\alpha,
\]
the threshold depending on the test point only through $w_Z(z)$. Conditional on the calibration data, $\{(z,s):s\le\widehat Q_T(z)\}$ is a fixed measurable set, and under $Q$ the pair $(T(X),S(X,Y))$ has law $Q^{T,S}$, so the coverage differs from the surrogate coverage by at most
\[
\TV(Q^{T,S},\Qt_S)=\Delta_{T,S}.
\]

(i) $Q_{S|Z=z}$ and $P_{S|Z=z}$ are the images of $Q_{X|Z=z}$ and $P_{X|Z=z}$ under the same Markov kernel $x\mapsto\mathcal L(S(x,Y))$, $Y\sim P_{Y|X=x}$; apply Lemma~\ref{lem:dp}.

(ii) In the first case the fibre laws coincide (proof of Theorem~\ref{prop:main}(iv)(a)); in the second the kernel is constant on fibres.

(iii) If $S=s(T(x),y)$ then $Q_{S|Z=z}$ and $P_{S|Z=z}$ are the images of $Q_{Y|Z=z}$ and $P_{Y|Z=z}$ under $y\mapsto s(z,y)$, and Lemma~\ref{lem:dp} gives $\delta_{T,S}(z)\le\delta_T(z)$; the identification of the rules is immediate. For the last claim, take $Y$ independent of $X$ and $S(x,y)=f(x)$ with $f$ not $\sigma(T)$-measurable: then $\Delta_T=0$ while $\Delta_{T,S}>0$ as soon as the fibre laws differ on $f$; conversely a constant score has $\Delta_{T,S}=0$ in a problem with $\Delta_T>0$, such as Example~\ref{ex:toy}.
\end{proof}

The theorem moves the burden from the response to the score. For a well-fitted predictor and the absolute-residual score $|y-\hat\mu(x)|$, the conditional law of the score given $x$ is governed by the noise scale and by the remaining misfit $\mu(x)-\hat\mu(x)$, which is typically lower-dimensional and smaller than the regression function itself; the sketch must capture the directions along which the \emph{residual law} varies, not those along which the mean varies. The interval is built from the full predictor, so the width penalty of Section~\ref{sec:three} is gone, and the calibration cost is governed by $\chisqZ$ exactly as before. The lower bound of Section~\ref{sec:lb} applies verbatim: its score is a function of $y$ alone, so the pairs $(Z_i,S_i)$ have the same law under both kernels there, and $\Delta_{T,S}=\Delta_T$ in that construction. In the Gaussian model the gain is explicit.

\begin{proposition}[Full-score leakage in the Gaussian model]\label{prop:fullgauss}
In Example~\ref{ex:gauss} let $S(x,y)=|y-\theta^{\!\top}x-\theta_0|$ for fixed $(\theta,\theta_0)$, and write $\gamma:=\beta-\theta$, $\gamma_\perp:=\gamma-\Pi^{\!\top}\Pi\gamma$, $s_\theta^2:=\sigma^2+\norm{\gamma_\perp}^2$. Then
\[
\Delta_{\Pi,S}\ \le\ 2\Phi\Bigl(\frac{\mu\,|\langle\gamma_\perp,a_\perp\rangle|}{2s_\theta}\Bigr)-1\ \le\ \frac{\mu\,\norm{\gamma_\perp}\,\norm{a_\perp}}{\sqrt{2\pi}\,s_\theta}\ \le\ \frac{\mu\,\norm{\beta-\theta}\,\norm{a_\perp}}{\sqrt{2\pi}\,\sigma}.
\]
If instead $Y=\langle\beta,X\rangle+\sigma(\langle\beta,X\rangle)\xi$ with $\sigma(\cdot)\ge\sigma_{\min}>0$ and $L$-Lipschitz, then
\[
\Delta_{\Pi,S}\ \le\ \frac{\mu}{\sigma_{\min}}\Bigl(\frac{|\langle\gamma_\perp,a_\perp\rangle|}{\sqrt{2\pi}}+2\varphi(1)L\,|\langle\beta_\perp,a_\perp\rangle|\Bigr).
\]
\end{proposition}
\begin{proof}
Appendix~\ref{app:lip}.
\end{proof}

In the homoscedastic case the leakage of the full-score rule is driven by the \emph{misfit} $\beta-\theta$ rather than by the signal $\beta$: for a least-squares fit at the sample sizes of Section~\ref{sec:exp}, $\norm{\beta-\theta}\approx0.14$ against $\norm\beta=3$, so even a random sketch, which discards most of the signal, has leakage of order $0.01$ (experiment E6), where the sketched-score rule with the same sketch tripled the interval length. In the heteroscedastic case the second term survives: the sketch must retain the direction along which the noise scale varies, and the misfit term is again small. This is what the default sketch of Section~\ref{sec:choose} targets.

\subsection{Estimated sketched ratios and weight clipping}

The surrogate-target argument in the proof of Theorem~\ref{prop:main} gives bias bounds for two further modifications of the weights. Both use the scale invariance of Algorithm~\ref{alg:wcp}.

\begin{proposition}[Estimated sketched ratio]\label{prop:est}
Let $\hat w_Z\ge0$ be an estimate of $w_Z$ fitted on auxiliary data independent of the calibration fold and of the test pair, with $0<\E_{P_Z}[\hat w_Z]<\infty$, and let $\bar w_Z:=\hat w_Z/\E_{P_Z}[\hat w_Z]$. With the zero-denominator convention of Algorithm~\ref{alg:wcp}, the sketched rule run with $\hat w_Z$ in place of $w_Z$ satisfies, conditionally on the auxiliary data,
\[
\Prob_Q\bigl\{Y\in\Cset(T(X))\bigr\}\ \ge\ 1-\alpha-\Delta_T-\tfrac12\,\E_{P_Z}\abs{\bar w_Z-w_Z} ;
\]
unconditionally, the last term is replaced by its expectation over the auxiliary data.
\end{proposition}
\begin{proof}
Condition on the auxiliary data, so that $\hat w_Z$ is a fixed function. By scale invariance the rule run with $\hat w_Z$ coincides with the rule run with $\bar w_Z$ (the all-zero-denominator event is a $\bar Q$-null event for the surrogate target below, since there $\bar w_Z(z)>0$ a.s., and the convention makes the rule well defined elsewhere), which is a density ratio ($\bar w_Z\ge0$, $\E_{P_Z}\bar w_Z=1$) and hence, by Theorem~\ref{thm:wcp} applied in $\Zc$, is exactly valid for the target $\bar Q:=(\bar w_Z\,P_Z)\otimes P_{Y|Z}$. Conditional on the calibration data, the coverage event is a fixed set, and by Lemma~\ref{lem:tvmix} and Scheff\'e,
\[
\TV(Q^T,\bar Q)\le\TV(Q^T,\Qt)+\TV(\Qt,\bar Q)=\Delta_T+\TV(Q_Z,\bar w_ZP_Z)=\Delta_T+\tfrac12\E_{P_Z}|w_Z-\bar w_Z|.
\]
\end{proof}

The normalisation matters: the bound with $\hat w_Z$ in place of $\bar w_Z$ is false in general, because the algorithm cannot see the scale of $\hat w_Z$ but the unnormalised $L^1$ distance does. This is the bound of \cite[Thm.\ 3]{lei2021} in the sketched space; the estimation problem it refers to is $\dim\Zc$-dimensional.
\begin{remark}[Score-specific leakage, and sharper calibration bounds]\label{rem:pournaderi}
Two refinements are proved in Appendix~\ref{app:sharper}. First, for the common-threshold version of the sketched rule only the part of the discarded shift $w-\E_P[w\mid Z]$ that is correlated with the sketched score matters: the relevant leakage is
\[
\Delta_S:=\sup_t\abs{Q^T\{S_Z\le t\}-\Qt\{S_Z\le t\}}\ \le\ \tfrac12\,\E_P\abs{\E_P[w-w_Z\circ T\mid S_Z]}\ \le\ \Delta_T
\]
(Proposition~\ref{prop:scoreleak}; this is \cite[eq.\ (15)]{pournaderi2026} specialised to the sketched ratio), which is the general form of the interaction phenomenon of Section~\ref{sec:gauss}. Second, under the assumptions of \cite{pournaderi2026} (bounded sketched ratio, $\alpha\le\tfrac12$, $m>\tfrac23\log(1/\delta)$) the training-conditional bound of \cite{pournaderi2026} applies to the surrogate problem and gives target coverage at least $1-\alpha-\gamma^Z_m(\delta)-\Delta_S$ with probability $1-\delta$, where $\gamma^Z_m(\delta)$ is the calibration term of that bound evaluated at the sketched quantities; and the leading variance proxy $\sigma^2_{\alpha,Z}=\alpha^2K_Z^2+(1-2\alpha)\Lambda_Z(\alpha)$ is no larger than the unsketched $\sigma^2_\alpha$ (Proposition~\ref{prop:proxy}). The comparison is of the proxy, not of every term in that bound: the ratio quantile $\rho_\alpha$ entering its lower-order term can increase under sketching.
\end{remark}

\begin{proposition}[Weight clipping]\label{prop:clip}
Under Assumption~\ref{ass:L}, for $M>0$ let $v_M:=w\wedge M$, $c_M:=\E_{\PX}[v_M]=1-\E_{\PX}[(w-M)_+]$, and $r_M:=v_M/c_M$. Algorithm~\ref{alg:wcp} run with the clipped weights $v_M$ satisfies
\[
\Prob_Q\{Y\in\Cset^{v_M}(X)\}\ \ge\ 1-\alpha-\tfrac12\E_{\PX}\abs{w-r_M}\ \ge\ 1-\alpha-\E_{\PX}\bigl[(w-M)_+\bigr].
\]
\end{proposition}
\begin{proof}
By scale invariance the rule with $v_M$ coincides with the rule with $r_M$, a density ratio, which is exactly valid for $(r_M\PX)\otimes P_{Y|X}$ by Theorem~\ref{thm:wcp}; the first inequality is the total-variation step. For the second, with $e:=\E[(w-M)_+]$ and $c_M=1-e$,
\[
\tfrac12\E|w-r_M|\le\tfrac12\E|w-v_M|+\tfrac12\E|v_M-r_M|=\tfrac12e+\tfrac12\E[v_M]\Bigl(\frac1{c_M}-1\Bigr)=\tfrac12e+\tfrac12(1-c_M)=e.
\]
\end{proof}

Clipping and sketching are therefore both ``surrogate-target'' methods, and their biases are comparable objects. The clipping bound $\E[(w-M)_+]$ depends on $w$ alone and is strictly positive whenever $\PX\{w>M\}>0$ (it is only an upper bound: if $w=2\cdot\ind_A$ with $\PX(A)=\tfrac12$ and $M=1$, then $r_M=w$ and clipping changes nothing, although the bound is $\tfrac12$); the sketching bias $\Delta_T$ depends on how the discarded part of $w$ interacts with the response kernel, and Theorem~\ref{prop:main}(iv) and Section~\ref{sec:gauss} identify situations where it is zero while $\chisqZ\ll\chisq$. Wang and Goel \cite{wang2026} give a finer analysis of clipping, including a data-driven correction of the target level; the comparison here is only meant to place the two biases side by side.

\subsection{Calibration bounds}

\begin{corollary}[Calibration error of sketched calibration]\label{cor:profile}
Under Assumption~\ref{ass:bdd}, if the sketched score has a continuous distribution under $\Qt$, then with probability at least $1-\delta$ over the calibration sample,
\[
\abs{\Prob_Q\{Y\in\Cset(T(X))\mid\mathcal D^{\mathrm{cal}}_m\}-(1-\alpha)}\ \le\ 2\eta^T_m(\delta)+\frac{\norm{w_Z}_\infty}{m}+\Delta_T,
\]
where $\eta^T_m(\delta)$ is \eqref{eq:eta} with $(\chisq,w_{\max})$ replaced by $(\chisqZ,\norm{w_Z}_\infty)$, both of which are no larger than the originals by Theorem~\ref{prop:main}(ii). If moreover the sketched kernel $P_{Y|Z}$ satisfies Assumption~\ref{ass:A1}, the $Z$-space quantile learner satisfies the accuracy condition A2($p$) of \cite{isaev2026}, and the sample-size and localisation conditions \cite[eq.\ (22)]{isaev2026} hold for the surrogate problem (with $w_Z$, $\chisqZ$ and $\norm{w_Z}_\infty$ in place of $w$, $\chisq$ and $w_{\max}$), then on the event of \cite[Thm.\ 5]{isaev2026} applied to the surrogate problem, the realised conditional-coverage profile of sketched weighted CQR satisfies
\[
\norm{\Cov^{Q}_{\Cset}(\cdot)-(1-\alpha)}_{\Lp{Q_Z}}
\ \le\ B_p\bigl(\epsilon^T_p,\ \chisqZ,\ \norm{w_Z}_\infty,\ m,\delta\bigr)\ +\ \norm{\delta_T}_{\Lp{Q_Z}},
\]
where $B_p$ is the right-hand side of \cite[eq.\ (24)]{isaev2026} and $\Cov^Q_{\Cset}(z)=Q_{Y|Z=z}(\Cset(z))$.
\end{corollary}

\begin{proof}
The first display is Theorem~\ref{thm:price}(a) applied to the surrogate problem $(\Zc,P^T,\Qt,w_Z)$, which gives
\[
\bigl|\Prob_{\Qt}\{Y\in\Cset(Z)\mid\mathcal D^{\mathrm{cal}}_m\}-(1-\alpha)\bigr|\le2\eta^T_m(\delta)+\norm{w_Z}_\infty/m,
\]
followed by the total-variation step of the proof of Theorem~\ref{prop:main}, which holds conditionally on the calibration data. For the second, $Q_{Y|Z=z}(\Cset(z))$ differs from $P_{Y|Z=z}(\Cset(z))$ by at most $\delta_T(z)$ pointwise, the latter is the surrogate profile bounded by \cite[Thm.\ 5]{isaev2026}, and the triangle inequality in $\Lp{Q_Z}$ finishes.
\end{proof}

By Theorem~\ref{prop:main}(ii) every shift-dependent quantity in these bounds is no larger than in the unsketched problem, so the sketch does not increase the calibration term and adds the leakage. Two caveats. Assumption~\ref{ass:A1} is not inherited automatically by $P_{Y|Z}$, which is a mixture over fibres; a mixture of laws with linear mass growth at their own quantiles can have a flat spot at the mixture quantile (two uniforms on disjoint intervals), although the density upper bound does survive mixing. And the first display requires bounded ratios; for unbounded $w$ with finite $\chi^2$, as in Section~\ref{sec:gauss}, \eqref{eq:val} and the projection identities hold but the Bernstein bound does not, and the role of $\chisqZ$ as the calibration variance is then a heuristic that the experiments test rather than a theorem.

\subsection{Choosing the sketch}\label{sec:choose}

The results so far say what a good sketch is: small $\chisqZ$ and small leakage. They do not say how to find one, and the obstacle is that the leakage involves the target response law, which is never observed in the covariate-shift setting. We record what can and cannot be certified, give an oracle inequality for selection from a finite family when a few target labels are available, and describe the default procedure used in the experiments.

\paragraph{Certifying leakage from source data.}
For the common-threshold version of either rule, the relevant leakage is the score-specific quantity $\Delta_S=\sup_t|E_P[(w-w_Z\circ T)\ind\{S\le t\}]|$ of Proposition~\ref{prop:scoreleak}, a functional of the source law and of $w$ alone. With $w$ known it is therefore estimable from labelled \emph{source} data, by the weighted Kolmogorov--Smirnov statistic
\[
\widehat\Delta_S:=\sup_t\Bigl|\frac1{n_s}\sum_{l=1}^{n_s}\bigl(w(X_l')-w_Z(T(X_l'))\bigr)\ind\{S_l'\le t\}\Bigr|
\]
on an independent source sample $(X_l',Y_l')_{l\le n_s}$. Its accuracy, however, is governed by the discarded shift.

\begin{proposition}[The cost of certification]\label{prop:certify}
Under Assumption~\ref{ass:chi}, for every $t\in\Rbar$ the summand $(w-w_Z\circ T)\ind\{S\le t\}$ has mean $E_P[(w-w_Z\circ T)\ind\{S\le t\}]$ and variance at most $\E_P[(w-w_Z\circ T)^2]=\chisq-\chisqZ$, with equality at $t=+\infty$, where the mean is $0$. Under Assumption~\ref{ass:bdd}, with probability at least $1-\delta$,
$|\widehat\Delta_S-\Delta_S|\le2\,\eta_{n_s}(\delta)$, with $\eta$ as in \eqref{eq:eta} evaluated at $(\chisq,w_{\max})$ and sample size $n_s$.
\end{proposition}
\begin{proof}
The variance bound is Theorem~\ref{prop:main}(ii). For the uniform bound, write the summand as the difference of the two nonnegative weighted indicators $w\ind\{S\le t\}$ and $w_Z\circ T\,\ind\{S\le t\}$; Steps 1--2 of the proof of Theorem~\ref{thm:price}(a) (Appendix~\ref{app:price}) bound the supremum deviation of each weighted empirical CDF from its mean by $\eta_{n_s}(\delta)$, each on an event of probability at least $1-\delta/2$, with the respective $(1+\chi^2,\|\cdot\|_\infty)$, both of which are no larger than $(1+\chisq,w_{\max})$ by Theorem~\ref{prop:main}(ii); on the intersection, the supremum of the difference deviates from $\Delta_S$ by at most the sum of the two suprema.
\end{proof}

The proposition describes a limitation of this particular certificate, not an impossibility result. A sketch is worth using when it discards most of the divergence, $\chisqZ\ll\chisq$; the variance bound for the summands of $\widehat\Delta_S$ is then close to $\chisq$, and the available uniform confidence bound depends on the full divergence, so this bound can remain large even when the sketched divergence is small, and it is of the same order as the calibration-error bound of \emph{full} weighting with $n_s$ points (Theorem~\ref{thm:price}(a)). This compares bounds, not realised accuracies, and whether every source-data certificate must pay this price is open (Section~\ref{sec:disc}). The certificate considered here does not use additional unlabelled target covariates. In this paper the leakage is controlled either by structure (a sufficiency assumption, as in Corollary~\ref{cor:informed} and Proposition~\ref{prop:fullgauss}) or by target labels.

\paragraph{Selection with a few target labels.}
When a small labelled target sample is available, a finite family of candidate rules can be compared on it.

\begin{proposition}[Oracle inequality for selection]\label{prop:select}
Let $\Cset_1,\dots,\Cset_J$ be prediction rules that are fixed functions of the training and calibration data $\mathcal D$ (for instance sketched calibration with $J$ different sketches, full weighting, clipping and the unweighted rule), with realised target coverages $\Cov_j:=\Prob_Q\{Y\in\Cset_j(X)\mid\mathcal D\}$. Let $(X_l',Y_l')_{l\le n_t}$ be i.i.d.\ from $Q$, independent of $\mathcal D$, $\widehat{\Cov}_j:=n_t^{-1}\sum_l\ind\{Y_l'\in\Cset_j(X_l')\}$, and $\hat\jmath\in\arg\min_j|\widehat{\Cov}_j-(1-\alpha)|$. Then, conditionally on $\mathcal D$, with probability at least $1-\delta$,
\[
\bigl|\Cov_{\hat\jmath}-(1-\alpha)\bigr|\ \le\ \min_{j\le J}\bigl|\Cov_j-(1-\alpha)\bigr|+2\sqrt{\frac{\log(2J/\delta)}{2n_t}} .
\]
\end{proposition}
\begin{proof}
Put $\varepsilon:=\sqrt{\log(2J/\delta)/(2n_t)}$. Hoeffding's inequality and a union bound give, with probability at least $1-\delta$,
\[
\max_{j\le J}\bigl|\widehat{\Cov}_j-\Cov_j\bigr|\le\varepsilon .
\]
On this event, for every $j$,
\[
|\Cov_{\hat\jmath}-(1-\alpha)|\le|\widehat{\Cov}_{\hat\jmath}-(1-\alpha)|+\varepsilon\le|\widehat{\Cov}_j-(1-\alpha)|+\varepsilon\le|\Cov_j-(1-\alpha)|+2\varepsilon .
\]
\end{proof}

The inequality is elementary, and its honest reading is that selection costs $\sqrt{\log J/n_t}$ once; the same $n_t$ labels could instead calibrate directly on the target at error of order $n_t^{-1/2}$. The guarantee does not by itself establish that selection is the better use of the labels, since its second term is of the same order as direct calibration. Selection can be the better use when some candidate has small leakage and small $\chisqZ$ and $m$ is large, so that $\min_j|\Cov_j-(1-\alpha)|$, which is of order $\Delta_{T_j}+\sqrt{(1+\chi^2(Q_{Z_j}\|P_{Z_j}))/m}$ by Corollary~\ref{cor:profile}, is much smaller than $n_t^{-1/2}$, but this is an empirical possibility rather than a proven comparison; with tens of labels the guarantee is loose, and in the experiments (E7--E8) the procedure nevertheless stays within two points of the target in every configuration, typically picking a sketched rule when the shift is strong and the unweighted rule when it is harmless. Among candidates whose estimated coverage is within $\varepsilon$ of the target one may also prefer the shortest; its coverage error is then at most $2\varepsilon$ on the same event, whenever such a candidate exists.

\paragraph{A default: the residual sketch.}
Theorem~\ref{thm:fullscore} says that with a full-covariate predictor the sketch must capture the directions along which the conditional law of the \emph{residual} varies. This suggests the following procedure, used in Sections~\ref{sec:exp}.
\begin{enumerate}[leftmargin=1.8em,itemsep=1pt]
\item Split the training fold in two halves; fit the predictor on each half and form cross-fitted absolute residuals $R_i=|Y_i-\hat\mu_{-i}(X_i)|$ on the whole fold; refit the predictor on the whole fold for the final score.
\item Estimate the directions along which $\mathcal L(R\mid X=x)$ varies: the least-squares coefficient of $R$ on $X$ (one direction; under a Gaussian or elliptical design this is proportional to the index of a single-index scale model \cite{brillinger1983,li1991}), or the leading sliced-inverse-regression directions of $R$ on $X$ \cite{li1991} for $k\ge2$.
\item Optionally append an estimate of the shift direction from unlabelled source and target covariates (a source-versus-target classifier). The true direction would remove the shift from the fibres and so the leakage, at the price of the full $\chi^2$ along it; an estimated direction does so up to an allowance for its angular error (Corollary~\ref{cor:informed}(c)).
\item Take as candidates the sketches with $k\in\{1,2\}$, with and without the shift direction, together with the trivial sketches (the constant sketch, i.e.\ unweighted calibration; the identity, i.e.\ full weighting; and clipping), and select by Proposition~\ref{prop:select} if target labels exist, otherwise by the structural assumption that justified the candidates.
\end{enumerate}
Under the heteroscedastic single-index model of Proposition~\ref{prop:fullgauss}, step 2 estimates the scale direction $\beta$ and the leakage allowance is linear in its angular error. Two caveats are important. First, the scale direction is statistically harder to estimate than the mean direction, because the residual magnitude carries the direction only through its scale: in experiment E6 the least-squares index of the mean has angular error $\sin\angle\approx0.035$ at $n=2000$, $d=40$, while the residual direction has $\sin\angle\approx0.33$, and the residual sketch undercovers there by about one point; its accuracy improves with $n$ (E7). Second, the procedure is only as good as the structural assumption behind step 2; a response whose noise scale depends on $x$ symmetrically defeats sliced inverse regression (sliced average variance estimation \cite{cook1998} is the standard remedy), and the bound of Proposition~\ref{prop:certify} applies to the certificate considered here if the assumption is checked from source data.
\section{The Gaussian single-index model}\label{sec:gauss}

Everything above becomes explicit in the following model, which also separates the roles of the shift and the signal.

\begin{example}\label{ex:gauss}
$\PX=N(0,I_d)$, $\QX=N(\mu a,I_d)$ with $\mu\ge0$ and $\norm a=1$; $Y=\langle\beta,X\rangle+\sigma\xi$ with $\beta\ne0$, $\sigma>0$ and $\xi\sim N(0,1)$ independent of $X$; $T(x)=\Pi x$ with $\Pi\in\R^{k\times d}$ having orthonormal rows, so $\Pi\Pi^{\!\top}=I_k$ and $\Pi^{\!\top}\Pi$ is the orthogonal projector onto $\operatorname{row}(\Pi)$. Write $v_\parallel=\Pi^{\!\top}\Pi v$, $v_\perp=v-v_\parallel$ for $v\in\{a,\beta\}$.
\end{example}

When $\mu>0$, $w$ is unbounded in this model, so Assumption~\ref{ass:bdd} fails, but $\E_P[w^2]=e^{\mu^2}<\infty$, so Assumption~\ref{ass:chi} holds: Theorem~\ref{thm:wcp} and Theorem~\ref{prop:main} apply, while the Bernstein bound of Theorem~\ref{thm:price}(a) does not (and the lower bounds of Theorems~\ref{thm:price}(b) and \ref{thm:lb} are existential constructions, not statements about this family).

\begin{proposition}\label{prop:gauss}
In Example~\ref{ex:gauss}, with $s^2:=\sigma^2+\norm{\beta_\perp}^2$:
\begin{gather*}
w(x)=e^{\mu\langle a,x\rangle-\mu^2/2},\qquad \chisq=e^{\mu^2}-1,\qquad \KL(\QX\|\PX)=\mu^2/2,\\
w_Z(z)=e^{\mu\langle\Pi a,z\rangle-\mu^2\norm{a_\parallel}^2/2},\qquad \chisqZ=e^{\mu^2\norm{a_\parallel}^2}-1,\\
\KL(\QX\|\PX)-\KL(Q_Z\|P_Z)=\tfrac{\mu^2}{2}\norm{a_\perp}^2,
\end{gather*}
\[
Y\mid Z=z\ \sim\ N\bigl(\langle\Pi\beta,z\rangle,\,s^2\bigr)\ \text{under }P,\qquad
Y\mid Z=z\ \sim\ N\bigl(\langle\Pi\beta,z\rangle+\mu\langle\beta_\perp,a_\perp\rangle,\,s^2\bigr)\ \text{under }Q,
\]
so the leakage is constant in $z$ and equals
\begin{equation}\label{eq:gleak}
\delta_\Pi\equiv\Delta_\Pi=2\Phi\!\Bigl(\frac{\mu\,|\langle\beta_\perp,a_\perp\rangle|}{2s}\Bigr)-1
\ \le\ \frac{\mu\,|\langle\beta_\perp,a_\perp\rangle|}{\sqrt{2\pi}\,s}
\ \le\ \frac{\mu\norm{a_\perp}}{\sqrt{2\pi}} .
\end{equation}
The sketched oracle interval (equal-tailed quantiles of $Y\mid Z$) has length $2z_{1-\alpha/2}s$ against $2z_{1-\alpha/2}\sigma$ for the full oracle.
\end{proposition}

\begin{proof}
For $N(0,I_d)$ versus $N(\mu a,I_d)$,
\[
w(x)=\exp(\mu\langle a,x\rangle-\mu^2/2),\qquad \E_P[w^2]=\exp(\mu^2),\qquad \KL=\mu^2/2.
\]
Since $\Pi\Pi^{\!\top}=I_k$, $Q_Z=N(\mu\Pi a,I_k)$ and $P_Z=N(0,I_k)$, and $\norm{\Pi a}^2=\norm{a_\parallel}^2$; the sketched formulas follow. On the fibre $\{\Pi x=z\}$, write
\[
X=\Pi^{\!\top}z+X_\perp,\qquad X_\perp=(I-\Pi^{\!\top}\Pi)X,
\]
where $X_\perp$ is independent of $\Pi X$ under both laws (uncorrelated jointly Gaussian), with
\[
X_\perp\sim N(0,I-\Pi^{\!\top}\Pi)\ \text{ under }P,\qquad X_\perp\sim N(\mu a_\perp,I-\Pi^{\!\top}\Pi)\ \text{ under }Q,
\]
because the orthogonal projection of $\mu a$ onto $\ker\Pi$ is $\mu a_\perp$. Then, since $\langle\beta,X_\perp\rangle=\langle\beta_\perp,X_\perp\rangle$,
\[
Y=\langle\beta,\Pi^{\!\top}z\rangle+\langle\beta,X_\perp\rangle+\sigma\xi=\langle\Pi\beta,z\rangle+\langle\beta_\perp,X_\perp\rangle+\sigma\xi .
\]
Under $P$, $\langle\beta_\perp,X_\perp\rangle\sim N(0,\norm{\beta_\perp}^2)$; under $Q$ it is $N(\mu\langle\beta_\perp,a_\perp\rangle,\norm{\beta_\perp}^2)$. Lemma~\ref{lem:gausstv} gives the total variation between the two conditional laws. The first inequality in \eqref{eq:gleak} is
\[
2\Phi(t/2)-1\le t\varphi(0)=t/\sqrt{2\pi}\qquad(t\ge0),
\]
by concavity of $\Phi$ on $[0,\infty)$; the second uses Cauchy--Schwarz, $|\langle\beta_\perp,a_\perp\rangle|\le\norm{\beta_\perp}\norm{a_\perp}$, and $\norm{\beta_\perp}\le s$.
\end{proof}

Three consequences of \eqref{eq:gleak} organise the rest of this section.

\paragraph{Leakage is an interaction.} By \eqref{eq:gleak} the discarded part of the shift, $a_\perp$, costs nothing unless it is correlated with the discarded part of the signal, $\beta_\perp$. A shift orthogonal to the lost signal has $\Delta_\Pi=0$ even though $w_Z\ne w$ and $\chisqZ<\chisq$: that part of the chi-square price was being paid for nothing. This is sharper than the Kullback--Leibler bound of Theorem~\ref{prop:main}(iii), which only sees $\norm{a_\perp}$, and it is the mechanism behind the non-monotonicity in Remark~\ref{rem:family}.

\begin{corollary}[Interaction bound in the Gaussian model]\label{cor:gaussinter}
In Example~\ref{ex:gauss}, $\Delta_\Pi=0$ if $a_\perp=0$ or $\beta_\perp=0$, and otherwise
\[
\Delta_\Pi\ \le\ \tfrac12\,\bigl|\cos\angle(a_\perp,\beta_\perp)\bigr|\,\sqrt{e^{\mu^2\norm{a_\perp}^2}-1}.
\]
\end{corollary}
\begin{proof}
If $a_\perp=0$ or $\beta_\perp=0$ then $\langle\beta_\perp,a_\perp\rangle=0$ and \eqref{eq:gleak} gives $\Delta_\Pi=0$. Otherwise, on the fibre over $z$, $X_\perp\sim N(0,I-\Pi^{\!\top}\Pi)$ under $P$; the ratio
\[
\frac{w}{w_Z(z)}=\exp\bigl(\mu\langle a_\perp,X_\perp\rangle-\mu^2\norm{a_\perp}^2/2\bigr)
\]
depends on $X_\perp$ through $U=\langle a_\perp,X_\perp\rangle$, and the response kernel $N(\langle\Pi\beta,z\rangle+\langle\beta_\perp,X_\perp\rangle,\sigma^2)$ depends on it through $V=\langle\beta_\perp,X_\perp\rangle$. The pair $(U,V)$ is centred Gaussian with correlation
\[
\frac{\langle a_\perp,\beta_\perp\rangle}{\norm{a_\perp}\,\norm{\beta_\perp}},
\]
whose maximal correlation is its absolute value \cite{lancaster1957}, and
\[
\chi^2(Q_{X|Z=z}\|P_{X|Z=z})=e^{\mu^2\norm{a_\perp}^2}-1
\]
by the lognormal second moment. Apply \eqref{eq:maxcorr}.
\end{proof}

For small $\mu\norm{a_\perp}$ the right-hand side is $\approx\tfrac12\mu|\langle a_\perp,\beta_\perp\rangle|/\norm{\beta_\perp}$, which exceeds the exact first-order leakage $\mu|\langle a_\perp,\beta_\perp\rangle|/(\sqrt{2\pi}s)$ of \eqref{eq:gleak} by the factor $(\sqrt{2\pi}/2)\,s/\norm{\beta_\perp}$: about $1.25$ when the noise is negligible relative to the discarded signal, and arbitrarily large when the noise dominates, because the general bound does not see $\sigma$ (it bounds the variability of $P_{Y|X}(A)$ on the fibre by $\tfrac14$). It also grows with the standard deviation of the fibre ratio, so it is a small-shift bound. On the realisation of experiment E3 (random sketch, $k=8$, $d=40$, $s/\norm{\beta_\perp}=1.06$) the exact leakage, the interaction bound and the Kullback--Leibler bound of Theorem~\ref{prop:main}(iii) are $(0.012,0.017,0.22)$ at $\mu=0.5$, $(0.035,0.078,0.67)$ at $\mu=1.5$ and $(0.059,0.41,1.11)$ at $\mu=2.5$: the interaction bound is the right order of magnitude where the KL bound is vacuous, and it captures the mechanism, since it vanishes with $\langle a_\perp,\beta_\perp\rangle$ while the KL bound depends on $\norm{a_\perp}$ alone.

\paragraph{Random sketches retain, on average, a $k/d$ fraction of everything.} If $\Pi$ is a Haar-random $k$-frame then $\E\norm{a_\parallel}^2=k/d$ and $\E\norm{\beta_\perp}^2=(1-k/d)\norm\beta^2$ (Lemma~\ref{lem:haar}). A Johnson--Lindenstrauss map, which is $\sqrt{d/k}\,\Pi$ up to the choice of distribution, preserves pairwise distances among a \emph{given} set of $n$ points with $k=O(\varepsilon^{-2}\log(n/\sqrt\delta))$ (Lemma~\ref{lem:jl}), but the retained energy of a fixed direction not determined by those points is $k/d$ on average, whatever the scaling. Sketched calibration is invariant to invertible linear maps of $Z$ (the ratio $w_Z$ and the score fitted in $Z$ transform together), so the scaling is immaterial here; it matters only for distance-based downstream uses (Section~\ref{sec:disc}). Random sketching is therefore the right tool for the parts of conformal inference that consume distances and the wrong tool for recovering a direction.

\paragraph{Informed sketches pay only for the shift that matters.}
\begin{corollary}\label{cor:informed}
In Example~\ref{ex:gauss}:
\begin{enumerate}[leftmargin=1.8em,itemsep=1pt]
\item[(a)] (Sufficient sketch) If $\beta\in\operatorname{row}(\Pi)$ then $\Delta_\Pi=0$ for every $a$ and every $\mu$, the sketched oracle equals the full oracle, and $\chisqZ=e^{\mu^2\norm{a_\parallel}^2}-1$ with $\norm{a_\parallel}^2\ge\langle a,\beta\rangle^2/\norm\beta^2$. For $\Pi=\beta^{\!\top}/\norm\beta$ (one row) this is $e^{\mu^2\langle a,\beta\rangle^2/\norm\beta^2}-1$.
\item[(b)] (Estimated sufficient sketch) If $\operatorname{row}(\Pi)\ni\hat\beta$ for an estimate $\hat\beta$ independent of the calibration fold, then
$\Delta_\Pi\le\mu\norm\beta\sin\angle(\beta,\hat\beta)\big/\bigl(\sqrt{2\pi}\,\sigma\bigr)$.
\item[(c)] (Shift-aware sketch) If $\operatorname{row}(\Pi)\ni\hat a$ for an estimate $\hat a$ of the shift direction independent of the calibration fold, then
$\Delta_\Pi\le\mu\sin\angle(a,\hat a)/\sqrt{2\pi}$.
\end{enumerate}
\end{corollary}
\begin{proof}
(a) $\beta_\perp=0$, so $s=\sigma$ and \eqref{eq:gleak} is zero; $\norm{a_\parallel}\ge|\langle a,\beta/\norm\beta\rangle|$ because $\beta/\norm\beta\in\operatorname{row}(\Pi)$. (b) $\norm{\beta_\perp}$ is the distance from $\beta$ to $\operatorname{row}(\Pi)\ni\hat\beta$, hence at most the distance from $\beta$ to the line spanned by $\hat\beta$:
\[
\norm{\beta_\perp}\le\norm\beta\sin\angle(\beta,\hat\beta);
\]
use the first inequality in \eqref{eq:gleak} with $s\ge\sigma$. (c) Likewise $\norm{a_\perp}\le\sin\angle(a,\hat a)$; use the last inequality in \eqref{eq:gleak}.
\end{proof}

The Gaussian response is used in Corollary~\ref{cor:informed}(b) only through the smoothing it provides; the following version isolates that.
\begin{proposition}[Lipschitz response kernels]\label{prop:lip}
In Example~\ref{ex:gauss} replace the response by $Y\mid X=x\sim K(\langle\beta,x\rangle,\cdot)$ for a Markov kernel $K$ on $\R$ satisfying $\TV(K(t,\cdot),K(t',\cdot))\le L_K|t-t'|$. Then $\delta_\Pi(z)\le L_K\,\mu\,|\langle\beta_\perp,a_\perp\rangle|$ for every $z$, and if $\hat\beta\in\operatorname{row}(\Pi)$ then $\Delta_\Pi\le L_K\,\mu\norm\beta\sin\angle(\beta,\hat\beta)$. The homoscedastic Gaussian kernel $N(t,\sigma^2)$ has $L_K=1/(\sqrt{2\pi}\sigma)$, which recovers Corollary~\ref{cor:informed}(b). A heteroscedastic Gaussian kernel $N(t,\sigma(t)^2)$ with $\sigma(\cdot)\ge\sigma_{\min}>0$ and $\sigma$ $L$-Lipschitz has
\[
L_K\ \le\ \frac{1/\sqrt{2\pi}+2\varphi(1)L}{\sigma_{\min}} ,
\]
with no upper bound on $\sigma$ required.
\end{proposition}
\begin{proof}
Appendix~\ref{app:lip}.
\end{proof}

The heteroscedastic model used in experiment E1 below has $\sigma(t)=0.5+0.5\,(t/\norm\beta)_+$ with $\norm\beta=3$, so $\sigma_{\min}=\tfrac12$ and $L=\tfrac16$, and Proposition~\ref{prop:lip} gives $L_K\le0.960$: for a sketch containing $\hat\beta$, the leakage in that model is at most $0.960\,\mu\norm\beta\sin\angle(\beta,\hat\beta)$. The oracle index sketch has zero leakage there by Theorem~\ref{prop:main}(iv)(b).

Case (a) is the point of the paper in one line. With the one-dimensional sufficient sketch, exact weighting pays $e^{\mu^2}-1$ and sketched calibration pays $e^{\mu^2\langle a,\beta\rangle^2/\norm\beta^2}-1$, for the same coverage guarantee. If the signal lives in one direction that carries a fraction $1/d$ of the shift energy, $1+\chi^2$ drops from $e^{\mu^2}$ to $e^{\mu^2/d}$. Case (b) says the saving survives estimation with a leakage allowance linear in the angular error of $\hat\beta$, which can be any training-fold estimate (least squares, a first sliced-inverse-regression direction \cite{li1991}, the gradient of a fitted network). Case (c) is the complementary strategy when the response model is unknown but unlabeled target covariates are available, so that $a$ can be estimated by a source-versus-target classifier; the two can be combined by appending both directions.

\paragraph{A finite-sample bound for the Gaussian model.}\label{sec:lognormal}

The ratio $w=e^{\mu\langle a,X\rangle-\mu^2/2}$ of Example~\ref{ex:gauss} is lognormal under $P$ and unbounded, so the Bernstein bound of Theorem~\ref{thm:price}(a) does not apply directly to the experiments of Section~\ref{sec:exp}. Clipping closes the gap: the clipped rule is exactly valid for a surrogate target with a bounded ratio, its bias has a closed form, and both can be balanced.

\begin{theorem}[Clipped weights in the Gaussian model]\label{thm:lognormal}
In Example~\ref{ex:gauss} with $\mu>0$, run Algorithm~\ref{alg:wcp} with the clipped weights $v_M=w\wedge M$ for a fixed $M>0$, and suppose the score has a continuous distribution under $Q$. Let $c:=(\log M)/\mu+\mu/2$ and
\begin{gather*}
e_M:=\bar\Phi(c-\mu)-M\bar\Phi(c),\qquad c_M:=1-e_M,\\
B_M:=\frac{M}{c_M},\qquad K_M^2:=\frac{e^{\mu^2}\Phi(c-2\mu)+M^2\bar\Phi(c)}{c_M^2},
\end{gather*}
with $\bar\Phi=1-\Phi$. Then with probability at least $1-\delta$ over the calibration sample,
\[
\abs{\Prob_Q\{Y\in\Cset^{v_M}(X)\mid\mathcal D^{\mathrm{cal}}_m\}-(1-\alpha)}\ \le\ 2\eta_m(\delta;K_M^2,B_M)+\frac{B_M}{m}+e_M,
\]
where $\eta_m(\delta;K^2,B):=\sqrt{2K^2\log(16\sqrt m/\delta)/m}+2B\log(16\sqrt m/\delta)/(3m)+m^{-1/2}$ is \eqref{eq:eta} with $(1+\chisq,w_{\max})$ replaced by $(K^2,B)$. For the sketched rule with an orthonormal $\Pi$ such that $\norm{a_\parallel}>0$ and clipped sketched weights $w_Z\wedge M$, the same bound holds with $\mu$ replaced by $\mu\norm{a_\parallel}$ throughout and $\Delta_\Pi$ added to the right-hand side (if $\norm{a_\parallel}=0$ then $w_Z\equiv1$ and Theorem~\ref{thm:price}(a) applies with $K^2=B=1$). For the unclipped rule and any fixed $M>0$, with probability at least $1-\delta-\delta'$, simultaneously for every $M'>0$,
\begin{multline*}
\abs{\Prob_Q\{Y\in\Cset^{w}(X)\mid\mathcal D^{\mathrm{cal}}_m\}-(1-\alpha)}\\
\le\ 2\Bigl(\eta_m(\delta;c_M^2K_M^2,M)+2e_M+\sqrt{\tfrac{e^{\mu^2}\bar\Phi(c-2\mu)}{m\delta'}}\Bigr)+\frac{M'}{m}+\bar\Phi\Bigl(\frac{\log M'}{\mu}-\frac\mu2\Bigr).
\end{multline*}
\end{theorem}
\begin{proof}
Appendix~\ref{app:lognormal}.
\end{proof}

The bound is explicit but loose, as finite-sample bounds built from a uniform Bernstein step are. Minimising the clipped-rule bound $2\eta_m(\delta;K_M^2,B_M)+B_M/m+e_M$ over $M$ at $\delta=0.05$ gives, for clipped full-space weights and for clipped sketched weights with a sketch retaining $\norm{a_\parallel}^2=0.49$ of the shift (the index component of experiment E1),
\begin{center}\small
\begin{tabular}{lcccccc}
\toprule
 & \multicolumn{3}{c}{clipped full-space weights} & \multicolumn{3}{c}{clipped sketched weights, $\norm{a_\parallel}^2=0.49$}\\
$m$ & $\mu=1$ & $\mu=1.5$ & $\mu=2.5$ & $\mu=1$ & $\mu=1.5$ & $\mu=2.5$\\
\midrule
$5\cdot10^3$ & 0.26 & 0.42 & 0.96 & 0.21 & 0.28 & 0.53\\
$5\cdot10^4$ & 0.085 & 0.15 & 0.48 & 0.066 & 0.090 & 0.20\\
$5\cdot10^5$ & 0.027 & 0.050 & 0.21 & 0.021 & 0.029 & 0.072\\
\bottomrule
\end{tabular}
\end{center}
At $m=500$, the calibration size of the experiments, every entry exceeds $0.6$, so the theorem certifies the experiments only in the sense of giving a finite bound that shrinks at the root-$m$ rate, up to logarithmic factors, with the lognormal variance in front; the ordering between full and sketched weighting is the one observed, and the sharper constants of Appendix~\ref{app:sharper} improve each entry by a constant factor without changing the picture. The value of the result is the explicit dependence on $(\mu,M,m)$, which the bounded-ratio theory did not provide for this model.
\section{Leakage is unavoidable for fixed-score threshold calibration}\label{sec:lb}

The leakage in \eqref{eq:val} is not an artefact of the proof. For the problem that Theorem~\ref{thm:price}(b) addresses, choosing a threshold for a prescribed score, any rule that sees only sketched data must pay it, even with unlimited calibration data, full knowledge of the sketched laws, and randomisation.

\begin{theorem}[Lower bound]\label{thm:lb}
Fix $\alpha\in(0,1)$ and $0<\xi\le\tfrac12\min\{\alpha,1-\alpha\}$. There exist a covariate space $\mathcal X$ with three points, a sketch $T$ identifying two of them, source and target marginals $\PX,\QX$ with $\QX\ll\PX$, a fixed score $S^\star$, and two response kernels $P^0_{Y|X},P^1_{Y|X}$, both satisfying Assumption~\ref{ass:A1} with $(\mu_{\mathrm{low}},\mu_{\mathrm{up}},r_0)=(\tfrac14,\tfrac34,\tfrac{2\alpha}{3})$, such that:
\begin{enumerate}[leftmargin=1.8em,itemsep=1pt]
\item[(a)] under both kernels the sketched source law $P^T$, the sketched target marginal $Q_Z$ and the sketched ratio $w_Z$ are identical, and the leakage is $\Delta_T=\xi/2$;
\item[(b)] for every rule $\widehat q=\widehat q(\mathcal D,U)$ that maps sketched calibration data $\mathcal D$ of any size $m$ (together with full knowledge of $P^T,Q_Z,w_Z$) and an independent randomisation variable $U$ to a threshold in $\Rbar$,
\[
\max_{j\in\{0,1\}}\ \E_{\mathcal D,U}\,\abs{\Cov^{(j)}(\widehat q)-(1-\alpha)}\ \ge\ \frac{\alpha\,\xi}{2\alpha+\xi}=\frac{\alpha\,\Delta_T}{\alpha+\Delta_T},
\]
where $\Cov^{(j)}(q):=\Prob\{S^\star(X,Y)\le q\}$ for $(X,Y)\sim\QX\otimes P^j_{Y|X}$ and fixed $q$,
and the bound is attained by a deterministic rule, so it is the exact minimax value of this two-kernel problem.
\end{enumerate}
\end{theorem}
\paragraph{Idea of the construction.}
The proof (Appendix~\ref{app:lb}) uses three covariate points, two of which the sketch identifies; the target sits on one of the identified points, and the two response kernels differ only by swapping the responses at the two identified points, with the uniform law $A$ on $[-1,1]$ against a perturbation $B^{(\xi)}$ that moves mass $\xi$ from the inner score band to the outer one. The sketched data then have the same law under both kernels, so every rule outputs the same distribution of thresholds under both, while the target coverages of a threshold differ by $\xi$ times a known slope. Clipping the threshold to the outer band costs nothing, and the minimax value of the resulting one-dimensional game is $\alpha\xi/(2\alpha+\xi)$, attained by a constant threshold. The restriction $\xi\le\tfrac12\min\{\alpha,1-\alpha\}$ is what keeps both kernels inside Assumption~\ref{ass:A1} (Lemma~\ref{lem:A1check}).

\begin{remark}[The threshold restriction is essential]\label{rem:nonthreshold}
Theorem~\ref{thm:lb} is a statement about threshold rules for the prescribed score $S^\star$, and it does not extend to arbitrary set-valued rules. In the construction of Appendix~\ref{app:lb}, where $A$ is the uniform law on $[-1,1]$ and the score $S^\star(x,y)=|y|-(1-\alpha)$ splits it into the bands $B_0=\{|y|\le1-\alpha\}$ and $B_1=\{1-\alpha<|y|\le1\}$ of mass $1-\alpha$ and $\alpha$, the set $C$ consisting of a $(1-\alpha)$ fraction of $B_0$ together with a $(1-\alpha)$ fraction of $B_1$ (for instance $C=\{|y|\le(1-\alpha)^2\}\cup\{1-\alpha<|y|\le1-\alpha+\alpha(1-\alpha)\}$) has $A(C)=B(C)=1-\alpha$, so a rule that outputs $C$ has zero error under both kernels. Such a set is not a sublevel set of $S^\star$; it uses knowledge of where the two kernels differ. The lower bound therefore quantifies the price of the \emph{calibration} step of split conformal prediction, where the score is fixed, which is the setting of Theorem~\ref{thm:price}(b) and of \cite{isaev2026}.
\end{remark}

\paragraph{Upper and lower bounds together.}
For the oracle surrogate threshold $q^\ast$ at which the surrogate coverage is exactly $1-\alpha$ (which exists whenever the sketched score is continuous under $\Qt$; under the bounded-ratio Assumption~\ref{ass:bdd}, Corollary~\ref{cor:profile} shows that the sketched rule's surrogate coverage approaches $1-\alpha$ as $m\to\infty$, whether or not $q^\ast$ is unique), the total-variation step of Theorem~\ref{prop:main} gives the two-sided bound
$\abs{Q^T\{S_Z\le q^\ast\}-(1-\alpha)}\le\Delta_T$.
Theorem~\ref{thm:lb} says that, for threshold calibration of a prescribed score, no rule can guarantee worst-case error below $\alpha\Delta_T/(\alpha+\Delta_T)$. In the infinite-calibration limit the price of sketching is therefore between $\alpha\Delta_T/(\alpha+\Delta_T)$ and $\Delta_T$, which agree to first order as $\Delta_T/\alpha\to0$. At finite $m$ the sketched rule's error is bounded by Corollary~\ref{cor:profile}, which adds the calibration term; the one-sided guarantee \eqref{eq:val} is not an upper bound on the absolute error, since overcoverage is possible even when $\Delta_T=0$.

\begin{remark}
The construction is the two-carrier experiment of Theorem~\ref{thm:price}(b) with the target atom doubled inside a fibre of the sketch. That theorem's $\sqrt{(1+\chi^2)/m}$ is a finite-sample lower bound on the expected calibration error that vanishes as $m\to\infty$; the leakage bound is a \emph{bias} lower bound that does not. Sketched calibration trades one for the other, and Theorem~\ref{prop:main} and Theorem~\ref{thm:lb} together say the trade is $\chisqZ$ against $\Delta_T$.
\end{remark}
\section{Experiments}\label{sec:exp}

\begin{figure}[!t]\centering
\includegraphics[width=\linewidth]{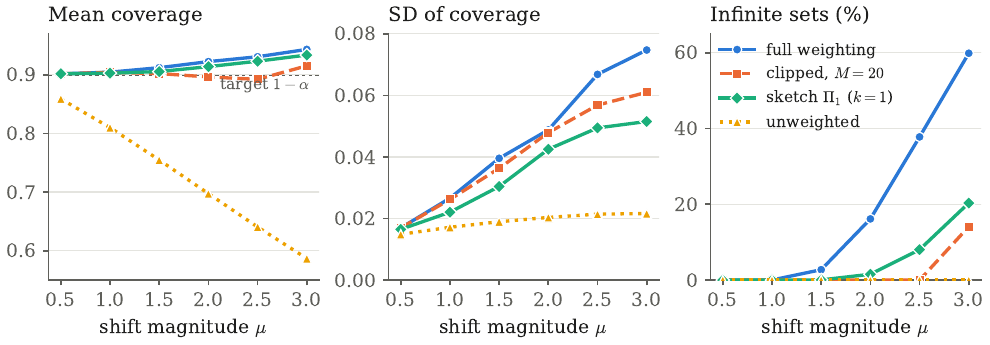}
\caption{A one-dimensional informed sketch keeps the coverage of full weighting with far fewer infinite sets, while unweighted calibration undercovers. E1: heteroscedastic single-index model, shift with component $0.7$ along the index, $200$ replications per point, $m=500$. Left: mean realised coverage. Middle: standard deviation of the realised coverage across replications (a measure of variability, not of uncertainty in the means, whose Monte Carlo standard errors are below $0.006$). Right: percentage of test points whose prediction set is $\R$. The full-space divergence $\chisq=e^{\mu^2}-1$ is $0.3,1.7,8.5,54,517,8100$ at $\mu=0.5,\dots,3$; the sketched divergence of $\Pi_1$ is $0.13,0.63,2.0,6.1,20,81$. The eight-dimensional sketch is omitted for legibility and reported in Table~\ref{tab:hetero}.}\label{fig:hetero}
\end{figure}

Experiments E1--E6 are in the Gaussian single-index model of Example~\ref{ex:gauss} or its heteroscedastic variant, with $d=40$, $n=2000$ training and $m=500$ calibration points, $\alpha=0.1$, $\norm\beta=3$, split conformal with the absolute-residual score $|y-\hat\mu(z)|$ on a least-squares fit in the sketched space, and thresholds computed by Algorithm~\ref{alg:wcp} exactly as stated. Target coverage is estimated on $2\cdot10^4$ test points per replication. Every table reports, per rule: \emph{Coverage}, the mean realised target coverage over replications; \emph{SD}, the standard deviation of the realised coverage across replications, which includes training, sketch and test-sample variability as well as calibration noise; \emph{Median length}, the median over replications of the within-replication median of $2\widehat Q^w(z)$ over test points, with infinite thresholds included (it is finite whenever fewer than half the thresholds are infinite, and is reported as $\infty$ otherwise); and \emph{Infinite sets}, the mean percentage of test points whose prediction set is $\R$. The rules are \emph{full weighting} (exact $w$, $\Pi=I$), \emph{clipped} ($w\wedge20$), \emph{sketch} $\Pi_1=\hat\beta^{\!\top}/\norm{\hat\beta}$ (one row, $\hat\beta$ the least-squares coefficient from the training fold), \emph{sketch} $\Pi_8=[\hat\beta/\norm{\hat\beta};\ 7\text{ random rows}]$, and \emph{unweighted}. Known weights are the closed forms of Proposition~\ref{prop:gauss}. Since $w$ is unbounded in these models, the SD column tests the heuristic that $\chisqZ$ governs calibration variability rather than a theorem; the comparison is between rules at equal $m$. Full details are in Appendix~\ref{app:exp}; the code, results and logs that reproduce every table and figure are in the public repository and in the supplementary material\footnote{\url{https://github.com/mehrdad-pournaderi/Sketched-Calibration-for-Conformal-Prediction-under-Covariate-Shift}}. We present first the setting in which weighting is necessary (E1) and its estimated-ratio version (E2), then three controlled experiments that isolate the mechanisms (E3--E5), the full-score variant (E6), a nonlinear model in which the sketch must be learned (E7), and two real data sets with controlled covariate-dependent sampling (E8).

\paragraph{E1: informed sketches where weighting is necessary.}
Figure~\ref{fig:hetero} and Table~\ref{tab:hetero} use the heteroscedastic model $\sigma(x)=0.5+0.5\,(\langle\beta,x\rangle/\norm\beta)_+$ with a shift direction $a$ that has component $0.7$ along $\beta/\norm\beta$ and a random orthogonal component, so that the shift changes the residual law and weighting is necessary: unweighted coverage falls from $0.86$ at $\mu=0.5$ to $0.64$ at $\mu=2.5$. Full weighting is valid on average at every $\mu$, but at $\mu=2.5$, where $\chisq=517$, it returns the whole real line for $37.8\%$ of test points and has coverage SD $0.067$, which is what an effective calibration size of $500/518\approx1$ looks like; at $\mu=3$ the fraction is $60\%$ and the median set is $\R$. Clipping at $M=20$ removes the infinite sets up to $\mu=2.5$ but at that shift produces coverage below $0.85$ in $21\%$ of replications, mean coverage $0.892$, and the longest intervals; its bias upper bound $\E[(w-20)_+]$ is positive at every $\mu\ge1.5$. The one-dimensional sketch $\Pi_1$ pays only for the component of the shift along the estimated index: at $\mu=2.5$ its divergence is $\chisqZ=20.3$ against $517$, it returns infinite sets for $8.0\%$ of test points against $37.8\%$, its coverage SD is $0.049$ against $0.067$, its mean coverage is $0.923$, and its fraction of replications with coverage below $0.85$ is $0.07$ (full weighting $0.06$, clipped $0.21$, unweighted $1.00$). The oracle index is sufficient for $Y$ in this model, so an oracle sketch along $\beta$ would have zero leakage by Theorem~\ref{prop:main}(iv)(b); the estimated sketch has leakage through the error of $\hat\beta$, at most $0.960\,\mu\norm\beta\sin\angle(\beta,\hat\beta)$ by Proposition~\ref{prop:lip}. For a typical least-squares fit at these sizes, $\norm{\hat\beta-\beta}^2\approx d\,\E[\sigma(X)^2]/n\approx0.011$ and $\sin\angle(\beta,\hat\beta)\approx0.035$, giving approximate allowances of $0.10$, $0.15$ and $0.25$ at $\mu=1,1.5,2.5$; the realised coverages are far inside them. Adding seven random rows ($\Pi_8$) raises $\chisqZ$ to $38.1$ and the infinite-set rate to $12.8\%$ at $\mu=2.5$ without visibly changing coverage: for an oracle index, extra directions that carry shift but no signal are pure cost (Corollary~\ref{cor:informed}(a)); for an estimated index they can retain part of the discarded signal and reduce leakage, an effect too small to see here. The unweighted intervals are the shortest in every row of Table~\ref{tab:hetero} and the least valid, which is why length is never read without coverage.

\begin{table}[!htbp]\centering\footnotesize\setlength{\tabcolsep}{5pt}
\begin{tabular}{lccccc}
\toprule
rule & $\chisqZ$ & Coverage & SD & Median length & Infinite sets (\%)\\
\midrule
\multicolumn{6}{l}{$\mu=1.0$,\quad $\chisq=1.7$}\\
full weighting & -- & 0.905 & 0.027 & 3.28 & 0.0\\
clipped, $M=20$ & -- & 0.904 & 0.026 & 3.28 & 0.0\\
sketch $\Pi_1$ & 0.63 & 0.903 & 0.022 & 3.31 & 0.0\\
sketch $\Pi_8$ & 0.79 & 0.903 & 0.024 & 3.29 & 0.0\\
unweighted & -- & 0.810 & 0.017 & 2.39 & 0.0\\
\addlinespace
\multicolumn{6}{l}{$\mu=1.5$,\quad $\chisq=8.5$}\\
full weighting & -- & 0.912 & 0.040 & 3.66 & 2.7\\
clipped, $M=20$ & -- & 0.902 & 0.036 & 3.65 & 0.0\\
sketch $\Pi_1$ & 2.01 & 0.906 & 0.030 & 3.76 & 0.0\\
sketch $\Pi_8$ & 2.71 & 0.906 & 0.034 & 3.75 & 0.2\\
unweighted & -- & 0.755 & 0.019 & 2.39 & 0.0\\
\addlinespace
\multicolumn{6}{l}{$\mu=2.5$,\quad $\chisq=517$}\\
full weighting & -- & 0.931 & 0.067 & 4.56 & 37.8\\
clipped, $M=20$ & -- & 0.892 & 0.057 & 5.03 & 0.0\\
sketch $\Pi_1$ & 20.3 & 0.923 & 0.049 & 4.64 & 8.0\\
sketch $\Pi_8$ & 38.1 & 0.927 & 0.051 & 4.59 & 12.8\\
unweighted & -- & 0.641 & 0.021 & 2.39 & 0.0\\
\bottomrule
\end{tabular}
\caption{An informed one-dimensional sketch keeps the coverage of full weighting while returning a fifth as many infinite sets; at the largest shift clipping has coverage below $0.85$ in a fifth of replications; unweighted calibration fails. E1: heteroscedastic model, $\langle a,\beta\rangle/\norm\beta=0.7$, $200$ replications. SD is the standard deviation of the realised coverage across replications. At $\mu=2.5$ the fraction of replications with coverage below $0.85$ is $0.06$ (full), $0.21$ (clipped), $0.07$ ($\Pi_1$), $0.05$ ($\Pi_8$), $1.00$ (unweighted). $\chisqZ$ for $\Pi_8$ is averaged over its random rows.}\label{tab:hetero}
\end{table}

\FloatBarrier
\paragraph{E2: estimated ratios.}
Table~\ref{tab:est} repeats E1 at $\mu\in\{1.5,2.5\}$ with a different random orthogonal component of $a$ and the ratio estimated by logistic regression of source-versus-target labels on $2000$ unlabeled covariates from each population (a sample independent of the calibration fold), either in $\R^{40}$ or in the sketched space; the estimated odds are used directly as weights, which is legitimate by scale invariance. The last column reports the normalised $L^1$ error $\tfrac12\E_{P_Z}|\bar w_Z-w_Z|$ of Proposition~\ref{prop:est}, computed in closed form: for a fitted logit $\hat f(z)=\hat b^{\!\top}z+\hat c$ and $P_Z=N(0,I_k)$, the normalised estimate is $\bar w_Z(z)=\exp(\hat b^{\!\top}z-\norm{\hat b}^2/2)$, the ratio of $N(\hat b,I_k)$ to $N(0,I_k)$, so the error is $\TV(N(\hat b,I_k),N(\mu\Pi a,I_k))$, which after rotating coordinates so that $\hat b-\mu\Pi a$ lies along the first axis reduces to the one-dimensional Lemma~\ref{lem:gausstv} and equals $2\Phi(\norm{\hat b-\mu\Pi a}/2)-1$ (a plug-in evaluation on a finite source sample would be dominated by the sampling error of the normaliser). Logistic regression is well specified for a Gaussian mean shift, so even the $40$-dimensional estimate is accurate enough not to move coverage at these sample sizes; the estimation error term is nevertheless an order of magnitude smaller in the one-dimensional sketched space ($0.012$ versus $0.102$ at $\mu=1.5$, $0.017$ versus $0.149$ at $\mu=2.5$), and the sketched rules keep their lower SD and far lower infinite-set rate with estimated weights: the benefit of E1 survives ratio estimation.

\begin{table}[!htbp]\centering\footnotesize\setlength{\tabcolsep}{5pt}
\begin{tabular}{lccccc}
\toprule
rule & Coverage & SD & Median length & Infinite sets (\%) & $\tfrac12\E|\bar w_Z-w_Z|$\\
\midrule
\multicolumn{6}{l}{$\mu=1.5$,\quad $\chisq=8.5$}\\
full weighting, known $w$ & 0.913 & 0.042 & 3.72 & 2.7 & --\\
full weighting, estimated $w$ ($d=40$) & 0.912 & 0.043 & 3.72 & 2.9 & 0.102\\
sketch $\Pi_1$, known $w_Z$ & 0.903 & 0.032 & 3.75 & 0.0 & --\\
sketch $\Pi_1$, estimated $w_Z$ & 0.902 & 0.032 & 3.77 & 0.1 & 0.012\\
sketch $\Pi_8$, estimated $w_Z$ & 0.906 & 0.034 & 3.78 & 0.2 & 0.040\\
\addlinespace
\multicolumn{6}{l}{$\mu=2.5$,\quad $\chisq=517$}\\
full weighting, known $w$ & 0.938 & 0.051 & 4.74 & 37.7 & --\\
full weighting, estimated $w$ ($d=40$) & 0.934 & 0.051 & 4.66 & 37.4 & 0.149\\
sketch $\Pi_1$, known $w_Z$ & 0.920 & 0.044 & 4.52 & 7.9 & --\\
sketch $\Pi_1$, estimated $w_Z$ & 0.920 & 0.044 & 4.54 & 7.9 & 0.017\\
sketch $\Pi_8$, estimated $w_Z$ & 0.924 & 0.056 & 4.56 & 12.6 & 0.052\\
\bottomrule
\end{tabular}
\caption{The advantage of the informed sketch survives estimation of the ratio, and the normalised estimation error is an order of magnitude smaller for the one-dimensional sketch (about threefold smaller for the eight-dimensional one). E2: heteroscedastic model of E1 with a different random shift component, $200$ replications; ratios estimated by logistic regression on $2000+2000$ unlabeled covariates. SD as in Table~\ref{tab:hetero}. The last column is the closed-form normalised $L^1$ error of Proposition~\ref{prop:est}, averaged over replications.}\label{tab:est}
\end{table}

\FloatBarrier
\paragraph{E3: a random sketch shows all three effects.}
Table~\ref{tab:random} uses the homoscedastic model ($\sigma=1$), a random orthonormal $\Pi$ with $k=8$, a random unit shift direction $a$ and $\mu=1.5$, over $300$ replications. The closed-form leakage is $\Delta_\Pi=0.035$, so Theorem~\ref{prop:main} guarantees coverage $\ge0.865$; the realised value is $0.901$. The three effects of Section~\ref{sec:three} are all visible. Weight variability: the sketched rule's SD is $0.018$ against $0.030$ for full weighting, consistent with the drop in $\chi^2$ from $8.5$ to $0.61$, and it returns no infinite sets where full weighting returns $2.8\%$. Coverage loss: no undercoverage relative to the nominal $0.90$ (mean coverage $0.901$ against $0.909$ for full weighting), because the discarded shift is nearly orthogonal to the discarded signal and the leakage is only $0.035$. Interval width: the random sketch discards $\norm{\beta_\perp}^2=7.99$ of signal variance against $\sigma^2=1$, and the sketched median length $9.97$ matches the sketched oracle $2z_{0.95}\sqrt{8.99}=9.86$, three times the full oracle $3.29$ that full weighting attains. The shift-aware sketch, which appends $a_\perp/\norm{a_\perp}$ to $\Pi$, recovers the full ratio and the full $\chi^2$ but not the signal, and so combines the variability of full weighting with the width of the random sketch. In this homoscedastic model the residual law of the full oracle predictor is invariant to the shift, and the residual of the sketched oracle shifts in mean only by $\mu\langle\beta_\perp,a_\perp\rangle$, which is small here (leakage $0.035$), so the unweighted rules are approximately valid as well; E3 isolates the mechanisms, not the need for weights.

\begin{table}[!htbp]\centering\footnotesize\setlength{\tabcolsep}{5pt}
\begin{tabular}{lccccc}
\toprule
rule & $\chi^2$ in use & Coverage & SD & Median length & Infinite sets (\%)\\
\midrule
full weighting ($\Pi=I$) & 8.49 & 0.909 & 0.030 & 3.34 & 2.8\\
random sketch, $k=8$, weighted & 0.61 & 0.901 & 0.018 & 9.97 & 0.0\\
shift-aware sketch, $k=9$, weighted & 8.49 & 0.904 & 0.033 & 9.73 & 2.8\\
random sketch, $k=8$, unweighted & -- & 0.899 & 0.014 & 9.93 & 0.0\\
unweighted ($\Pi=I$) & -- & 0.902 & 0.012 & 3.35 & 0.0\\
\bottomrule
\end{tabular}
\caption{A random sketch reduces the coverage SD of full weighting by about $40\%$ and removes its infinite sets, but triples the interval length by discarding signal. E3: homoscedastic model, random sketch and random shift direction, $\mu=1.5$, $300$ replications; guarantee for the sketched weighted rule: coverage $\ge1-\alpha-\Delta_\Pi=0.865$. SD as in Table~\ref{tab:hetero}.}\label{tab:random}
\end{table}

\FloatBarrier
\paragraph{E4: the adversarial case, shift along the signal.}
Table~\ref{tab:signal} takes $a=\beta/\norm\beta$ in the homoscedastic model, so that the discarded shift is perfectly correlated with the discarded signal, and sweeps the dimension $k$ of a random sketch ($\mu=1.5$, $150$ replications). Every row satisfies coverage $\ge1-\alpha-\Delta_\Pi$, and the total-variation bound is visibly loose: it is a worst case over all measurable sets (for Gaussian conditionals it is attained by a half-line), and it need not be sharp for the symmetric residual-score intervals used here. A random sketch with $k\ll d$ is useless here, exactly as Proposition~\ref{prop:gauss} predicts; appending the single direction $\beta/\norm\beta$ restores the full density ratio and zero leakage at any $k$ (not shown).

\begin{table}[!htbp]\centering\footnotesize\setlength{\tabcolsep}{2.4pt}
\begin{tabular}{rcccc|cccc|cc}
\toprule
 & & & & & \multicolumn{4}{c|}{random sketch, weighted} & \multicolumn{2}{c}{full weighting}\\
$k$ & $\norm{a_\parallel}^2$ & $\chisqZ$ & $\Delta_\Pi$ & $1-\alpha-\Delta_\Pi$ & Coverage & SD & Med.\ length & Inf.\ (\%) & Coverage & Med.\ length\\
\midrule
 2 & 0.055 & 0.13 & 0.510 & 0.390 & 0.608 & 0.031 & 10.19 & 0.0 & 0.905 & 3.28\\
 5 & 0.092 & 0.23 & 0.500 & 0.400 & 0.620 & 0.031 & 10.00 & 0.0 & 0.903 & 3.27\\
10 & 0.121 & 0.31 & 0.492 & 0.408 & 0.629 & 0.030 & 9.90 & 0.0 & 0.910 & 3.36\\
20 & 0.483 & 1.96 & 0.375 & 0.525 & 0.751 & 0.037 & 7.88 & 0.0 & 0.904 & 3.28\\
30 & 0.780 & 4.78 & 0.226 & 0.674 & 0.849 & 0.036 & 5.68 & 0.9 & 0.909 & 3.36\\
40 & 1.000 & 8.49 & 0.000 & 0.900 & 0.908 & 0.029 & 3.30 & 2.8 & 0.908 & 3.30\\
\bottomrule
\end{tabular}
\caption{When the shift is aligned with the signal, a random sketch leaks and undercovers, and among the random sketches shown only the full-dimensional one attains nominal coverage (an informed one-dimensional sketch along $\beta$ would be valid in the same model); the coverage guarantee $1-\alpha-\Delta_\Pi$ holds in every row but is loose. E4: homoscedastic model, $a=\beta/\norm\beta$, $\mu=1.5$, $\chisq=8.49$, random sketch of dimension $k$, $150$ replications. Full weighting returns infinite sets for $2.8\%$ of test points in every row. SD as in Table~\ref{tab:hetero}.}\label{tab:signal}
\end{table}

\FloatBarrier
\paragraph{E5: a control where weighting is not needed.}
Table~\ref{tab:homo} repeats E1 in the homoscedastic model with a random shift direction. For the oracle predictor the residual law there does not depend on $x$; for the fitted predictor the residual $Y-\hat\mu(X)=(\beta-\hat\beta)^{\!\top}X-\hat\beta_0+\sigma\xi$ does change under the shift, but by an amount negligible at $n=2000$, which is why the unweighted rows are valid to three decimals. The experiment is included because it isolates the stability of the different weightings from the need for weights: full weighting still returns the whole real line for $38.2\%$ of test points at $\mu=2.5$ with coverage SD $0.079$, and clipping still lengthens the intervals, although here it costs no coverage; while the one-dimensional sketch has coverage $0.901$, SD $0.014$, equal to the unweighted rule, which is the no-shift floor, oracle length, and no infinite sets, because the shift has almost no component along the index ($\chisqZ\le0.13$). It shows that the advantage of $\Pi_1$ in E1 is not an artefact of the heteroscedastic noise.

\begin{table}[!htbp]\centering\footnotesize\setlength{\tabcolsep}{5pt}
\begin{tabular}{lccccc}
\toprule
rule & $\chisqZ$ & Coverage & SD & Median length & Infinite sets (\%)\\
\midrule
\multicolumn{6}{l}{$\mu=1.0$,\quad $\chisq=1.7$}\\
full weighting & -- & 0.903 & 0.021 & 3.34 & 0.0\\
clipped, $M=20$ & -- & 0.903 & 0.021 & 3.34 & 0.0\\
sketch $\Pi_1$ & 0.02 & 0.902 & 0.014 & 3.35 & 0.0\\
sketch $\Pi_8$ & 0.22 & 0.902 & 0.015 & 3.35 & 0.0\\
unweighted & -- & 0.901 & 0.014 & 3.34 & 0.0\\
\addlinespace
\multicolumn{6}{l}{$\mu=1.5$,\quad $\chisq=8.5$}\\
full weighting & -- & 0.907 & 0.031 & 3.31 & 2.8\\
clipped, $M=20$ & -- & 0.906 & 0.027 & 3.34 & 0.0\\
sketch $\Pi_1$ & 0.05 & 0.902 & 0.013 & 3.35 & 0.0\\
sketch $\Pi_8$ & 0.59 & 0.902 & 0.016 & 3.35 & 0.0\\
unweighted & -- & 0.901 & 0.014 & 3.34 & 0.0\\
\addlinespace
\multicolumn{6}{l}{$\mu=2.5$,\quad $\chisq=517$}\\
full weighting & -- & 0.920 & 0.079 & 3.60 & 38.2\\
clipped, $M=20$ & -- & 0.923 & 0.036 & 3.90 & 0.0\\
sketch $\Pi_1$ & 0.13 & 0.901 & 0.014 & 3.35 & 0.0\\
sketch $\Pi_8$ & 3.15 & 0.903 & 0.021 & 3.35 & 0.6\\
unweighted & -- & 0.901 & 0.014 & 3.34 & 0.0\\
\bottomrule
\end{tabular}
\caption{When the shift does not change the residual law, the one-dimensional sketch matches the unweighted rule's variability and length with no infinite sets, while full weighting still returns the whole real line for $38\%$ of test points at the largest shift. E5: homoscedastic model, $\sigma=1$, random shift direction, oracle length $3.29$, mean leakage of both sketches $\le0.017$, $200$ replications. SD as in Table~\ref{tab:hetero}; $\chisqZ$ for $\Pi_8$ averaged over its random rows.}\label{tab:homo}
\end{table}

\FloatBarrier

\paragraph{E6: sketched weights with a full-covariate score.}
Table~\ref{tab:fullscore} revisits the settings of E1 and E3 with the rule of Definition~\ref{def:fullscore}: the predictor is the least-squares fit on all $40$ covariates, and only the weights are sketched, by $\Pi_1$, by $\Pi_8$, by the random $k=8$ sketch of E3, and by the residual sketch of Section~\ref{sec:choose} ($\Pi_{\mathrm{res}}$: the least-squares direction of the cross-fitted absolute residuals on $X$, one row). In the homoscedastic setting the full-score rule with the random sketch has coverage $0.899$ and median length $3.31$, against $9.93$ for the sketched-score rule with the same sketch: the width penalty of E3 is gone, as Proposition~\ref{prop:fullgauss} predicts (the misfit-driven leakage bound there evaluates to $0.009$), and $\chisqZ=0.61$ instead of $8.5$. In the heteroscedastic setting the same random sketch undercovers ($0.804$ at $\mu=1.5$, $0.744$ at $\mu=2.5$), because the residual law depends on the index direction, which the random sketch discards; the sketch must retain the direction along which the noise scale varies, and $\Pi_1$ and $\Pi_8$, which do, reproduce the behaviour of E1 with no loss of length. The residual sketch shows the difficulty of learning where the noise scale varies: at $n=2000$ its direction has angular error $\sin\angle(\Pi_{\mathrm{res}},\beta)\approx0.33$, against $0.035$ for the least-squares index of the mean, and it undercovers by about one point at $\mu=1.5$ ($0.890$) while at $\mu=2.5$ it reaches $0.899$ with $6.1\%$ infinite sets against $37.6\%$ for full weighting. E7 shows the angular error shrinking with the training size.

\begin{table}[!htbp]\centering\footnotesize\setlength{\tabcolsep}{5pt}
\begin{tabular}{llccccc}
\toprule
score & weights & $\chisqZ$ & Coverage & SD & Median length & Infinite sets (\%)\\
\midrule
\multicolumn{7}{l}{(a) heteroscedastic (setting of E1), $\mu=1.5$, $\chisq=8.5$}\\
full & full weighting & -- & 0.907 & 0.039 & 3.62 & 2.7\\
full & unweighted & -- & 0.756 & 0.021 & 2.40 & 0.0\\
sketched & $\Pi_1$ & 2.01 & 0.901 & 0.032 & 3.75 & 0.1\\
sketched & random $k=8$ & 0.83 & 0.787 & 0.026 & 8.73 & 0.0\\
full & $\Pi_1$ & 2.01 & 0.901 & 0.032 & 3.75 & 0.1\\
full & $\Pi_8$ & -- & 0.901 & 0.034 & 3.69 & 0.2\\
full & $\Pi_{\mathrm{res}}$ & 1.73 & 0.890 & 0.032 & 3.58 & 0.0\\
full & random $k=8$ & 0.83 & 0.804 & 0.031 & 2.72 & 0.0\\
\addlinespace
\multicolumn{7}{l}{(a) heteroscedastic, $\mu=2.5$, $\chisq=517$}\\
full & full weighting & -- & 0.931 & 0.052 & 4.64 & 37.6\\
full & unweighted & -- & 0.641 & 0.023 & 2.40 & 0.0\\
sketched & $\Pi_1$ & 20.4 & 0.917 & 0.044 & 4.48 & 8.1\\
sketched & random $k=8$ & 4.33 & 0.631 & 0.050 & 8.75 & 0.7\\
full & $\Pi_1$ & 20.4 & 0.917 & 0.044 & 4.48 & 8.1\\
full & $\Pi_8$ & -- & 0.919 & 0.049 & 4.41 & 12.9\\
full & $\Pi_{\mathrm{res}}$ & 16.0 & 0.899 & 0.049 & 4.27 & 6.1\\
full & random $k=8$ & 4.33 & 0.744 & 0.059 & 2.97 & 0.7\\
\addlinespace
\multicolumn{7}{l}{(b) homoscedastic (setting of E3), $\mu=1.5$, $\chisq=8.5$}\\
full & full weighting & -- & 0.905 & 0.031 & 3.28 & 2.7\\
full & unweighted & -- & 0.899 & 0.014 & 3.32 & 0.0\\
sketched & $\Pi_1$ & 0.00 & 0.900 & 0.014 & 3.33 & 0.0\\
sketched & random $k=8$ & 0.61 & 0.900 & 0.016 & 9.93 & 0.0\\
full & $\Pi_1$ & 0.00 & 0.900 & 0.014 & 3.33 & 0.0\\
full & $\Pi_{\mathrm{res}}$ & 0.06 & 0.900 & 0.014 & 3.33 & 0.0\\
full & random $k=8$ & 0.61 & 0.899 & 0.017 & 3.31 & 0.0\\
\bottomrule
\end{tabular}
\caption{Keeping the predictor on all covariates and sketching only the weights removes the width penalty of a random sketch (length $3.31$ against $9.93$ in (b)) while the sketch must still retain the direction along which the residual law varies (random sketches undercover in (a)); the learned residual direction undercovers by about one point at $n=2000$. E6: $200$ replications; ``score'' says whether the predictor is fitted on the sketch or on all covariates. SD as in Table~\ref{tab:hetero}; $\chisqZ$ is averaged over replications for the learned sketches and omitted for $\Pi_8$.}\label{tab:fullscore}
\end{table}

\FloatBarrier

\paragraph{E7: a nonlinear multiple-index model with learned sketches.}
The Gaussian single-index model hands the sufficient direction to the sketch; E7 asks what happens when it must be learned, and when the predictor is a flexible learner whose misfit is not negligible. Covariates are $X\sim N(0,I_{40})$ under $P$ and $N(\mu a,I_{40})$ under $Q$, the response is
\[
Y=2\sin\langle b_1,X\rangle+\tfrac12\langle b_2,X\rangle^2+\sigma(\langle b_3,X\rangle)\,\xi,\qquad\sigma(t)=0.5+0.5\,t_+,
\]
with $b_1,b_2,b_3$ orthonormal, and the shift direction $a=0.6b_3+0.5b_1+\sqrt{0.39}\,g$ has components along the noise-scale direction, along a mean direction, and along a random direction $g$ orthogonal to all three. The predictor is gradient boosting on all $40$ covariates; all sketched rules are full-score rules (Definition~\ref{def:fullscore}). The \emph{response} law depends on $x$ through the three indices $(\langle b_1,x\rangle,\langle b_2,x\rangle,\langle b_3,x\rangle)$ only, so $[b_1,b_2,b_3]$ is a response-sufficient oracle sketch (the \emph{response-index oracle}) with $\chisqZ=e^{0.61\mu^2}-1$ ($2.95$ at $\mu=1.5$, $44$ at $\mu=2.5$, against $8.5$ and $517$); the fitted predictor, however, uses all $40$ coordinates, so the score $|y-\hat\mu(x)|$ can depend on the remaining coordinates through $\hat\mu$, and Theorem~\ref{thm:fullscore} does not make this oracle leakage-free. The direction $b_3$ alone is sufficient for the noise scale but not for the misfit. The learned candidates are the residual sketch of Section~\ref{sec:choose} ($k=1$, and $k=2$ by sliced inverse regression), a mean-plus-residual sketch that appends the residual direction to the two leading sliced-inverse-regression directions of $Y$ on $X$ ($k=3$), the shift direction estimated by logistic regression on $2000+2000$ unlabelled covariates ($k=1$), and a random sketch ($k=8$). Selection uses $n_t=100$ labelled target points and Proposition~\ref{prop:select} over nine candidates: full weighting, clipping, the unweighted rule, the five learned sketches just listed, and a two-dimensional sketch combining the residual and shift directions, which behaves like the shift direction alone and is omitted from Table~\ref{tab:nonlinear}. Sketched ratios are the closed forms for the realised sketches, so only the sketches are learned. Figure~\ref{fig:nonlinear} varies the training size $n\in\{1000,4000,16000\}$ at $m=500$, and Table~\ref{tab:nonlinear} reports $n=4000$ in full together with a smaller calibration set.

\begin{figure}[!t]\centering
\includegraphics[width=\linewidth]{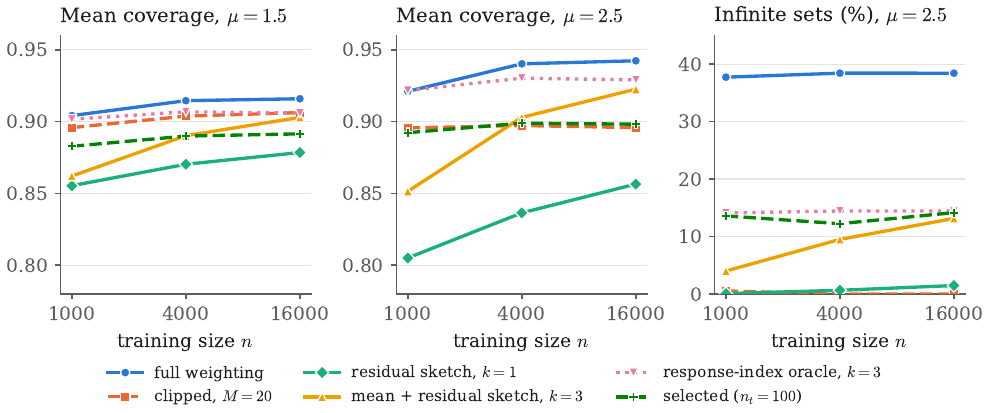}
\caption{Learned sketches improve with the training size and approach the response-index oracle, while a sketch that captures only the noise direction cannot remove the misfit leakage. E7: nonlinear multiple-index model, gradient-boosting predictor on all covariates, $m=500$, $100$ replications per point. Left and middle: mean realised coverage at $\mu=1.5$ and $\mu=2.5$ (dashed line at $0.90$). Right: percentage of test points with an infinite set at $\mu=2.5$. The residual sketch has angular error $\sin\angle(\hat b,b_3)=0.81,0.50,0.28$ at $n=1000,4000,16000$; the estimated shift direction has angular error $\approx0.15$ at every $n$.}\label{fig:nonlinear}
\end{figure}

Three things happen. First, the learned sketches improve with $n$: the residual direction's angular error falls from $0.81$ to $0.28$, and the three-dimensional mean-plus-residual sketch goes from undercoverage ($0.862$ and $0.851$ at $n=1000$) to $0.903$ and $0.923$ at $n=16000$, where it nearly matches the response-index oracle ($0.906$, $0.929$) while returning infinite sets for $13\%$ of test points at $\mu=2.5$ against $38\%$ for full weighting. The oracle exceeds $0.90$ by $0.6$ and $2.9$ points; since weighted conformal prediction can over-cover with zero leakage through finite-sample conservativeness, including the test-point mass at $+\infty$, this experiment does not separate the oracle's possible score leakage from that conservativeness. Second, the one-dimensional sketches fail in an instructive way: the residual direction alone undercovers at every $n$ (by $4$ to $5$ points at $\mu=2.5$ even at $n=16000$), and so does the \emph{oracle} noise direction $b_3$ ($0.873$ at $\mu=2.5$, $n=16000$), because the gradient-boosting misfit varies along $b_1$, which carries part of the shift ($\langle a,b_1\rangle=0.5$; $b_2$ is orthogonal to $a$); this is the misfit term of Proposition~\ref{prop:fullgauss} in a non-Gaussian guise, and it is why the default procedure of Section~\ref{sec:choose} appends mean directions. The estimated shift direction, by contrast, is valid at every $n$ but at essentially the full price ($\chisqZ\approx455$ at $\mu=2.5$), since a sketch containing the shift direction retains the whole shift. Third, selection with $100$ target labels lands within two points of $0.90$ in every configuration ($0.883$--$0.901$), with roughly a third of the infinite sets of full weighting and a smaller coverage SD; its gap to the best candidate in hindsight is at most $1.4$ points. With $m=200$ (Table~\ref{tab:nonlinear}, last block) full weighting returns the whole real line for more than half of the test points and clipping over-covers at $0.961$, while the learned three-dimensional sketch and the selected rule keep coverage near $0.92$ and $0.90$ with $22$--$24\%$ infinite sets.

\begin{table}[!htbp]\centering\footnotesize\setlength{\tabcolsep}{5pt}
\begin{tabular}{lcccc}
\toprule
rule & Coverage & SD & Median length & Infinite sets (\%)\\
\midrule
\multicolumn{5}{p{0.96\linewidth}}{$\mu=1.5$, $n=4000$, $m=500$, $\chisq=8.5$;\ $\chisqZ$: residual $k{=}1$ 0.7, mean$+$residual 2.3, random 0.6, oracle $b_3$ 1.25, oracle $[b_1,b_2,b_3]$ 2.95}\\
full weighting & 0.915 & 0.037 & 4.78 & 2.8\\
clipped, $M=20$ & 0.904 & 0.034 & 4.71 & 0.0\\
unweighted & 0.820 & 0.015 & 3.71 & 0.0\\
residual sketch, $k=1$ & 0.870 & 0.022 & 4.30 & 0.0\\
residual sketch (SIR), $k=2$ & 0.871 & 0.022 & 4.29 & 0.0\\
mean (SIR on $Y$, $k=2$) $+$ residual sketch, $k=3$ & 0.890 & 0.029 & 4.51 & 0.1\\
estimated shift direction, $k=1$ & 0.911 & 0.037 & 4.65 & 2.5\\
random sketch, $k=8$ & 0.836 & 0.023 & 3.85 & 0.0\\
oracle noise direction $b_3$, $k=1$ & 0.889 & 0.022 & 4.53 & 0.0\\
response-index oracle $[b_1,b_2,b_3]$, $k=3$ & 0.907 & 0.032 & 4.74 & 0.2\\
selected by Prop.~\ref{prop:select} ($n_t=100$) & 0.890 & 0.026 & 4.44 & 0.9\\
best candidate in hindsight & 0.895 & 0.014 & 4.54 & 0.7\\
\addlinespace
\multicolumn{5}{p{0.96\linewidth}}{$\mu=2.5$, $n=4000$, $m=500$, $\chisq=517$;\ $\chisqZ$: residual $k{=}1$ 3.7, mean$+$residual 26.7, random 3.0, oracle $b_3$ 8.49, oracle $[b_1,b_2,b_3]$ 44.3}\\
full weighting & 0.940 & 0.050 & 5.81 & 38.4\\
clipped, $M=20$ & 0.897 & 0.056 & 6.58 & 0.0\\
unweighted & 0.733 & 0.018 & 3.71 & 0.0\\
residual sketch, $k=1$ & 0.836 & 0.042 & 4.70 & 0.6\\
residual sketch (SIR), $k=2$ & 0.842 & 0.046 & 4.73 & 1.1\\
mean (SIR on $Y$, $k=2$) $+$ residual sketch, $k=3$ & 0.903 & 0.054 & 4.98 & 9.5\\
estimated shift direction, $k=1$ & 0.934 & 0.051 & 5.55 & 37.2\\
random sketch, $k=8$ & 0.768 & 0.043 & 3.93 & 0.5\\
oracle noise direction $b_3$, $k=1$ & 0.880 & 0.049 & 5.07 & 2.8\\
response-index oracle $[b_1,b_2,b_3]$, $k=3$ & 0.930 & 0.042 & 5.27 & 14.4\\
selected by Prop.~\ref{prop:select} ($n_t=100$) & 0.899 & 0.036 & 5.14 & 12.2\\
best candidate in hindsight & 0.903 & 0.023 & 5.21 & 11.1\\
\addlinespace
\multicolumn{5}{p{0.96\linewidth}}{$\mu=2.5$, $n=4000$, $m=200$, $\chisq=517$;\ $\chisqZ$: residual $k{=}1$ 3.5, mean$+$residual 27.5, random 3.2, oracle $b_3$ 8.49, oracle $[b_1,b_2,b_3]$ 44.3}\\
full weighting & 0.937 & 0.064 & $\infty$ & 52.6\\
clipped, $M=20$ & 0.961 & 0.030 & $\infty$ & 55.2\\
unweighted & 0.735 & 0.031 & 3.70 & 0.0\\
residual sketch, $k=1$ & 0.848 & 0.057 & 4.63 & 3.0\\
residual sketch (SIR), $k=2$ & 0.850 & 0.056 & 4.65 & 5.0\\
mean (SIR on $Y$, $k=2$) $+$ residual sketch, $k=3$ & 0.920 & 0.054 & 5.28 & 22.1\\
estimated shift direction, $k=1$ & 0.938 & 0.061 & $\infty$ & 51.7\\
random sketch, $k=8$ & 0.790 & 0.065 & 3.95 & 2.7\\
oracle noise direction $b_3$, $k=1$ & 0.897 & 0.047 & 5.19 & 10.0\\
response-index oracle $[b_1,b_2,b_3]$, $k=3$ & 0.941 & 0.043 & 5.55 & 28.5\\
selected by Prop.~\ref{prop:select} ($n_t=100$) & 0.901 & 0.034 & 5.30 & 23.5\\
best candidate in hindsight & 0.905 & 0.025 & 5.41 & 24.5\\

\bottomrule
\end{tabular}
\caption{In the nonlinear model a learned three-dimensional sketch (mean directions plus residual direction) keeps coverage near $0.90$ ($0.890$--$0.920$ at $n=4000$) with, at $\mu=2.5$, between a quarter and a half of the infinite sets of full weighting (at $\mu=1.5$ both are below $3\%$), one-dimensional sketches undercover because of misfit leakage, and selection with $100$ target labels tracks the best candidate within $1.4$ points. E7, $100$ replications; the full-score rule is used for every sketch. SD as in Table~\ref{tab:hetero}; $\chisqZ$ values in the block headings are averaged over replications for learned sketches. The estimated shift direction behaves like full weighting ($\chisqZ\approx455$ at $\mu=2.5$).}\label{tab:nonlinear}
\end{table}

\FloatBarrier
\paragraph{E8: real data with controlled covariate-dependent sampling.}
Tables~\ref{tab:california} and~\ref{tab:protein} use two regression data sets, California housing ($20{,}640$ observations, $8$ covariates, response $\log$ median house value) and the UCI protein tertiary structure data ($45{,}730$ observations, $9$ covariates, response RMSD). Covariates are Gaussianised by a quantile transform, winsorised at $\pm3$ and standardised on the pool. A target population is created by exponential tilting: pool points enter the target with probability proportional to $e^{\gamma\langle a,x\rangle}$, with tilt direction $a$ either a random unit vector or a single feature (median income; the first protein feature), and the rules are run with $w(x)\propto e^{\gamma\langle a,x\rangle}$ as the weight. In each of $100$ replications the pool is split into a training sample ($3000$), a calibration sample ($m$), an unlabelled pool used to estimate sketched ratios, and a pool from which $2000$ target test points and $n_t=100$ target selection points are drawn with the tilted probabilities. The predictor is gradient boosting on all covariates and every sketched rule is a full-score rule. This is an approximate finite-pool benchmark rather than an instance of the theorems: calibration points are drawn from one part of a finite data set and target points from a disjoint remainder, so the samples are not i.i.d.\ from two laws with density ratio $w$, and $w$ is the exact ratio only of the tilted and untilted \emph{pool} laws. Moreover, on a finite support a generic linear sketch is injective on the distinct covariate vectors, so the exact sketched ratio $\E[w\mid Z]$ of the pool law equals $w$ itself and the pool divergence does not decrease. The quantities we call $w_Z$ are therefore \emph{smoothed} estimates, obtained by gradient-boosting regression of the known $w$ on $Z$ over the unlabelled pool (disjoint from the calibration, selection and test samples, so that no observation is reused), and the ``$\hat\chi^2_Z$'' reported in the block headings is the dispersion of these estimated weights over the pool, not an estimate of a population divergence. The sketch that consists of the tilt direction itself has $w_Z=w$ and coincides with full weighting, so it is not shown. With these caveats, the experiment asks the practical question: does replacing the known tilt weight by a low-dimensional smoothed weight, chosen by the default procedure, trade coverage for stability in the way the theory suggests?

\begin{table}[!htbp]\centering\footnotesize\setlength{\tabcolsep}{5pt}
\begin{tabular}{lcccc}
\toprule
rule & Coverage & SD & Median length & Infinite sets (\%)\\
\midrule
\multicolumn{5}{p{0.96\linewidth}}{tilt along a random direction, $\gamma=1.5$, $m=500$, pool $\chisq=21.2$;\ $\hat\chi^2_Z$ (dispersion of estimated weights): residual $k{=}1$ 1.1, SIR $k{=}2$ 4.2, mean$+$residual 4.7, random 4.4}\\
full weighting & 0.922 & 0.040 & 0.97 & 6.9\\
clipped, $M=20$ & 0.899 & 0.040 & 0.97 & 0.0\\
unweighted & 0.820 & 0.031 & 0.75 & 0.0\\
residual sketch, $k=1$ & 0.876 & 0.038 & 0.92 & 0.0\\
residual sketch (SIR), $k=2$ & 0.878 & 0.046 & 0.91 & 0.2\\
mean $+$ residual sketch, $k=3$ & 0.889 & 0.041 & 0.97 & 0.1\\
random sketch, $k=2$ & 0.862 & 0.046 & 0.84 & 0.2\\
selected by Prop.~\ref{prop:select} ($n_t=100$) & 0.893 & 0.027 & 0.92 & 1.7\\
best candidate in hindsight & 0.898 & 0.015 & 0.96 & 1.7\\
\addlinespace
\multicolumn{5}{p{0.96\linewidth}}{tilt along a random direction, $\gamma=1.5$, $m=200$, pool $\chisq=21.2$;\ $\hat\chi^2_Z$ (dispersion of estimated weights): residual $k{=}1$ 1.1, SIR $k{=}2$ 4.2, mean$+$residual 4.7, random 4.4}\\
full weighting & 0.932 & 0.048 & 0.92 & 13.4\\
clipped, $M=20$ & 0.928 & 0.042 & 0.97 & 7.4\\
unweighted & 0.824 & 0.038 & 0.77 & 0.0\\
residual sketch, $k=1$ & 0.883 & 0.048 & 0.95 & 0.1\\
residual sketch (SIR), $k=2$ & 0.877 & 0.051 & 0.88 & 0.9\\
mean $+$ residual sketch, $k=3$ & 0.880 & 0.055 & 0.88 & 0.7\\
random sketch, $k=2$ & 0.865 & 0.057 & 0.83 & 1.6\\
selected by Prop.~\ref{prop:select} ($n_t=100$) & 0.890 & 0.033 & 0.87 & 4.0\\
best candidate in hindsight & 0.899 & 0.016 & 0.91 & 4.1\\
\addlinespace
\multicolumn{5}{p{0.96\linewidth}}{tilt along the feature MedInc, $\gamma=2.0$, $m=500$, pool $\chisq=15.1$;\ $\hat\chi^2_Z$ (dispersion of estimated weights): residual $k{=}1$ 2.2, SIR $k{=}2$ 5.7, mean$+$residual 11.5, random 4.8}\\
full weighting & 0.908 & 0.035 & 0.52 & 9.2\\
clipped, $M=20$ & 0.910 & 0.026 & 0.55 & 0.0\\
unweighted & 0.959 & 0.010 & 0.75 & 0.0\\
residual sketch, $k=1$ & 0.922 & 0.020 & 0.59 & 0.0\\
residual sketch (SIR), $k=2$ & 0.921 & 0.028 & 0.58 & 0.1\\
mean $+$ residual sketch, $k=3$ & 0.899 & 0.030 & 0.52 & 1.7\\
random sketch, $k=2$ & 0.951 & 0.023 & 0.71 & 0.2\\
selected by Prop.~\ref{prop:select} ($n_t=100$) & 0.915 & 0.024 & 0.55 & 4.4\\
best candidate in hindsight & 0.905 & 0.011 & 0.53 & 2.0\\
\addlinespace
\multicolumn{5}{p{0.96\linewidth}}{tilt along a random direction, $\gamma=2.0$, $m=500$, pool $\chisq=263.8$;\ $\hat\chi^2_Z$ (dispersion of estimated weights): residual $k{=}1$ 8.1, SIR $k{=}2$ 34.4, mean$+$residual 36.5, random 34.0}\\
full weighting & 0.937 & 0.053 & 1.19 & 28.5\\
clipped, $M=20$ & 0.889 & 0.072 & 1.21 & 0.0\\
unweighted & 0.759 & 0.076 & 0.75 & 0.0\\
residual sketch, $k=1$ & 0.855 & 0.078 & 0.99 & 0.8\\
residual sketch (SIR), $k=2$ & 0.854 & 0.079 & 0.97 & 1.5\\
mean $+$ residual sketch, $k=3$ & 0.872 & 0.072 & 1.08 & 1.1\\
random sketch, $k=2$ & 0.833 & 0.099 & 0.89 & 2.4\\
selected by Prop.~\ref{prop:select} ($n_t=100$) & 0.896 & 0.031 & 1.08 & 5.6\\
best candidate in hindsight & 0.896 & 0.022 & 1.09 & 5.2\\

\bottomrule
\end{tabular}
\caption{California housing (approximate finite-pool benchmark). Under a random-direction tilt the learned sketches reduce the estimated-weight dispersion and the frequency of infinite sets, but undercover by one to seven points and, at the stronger tilt, have larger coverage SD than full weighting, and selection with $100$ target labels recovers coverage near $0.90$ with about one fifth to one third of full weighting's infinite sets; under a tilt along median income unweighted calibration over-covers and every rule is valid. E8, $100$ replications, lengths in $\log$-price units. SD as in Table~\ref{tab:hetero}; $\hat\chi^2_Z$ in the block headings is the dispersion of the smoothed estimated weights over the pool, not a population divergence.}\label{tab:california}
\end{table}

On California housing with a random-direction tilt of divergence $21$ (first two blocks), full weighting covers at $0.92$--$0.93$ with $7$--$13\%$ infinite sets, clipping is valid on average at $m=500$ but has coverage below $0.85$ in $13\%$ of replications, the unweighted rule covers at $0.82$, and the learned sketches sit in between: the residual sketch has estimated-weight dispersion $\hat\chi^2_Z\approx1.1$ against $21$, no infinite sets, and coverage $0.876$--$0.883$; the three-dimensional sketch reaches $0.880$--$0.889$. This experiment measures the coverage effect of replacing the tilt weights by smoothed low-dimensional weights under the stated finite-pool sampling scheme; it does not separately estimate compression leakage, weight-estimation error or sampling effects. Selection with $100$ target labels gives $0.890$--$0.893$ with coverage SD $0.027$--$0.033$ against $0.040$--$0.048$ for full weighting and about one fifth to one third of its infinite sets; the best candidate in hindsight is at $0.898$. With the stronger tilt $\gamma=2$ (divergence $264$) the same pattern is sharper: full weighting returns the whole real line for $28\%$ of test points, clipping has coverage below $0.85$ in a quarter of the replications, the learned sketches undercover by three to seven points, and the selected rule has coverage $0.896$ with $5.6\%$ infinite sets and SD $0.031$ against $0.053$. The tilt along median income (third block) is the opposite situation: the shift moves the test points towards high-income districts where the model's residuals are smaller, so the unweighted rule over-covers at $0.959$, every weighted rule is valid, and the sketches shorten the intervals relative to the unweighted rule (median length $0.58$ against $0.75$) at the cost of up to $9\%$ infinite sets for full weighting; selection picks full weighting or the residual sketch and lands at $0.915$.

\begin{table}[!htbp]\centering\footnotesize\setlength{\tabcolsep}{5pt}
\begin{tabular}{lcccc}
\toprule
rule & Coverage & SD & Median length & Infinite sets (\%)\\
\midrule
\multicolumn{5}{p{0.96\linewidth}}{tilt along a random direction, $\gamma=3.0$, $m=500$, pool $\chisq=21.1$;\ $\hat\chi^2_Z$ (dispersion of estimated weights): residual $k{=}1$ 0.2, SIR $k{=}2$ 2.4, mean$+$residual 10.6, random 6.0}\\
full weighting & 0.913 & 0.035 & 13.96 & 11.8\\
clipped, $M=20$ & 0.914 & 0.033 & 14.42 & 0.0\\
unweighted & 0.927 & 0.015 & 15.21 & 0.0\\
residual sketch, $k=1$ & 0.934 & 0.015 & 15.66 & 0.0\\
residual sketch (SIR), $k=2$ & 0.915 & 0.026 & 14.67 & 0.0\\
mean $+$ residual sketch, $k=3$ & 0.905 & 0.032 & 13.41 & 4.2\\
random sketch, $k=2$ & 0.905 & 0.032 & 13.86 & 1.4\\
selected by Prop.~\ref{prop:select} ($n_t=100$) & 0.907 & 0.021 & 13.96 & 4.0\\
best candidate in hindsight & 0.905 & 0.010 & 13.72 & 3.2\\
\addlinespace
\multicolumn{5}{p{0.96\linewidth}}{tilt along the feature F1, $\gamma=2.0$, $m=500$, pool $\chisq=14.1$;\ $\hat\chi^2_Z$ (dispersion of estimated weights): residual $k{=}1$ 0.6, SIR $k{=}2$ 5.7, mean$+$residual 11.4, random 8.7}\\
full weighting & 0.924 & 0.040 & 17.12 & 6.4\\
clipped, $M=20$ & 0.901 & 0.038 & 17.23 & 0.0\\
unweighted & 0.865 & 0.017 & 15.21 & 0.0\\
residual sketch, $k=1$ & 0.874 & 0.021 & 15.74 & 0.0\\
residual sketch (SIR), $k=2$ & 0.893 & 0.037 & 16.52 & 0.0\\
mean $+$ residual sketch, $k=3$ & 0.910 & 0.039 & 16.95 & 1.6\\
random sketch, $k=2$ & 0.895 & 0.040 & 16.35 & 0.5\\
selected by Prop.~\ref{prop:select} ($n_t=100$) & 0.894 & 0.024 & 16.06 & 1.7\\
best candidate in hindsight & 0.894 & 0.013 & 16.25 & 1.3\\

\bottomrule
\end{tabular}
\caption{Protein tertiary structure (approximate finite-pool benchmark). Under a random-direction tilt the shift is harmless for the residual law, the unweighted rule over-covers, and the three-dimensional sketch gives valid intervals shorter than unweighted calibration with a third of full weighting's infinite sets; under a tilt along the first feature the unweighted rule undercovers and the three-dimensional sketch is valid with a quarter of the infinite sets. E8, $100$ replications, lengths in RMSD units. SD as in Table~\ref{tab:hetero}; $\hat\chi^2_Z$ as in Table~\ref{tab:california}.}\label{tab:protein}
\end{table}

On the protein data the two tilts show the two regimes side by side. Under the random-direction tilt (divergence $21$) the shift barely changes the residual law of a predictor that is, for this response, rather weak: the unweighted rule over-covers at $0.927$, and every sketched rule is valid. The sketches are then a device for \emph{shortening} the intervals: the three-dimensional mean-plus-residual sketch has coverage $0.905$ and median length $13.4$ against $15.2$ for the unweighted rule and $14.0$ for full weighting, with $4\%$ infinite sets against $12\%$; its estimated-weight dispersion is $10.6$ against $21$. Under the tilt along the first feature (divergence $14$) the unweighted rule undercovers at $0.865$, the one-dimensional residual sketch is only marginally better ($0.874$, $\hat\chi^2_Z=0.6$), and the three-dimensional sketch is valid at $0.910$ with $1.6\%$ infinite sets against $6.4\%$ for full weighting; selection lands at $0.894$ with coverage SD $0.024$ against $0.040$. Across the six real-data configurations the selected rule is within $1.5$ points of $0.90$ in every case, which is what Proposition~\ref{prop:select} promises up to its $\sqrt{\log J/n_t}$ term, while no single fixed candidate is.

\FloatBarrier

\section{Discussion and open problems}\label{sec:disc}

\paragraph{What sketching is for.}
The results draw a sharp line. Sketching is a device for \emph{not paying} for shift that does not interact with the response, and its bias is exactly the interaction that is lost. Random Johnson--Lindenstrauss sketches are the right tool wherever the downstream computation consumes pairwise distances among the sample, with kernel-specific caveats: for localised conformal prediction \cite{guan2023} with a Gaussian localiser of bandwidth $h$, a $(1\pm\varepsilon)$ distortion of squared distances changes the weight of a point at distance $r$ by a factor at most $e^{\varepsilon r^2/(2h^2)}$, which is close to one only when $\varepsilon r^2\ll h^2$, and a hard-radius localiser can switch a weight between $0$ and $1$; validity of the localised procedure is unaffected in either case, since it holds for any localiser. Density-ratio estimation for $w_Z$ becomes a $k$-dimensional problem (E2), and sketch-and-solve on the training fold \cite{raskutti2016} with calibration on raw held-out rows is valid without further argument. For capturing a direction, the sketch must be informed, and Corollary~\ref{cor:informed} quantifies how good the information must be.

\paragraph{Choosing the sketch.}
By Theorem~\ref{prop:main}(ii), the sketched divergence is the squared norm of the projection of $w-1$ onto $\sigma(T)$-measurable functions, so it is monotone under refinement of the sketch. The leakage is not (Remark~\ref{rem:family}), so the design objective, bias plus calibration error $\Delta_T+C\sqrt{(1+\chisqZ)/m}$, need not be unimodal along a nested family, and the right way to search is over \emph{informed} candidates: with a full-covariate predictor, directions along which the residual law varies (Section~\ref{sec:choose}); with a sketched score, sufficient-dimension-reduction directions \cite{li1991,cook1998} from the training fold, including the conditional-covariance criterion of \cite{stojanov2019} or the mutual-information forward selection of \cite{polo2022}; and in either case shift directions from unlabelled target covariates. Section~\ref{sec:choose} is honest about the limits: the available confidence bound for the natural source-data certificate of a sketch depends on the full divergence, like the calibration-error bound of full weighting on the same sample (Proposition~\ref{prop:certify}), so in this paper the choice rests either on a structural assumption or on a few target labels (Proposition~\ref{prop:select}), and the scale direction that the residual sketch must learn is estimated far less accurately than a mean direction at the same sample size (E6--E7). Both predecessors already report that reduced-space weighting outperforms full-space weighting on benchmark data, which is the risk-minimisation version of the stability gain in Tables~\ref{tab:hetero}--\ref{tab:est}.

\paragraph{Open problems.}
First, Theorem~\ref{thm:lb} is a lower bound for fixed-score threshold calibration, and Remark~\ref{rem:nonthreshold} shows that non-threshold sets can escape it; a lower bound over all set-valued rules would have to charge for score learning in the sketched space as well. Second, conditional coverage over a sketched function class in the sense of \cite{gibbs2023}: the $K$-atomic lower bound of \cite{isaev2026} suggests that asking for coverage over a $k$-dimensional class costs $\sqrt{(1+\chisqZ)k/m}$, which would give a principled dimension choice for approximate conditional coverage under shift. Third, the interaction bound of Proposition~\ref{prop:cov} is a small-shift bound: its factor $\sqrt{\chi^2(Q_{X|Z}\|P_{X|Z})}$ grows with the discarded shift even when the maximal correlation is tiny, whereas the exact Gaussian leakage stays bounded by one; a version with a bounded dependence measure in place of the fibre standard deviation would be the right general statement. Fourth, the certification cost of Proposition~\ref{prop:certify} is a statement about the natural estimator; a minimax lower bound for estimating the leakage from source data, which we expect to hold at the same $\sqrt{(\chisq-\chisqZ)/n_s}$ scale, would make the obstruction definitive. Fifth, the clipped-weight bound of Theorem~\ref{thm:lognormal} inherits the looseness of the uniform Bernstein step; a training-conditional analysis of the clipped lognormal rule in the style of \cite{pournaderi2026} would bring the constants down to the point where the bound is informative at the calibration sizes used in practice.

\section*{Acknowledgments}
Claude Fable 5.1 (Anthropic) and GPT~6 Astra were used in the preparation of this manuscript; the author thanks Javad Alipanah for providing access to them. The author retains full responsibility for the mathematical content and conclusions of the paper.

\appendix

\section{Proofs of the background results}\label{app:background}

\subsection{Proof of Theorem~\ref{thm:wcp}}\label{app:wcp}

Write $(X_{m+1},Y_{m+1})=(X,Y)$ for the test pair and $V_i=S(X_i,Y_i)$, $i\le m+1$. Under Assumption~\ref{ass:L} the joint law of $(X_i,Y_i)_{i\le m+1}$ has density $w(x_{m+1})$ with respect to $P^{\otimes(m+1)}$. Let $\mathcal F$ be the $\sigma$-field generated by the unordered multiset $\{(X_i,Y_i)\}_{i\le m+1}$. For a realisation $\{(x_i,y_i)\}$ of the multiset, the conditional probability that the test pair is the one labelled $j$ is, by the density formula and symmetry of $P^{\otimes(m+1)}$,
\[
\Prob\{(X_{m+1},Y_{m+1})=(x_j,y_j)\mid\mathcal F\}=\frac{w(x_j)}{\sum_{k=1}^{m+1}w(x_k)}=:p_j
\]
(the denominator is positive a.s.\ because $w(X_{m+1})>0$ $Q$-a.s.; if several pairs coincide, assign the total probability of that value to any one of its labels; what follows is unaffected). This is the \emph{weighted exchangeability} of \cite{tibshirani2019}. Let $v_j=S(x_j,y_j)$ and $G=\sum_{k=1}^{m+1}p_k\delta_{v_k}$, a probability measure on $\R$, and let $q=\Qf(1-\alpha;G)$. Then $G\{v\le q\}\ge1-\alpha$ by definition of the generalised inverse, i.e.
\[
\sum_{j:\,v_j\le q}p_j\ \ge\ 1-\alpha .
\]
Now fix the label $j$ of the test pair. Algorithm~\ref{alg:wcp} with test covariate $x_j$ computes $\widehat Q^w(x_j)=\Qf(1-\alpha;G_j)$, where $G_j$ is $G$ with the atom at $v_j$ moved to $+\infty$. We claim that $v_j\le\widehat Q^w(x_j)$ if and only if $v_j\le q$. Indeed, if $v_j\le q$ then for every $t<v_j$,
\[
G_j(-\infty,t]=G(-\infty,t]<1-\alpha
\]
(the strict inequality because $t<q$), so $\widehat Q^w(x_j)\ge v_j$; conversely, if $v_j>q$ then
\[
G_j(-\infty,q]=G(-\infty,q]\ge1-\alpha,
\]
so $\widehat Q^w(x_j)\le q<v_j$. Therefore
\[
\Prob\{Y\in\Cset^w(X)\mid\mathcal F\}=\Prob\{V_{m+1}\le\widehat Q^w(X_{m+1})\mid\mathcal F\}=\sum_{j:\,v_j\le q}p_j\ \ge\ 1-\alpha,
\]
and taking expectations proves the theorem.\qed

\subsection{Proof of Theorem~\ref{thm:price}}\label{app:price}

\paragraph{(a) Upper bound.}
Fix the training fold. Let $F^Q(t):=Q\{S(X,Y)\le t\}$, which equals $\E_P[w(X)\ind\{S(X,Y)\le t\}]$, be the target score CDF and
\[
\widehat F^w(t):=\frac1m\sum_{i=1}^m w(X_i)\ind\{S_i\le t\},\qquad \widehat F^w(+\infty)=\frac1m\sum_{i=1}^m w(X_i),
\]
the (unnormalised) weighted empirical CDF, so that
\[
\E\widehat F^w(t)=F^Q(t)\quad\text{for every }t\in\Rbar,\qquad \E\widehat F^w(+\infty)=1 .
\]

\emph{Step 1: pointwise Bernstein.} For fixed $t\in\Rbar$ the summands $w(X_i)\ind\{S_i\le t\}$ are i.i.d., take values in $[0,w_{\max}]$, and have variance at most $\E_P[w^2]=1+\chisq$. Bernstein's inequality gives, for $\varepsilon>0$,
\[
\Prob\bigl\{|\widehat F^w(t)-F^Q(t)|>\varepsilon\bigr\}\le2\exp\Bigl(-\frac{m\varepsilon^2}{2(1+\chisq)+\frac23w_{\max}\varepsilon}\Bigr),
\]
so with probability at least $1-\delta'$,
\begin{gather*}
|\widehat F^w(t)-F^Q(t)|\le\varepsilon_m(\delta'),\\
\text{where}\quad\varepsilon_m(\delta'):=\sqrt{2(1+\chisq)\log(2/\delta')/m}+2w_{\max}\log(2/\delta')/(3m).
\end{gather*}
The same holds for the left limits $\widehat F^w(t^-)$, $F^Q(t^-)$ (replace $\le$ by $<$).

\emph{Step 2: uniform bound by discretisation.} Let $N=\lceil\sqrt m\rceil$ and $t_j=\Qf(j/N;F^Q)$ for $j=1,\dots,N-1$, with $t_0=-\infty$, $t_N=+\infty$. By definition of the generalised inverse,
\[
F^Q(t_j^-)\le j/N\le F^Q(t_j).
\]
Apply Step 1 with $\delta'=\delta/(4N)$ to each of the $2N$ quantities $\widehat F^w(t_j),\widehat F^w(t_j^-)$, $j=1,\dots,N$ (with $t_N=+\infty$ this includes the total mass); on the resulting event $\mathcal E$, of probability at least $1-\delta/2$, each deviates from its mean by at most $\varepsilon_m(\delta/4N)$. For $t\in[t_j,t_{j+1})$, monotonicity gives
\[
\widehat F^w(t)-F^Q(t)\le\widehat F^w(t_{j+1}^-)-F^Q(t_j)\le\varepsilon_m(\delta/4N)+\bigl(F^Q(t_{j+1}^-)-F^Q(t_j)\bigr)\le\varepsilon_m(\delta/4N)+\tfrac1N,
\]
and symmetrically
\[
F^Q(t)-\widehat F^w(t)\le\varepsilon_m(\delta/4N)+1/N .
\]
Since $N\le2\sqrt m$, we have $2/\delta'=8N/\delta\le16\sqrt m/\delta$, and on $\mathcal E$,
\[
\sup_{t\in\Rbar}\abs{\widehat F^w(t)-F^Q(t)}\le\eta:=\eta_m(\delta),\qquad \abs{\widehat F^w(+\infty)-1}\le\eta ,
\]
with $\eta_m(\delta)$ as in \eqref{eq:eta}.

\emph{Step 3: localising the threshold.} For a test covariate $x$, Algorithm~\ref{alg:wcp} returns
\[
\widehat Q^w(x)=\Qf\bigl(\beta(x);\widehat F^w\bigr),\qquad \beta(x)=(1-\alpha)\bigl(\widehat F^w(+\infty)+w(x)/m\bigr);
\]
this is the same quantile as in the algorithm after multiplying through by the normalising constant. On $\mathcal E$, $\beta(x)\in[u_{\mathrm{lo}},u_{\mathrm{hi}}]$ with
\[
u_{\mathrm{lo}}:=(1-\alpha)(1-\eta)-\eta,\qquad u_{\mathrm{hi}}:=(1-\alpha)\bigl(1+\eta+w_{\max}/m\bigr)+\eta .
\]
If $q:=\widehat Q^w(x)<\infty$ then $\widehat F^w(q)\ge\beta(x)$, hence
\[
F^Q(q)\ge\beta(x)-\eta\ge u_{\mathrm{lo}};
\]
and for every $t<q$, $\widehat F^w(t)<\beta(x)$, hence
\[
F^Q(t)<\beta(x)+\eta\le u_{\mathrm{hi}},
\]
so $F^Q(q^-)\le u_{\mathrm{hi}}$. If $q=+\infty$ then $F^Q(t)<u_{\mathrm{hi}}$ for all $t$. By continuity of $F^Q$, in all cases
\[
u_{\mathrm{lo}}\le F^Q\bigl(\widehat Q^w(x)\bigr)\le u_{\mathrm{hi}}\qquad\text{for every }x\text{ with }w(x)\le w_{\max},
\]
with the convention $F^Q(+\infty)=1$. (For $x$ with $w(x)=0$ the level is $\beta(x)=(1-\alpha)\widehat F^w(+\infty)$, still in $[u_{\mathrm{lo}},u_{\mathrm{hi}}]$ on $\mathcal E$, so the conclusion holds there too; in any case $\QX\{w=0\}=0$, so only $0<w(x)\le w_{\max}$ matters for the target coverage.)

\emph{Step 4: from threshold level to coverage.} Conditional on the calibration sample, the target coverage is $\int F^Q_{x}(\widehat Q^w(x))\,\mathrm d\QX(x)$, where $F^Q_x(t)=P_{Y|X=x}\{S(x,Y)\le t\}$ is the conditional score CDF, $F^Q=\int F^Q_x\,\mathrm d\QX$, and $F^Q_x(+\infty)=1$. If $u_{\mathrm{lo}}\le0$ the lower bound $\text{coverage}\ge u_{\mathrm{lo}}$ is trivial; otherwise let $q_{\mathrm{lo}}:=\Qf(u_{\mathrm{lo}};F^Q)$. Let $q_{\mathrm{hi}}:=\sup\{t\in\Rbar:F^Q(t)\le u_{\mathrm{hi}}\}$. By Step 3 and the definition of the generalised inverse, $\widehat Q^w(x)\ge q_{\mathrm{lo}}$ for every $x$ (when $u_{\mathrm{lo}}>0$); by Step 3 and the definition of $q_{\mathrm{hi}}$, $\widehat Q^w(x)\le q_{\mathrm{hi}}$ for every $x$. Since each $F^Q_x$ is nondecreasing,
\[
F^Q_x(q_{\mathrm{lo}})\le F^Q_x\bigl(\widehat Q^w(x)\bigr)\le F^Q_x(q_{\mathrm{hi}}),
\]
and integrating against $\QX$,
\[
F^Q(q_{\mathrm{lo}})\ \le\ \Prob_Q\{Y\in\Cset^w(X)\mid\mathcal D^{\mathrm{cal}}_m\}\ \le\ F^Q(q_{\mathrm{hi}}).
\]
By continuity of $F^Q$, $F^Q(q_{\mathrm{lo}})=u_{\mathrm{lo}}$, and $F^Q(q_{\mathrm{hi}})=u_{\mathrm{hi}}$ if $q_{\mathrm{hi}}<\infty$ (continuity from both sides at $q_{\mathrm{hi}}$) while $F^Q(q_{\mathrm{hi}})=1\le u_{\mathrm{hi}}$ if $q_{\mathrm{hi}}=+\infty$. Finally,
\[
u_{\mathrm{hi}}-(1-\alpha)\le2\eta+w_{\max}/m,\qquad (1-\alpha)-u_{\mathrm{lo}}\le2\eta,
\]
which proves (a).\qed

\paragraph{(b) Lower bound.}
\emph{Construction.} Let $\mathcal X=\{x_0,x_D\}$ and
\[
\QX=\delta_{x_0},\qquad \PX=\frac1{1+\kappa}\delta_{x_0}+\frac\kappa{1+\kappa}\delta_{x_D},
\]
so that
\[
w(x_0)=1+\kappa,\qquad w(x_D)=0,\qquad \chisq=\E_P[w^2]-1=(1+\kappa)-1=\kappa .
\]
Let $A$, $B=B^{(\xi)}$ and $S^\star$ be as in the proof of Theorem~\ref{thm:lb} (Appendix~\ref{app:lb}), with $\xi$ to be chosen, and consider the two kernels
\[
P^0_{Y|X}:\ (x_0,x_D)\mapsto(A,A),\qquad P^1_{Y|X}:\ (x_0,x_D)\mapsto(B,A).
\]
Both satisfy Assumption~\ref{ass:A1} by Lemma~\ref{lem:A1check} once $\xi\le\tfrac12\min\{\alpha,1-\alpha\}$.

\emph{Separation.} As in the proof of Theorem~\ref{thm:lb} (Appendix~\ref{app:lb}), thresholds in $\Rbar$ are allowed and may be clipped to $[0,\alpha]$ without increasing either error; the target coverage errors of a clipped threshold written as $\alpha a$ are
\[
D_A(a)=\alpha a,\qquad D_B(a)=(\alpha+\xi)a-\xi,
\]
clipping does not help, and
\[
\inf_{a\in[0,1]}\bigl(|D_A(a)|+|D_B(a)|\bigr)=\frac{\alpha\xi}{\alpha+\xi},
\]
attained at $a=\xi/(\alpha+\xi)$.

\emph{Indistinguishability.} The source laws $P^j=\PX\otimes P^j_{Y|X}$ differ only when $X=x_0$, which has probability $1/(1+\kappa)$, and there the densities differ by the factors $1-\xi/(1-\alpha)$ on $B_0$ and $1+\xi/\alpha$ on $B_1$. Hence
\[
\chi^2(P^1\|P^0)=\frac{1}{1+\kappa}\Bigl[(1-\alpha)\Bigl(\frac{\xi}{1-\alpha}\Bigr)^2+\alpha\Bigl(\frac\xi\alpha\Bigr)^2\Bigr]=\frac{\xi^2}{(1+\kappa)\alpha(1-\alpha)} .
\]
Using $\KL\le\chi^2$, tensorisation of $\KL$, and Pinsker,
\[
\TV\bigl((P^1)^{\otimes m},(P^0)^{\otimes m}\bigr)\le\sqrt{m\,\chi^2(P^1\|P^0)/2}.
\]
Choose
\[
\xi^2=\frac{(1+\kappa)\alpha(1-\alpha)}{2m},
\]
so that $m\chi^2(P^1\|P^0)/2=1/4$ and the total variation is at most $1/2$. The restriction $\xi\le\tfrac12\min\{\alpha,1-\alpha\}$ is equivalent to
\[
m\ \ge\ \frac{2(1+\kappa)\alpha(1-\alpha)}{\min\{\alpha,1-\alpha\}^2}\ =\ m_\star .
\]

\emph{Le Cam.} Let $\nu_j$ be the law of the realised threshold $\widehat q(\mathcal D,U)$ under $P^j$ (the randomisation $U$ has the same law under both). Then
\begin{align*}
\max_j\E_j|D_j(\widehat q)|&\ \ge\ \tfrac12\bigl(\E_0|D_A(\widehat q)|+\E_1|D_B(\widehat q)|\bigr)\ \ge\ \tfrac12\int\bigl(|D_A|+|D_B|\bigr)\,\mathrm d(\nu_0\wedge\nu_1)\\
&\ \ge\ \tfrac12\bigl(1-\TV(\nu_0,\nu_1)\bigr)\inf_a\bigl(|D_A(a)|+|D_B(a)|\bigr),
\end{align*}
and
\[
\TV(\nu_0,\nu_1)\le\TV\bigl((P^0)^{\otimes m},(P^1)^{\otimes m}\bigr)\le\tfrac12
\]
by data processing (Lemma~\ref{lem:dp}; the map $(\mathcal D,U)\mapsto\widehat q$ composed with the independent law of $U$ is a Markov kernel from data to thresholds). Therefore
\[
\max_j\E_j|D_j(\widehat q)|\ \ge\ \frac14\cdot\frac{\alpha\xi}{\alpha+\xi}\ \ge\ \frac\xi8=\frac18\sqrt{\frac{(1+\kappa)\alpha(1-\alpha)}{2m}},
\]
using $\xi\le\alpha$. Since $\QX=\delta_{x_0}$, the target coverage of $\{S^\star\le q\}$ equals its conditional coverage at $x_0$, which is the quantity in the statement.\qed

\section{Sharper calibration bounds and their behaviour under sketching}\label{app:sharper}

\paragraph{The refined price of shift.}
Two refinements of Theorem~\ref{thm:price} are given in \cite{pournaderi2026} and should be kept in mind when the heuristic ``$m/(1+\chisq)$'' is used. First, for the common-threshold weighted quantile on a \emph{given} score, with a bounded ratio and a continuous score distribution, the asymptotic variance of the realised coverage is $\alpha^2K^2+(1-2\alpha)\,\E_Q[w\,\ind\{S>q\}]$ with $K^2=1+\chisq$ and $q$ the target $(1-\alpha)$-quantile of the score \cite[Prop.\ 10]{pournaderi2026}; this equals $\alpha(1-\alpha)K^2$ when $w$ is independent of the score, and for $\alpha\le\tfrac12$ rises to at most $\sigma_\alpha^2=\alpha^2K^2+(1-2\alpha)\Lambda(\alpha)$, with $\Lambda(\alpha)$ the integral of $w$ over the $\alpha$-fraction of the target population where it is largest, the maximum being attained when the upper tail of the score sits on the high-ratio region. Matching the exchangeable variance $\alpha(1-\alpha)$, the effective calibration size is therefore $m/(1+\chisq)$ in the independent case and
\begin{gather*}
m_{\mathrm{eff}}=\frac{m\,\alpha(1-\alpha)}{\sigma_\alpha^2}=\frac{m(1-\alpha)}{\alpha K^2+(1-2\alpha)\,\mathrm{CVaR}_\alpha(w)},\\
\mathrm{CVaR}_\alpha(w)=\Lambda(\alpha)/\alpha\in[K^2,w_{\max}],
\end{gather*}
under adversarial coupling, which is of the order of $m/\mathrm{CVaR}_\alpha(w)$ \cite[Rem.\ 3, Rem.\ 9]{pournaderi2026}; the lower bound of Theorem~\ref{thm:price}(b) certifies the $K^2$ scale, not the coupling term. For the Gaussian mean shift of Section~\ref{sec:gauss} with $\alpha=0.1$, $\mu=1.5$ and $m=500$, $\mathrm{CVaR}_\alpha(w)\approx55.6$ against $K^2\approx9.5$, and $m_{\mathrm{eff}}\approx9.9$ against $m/(1+\chisq)\approx52.7$ in the independent case, so the distinction is not academic. Second, the upper bound \cite[Thm.\ 1]{pournaderi2026} is proved by a single concentration inequality at a fixed population quantile, which removes the uniform-convergence step and the $\log m$ of Theorem~\ref{thm:price}(a); the structure (variance at the $K^2$ scale, supremum only at order $1/m$) is the same. Likewise \cite[Thms.\ 5--6]{isaev2026} remove the logarithm and control the full conditional-coverage profile in $L^p(\QX)$, and \cite[Thm.\ 8]{pournaderi2026} gives a two-point lower bound of the same form as Theorem~\ref{thm:price}(b) with a constant twice as large. Neither refinement changes the conclusion of the paper, because the leading variance proxy does not increase under sketching (Proposition~\ref{prop:proxy}).

\paragraph{Score-specific leakage.}
Sketched calibration with a \emph{common} threshold, $\Cset_\beta(z)=\{y:S_Z(z,y)\le\hat\tau(\beta)\}$ with $\hat\tau(\beta)$ the self-normalised weighted $\beta$-quantile of the sketched scores (no test-point term), is weighted conformal on $\mathcal X$ with the misspecified ratio $\tilde v=w_Z\circ T$ and the score $s(x,y)=S_Z(T(x),y)$. Because the test-point term can only enlarge the set \cite[eq.\ (3)]{pournaderi2026}, the sketched rule of Definition~\ref{def:sketch} contains $\Cset_{1-\alpha}$.

\begin{proposition}[Score-specific leakage]\label{prop:scoreleak}
Under Assumption~\ref{ass:L},
\[
\Delta_S:=\sup_{t}\abs{Q^T\{S_Z\le t\}-\Qt\{S_Z\le t\}}\ \le\ \tfrac12\,\E_P\abs{\E_P[w-w_Z\circ T\mid S_Z]}\ \le\ \Delta_T .
\]
\end{proposition}
\begin{proof}
For a fixed threshold $t$,
\[
Q^T\{S_Z\le t\}-\Qt\{S_Z\le t\}=\E_P\bigl[(w-\tilde v)\,\ind\{s\le t\}\bigr]=\E_P\bigl[\E_P[w-\tilde v\mid s]\,\ind\{s\le t\}\bigr],
\]
and since $\E_P[w-\tilde v]=0$, the supremum over $t$ of the absolute value is at most half the $L^1$ norm of $g:=\E_P[w-\tilde v\mid s]$: for any measurable set $A$,
\[
-\int g^-\le\int_Ag\le\int g^+,\qquad \int g^+=\int g^-=\tfrac12\int|g|,
\]
the equalities because $\int g=0$. The last inequality of the statement holds because $\{S_Z\in B\}$ is a measurable set in $(z,y)$, so the supremum over such sets is at most $\TV(Q^T,\Qt)=\Delta_T$.
\end{proof}

Only the part of the discarded shift $w-\E_P[w\mid Z]$ that is correlated with the sketched score contributes, which is the general form of the interaction phenomenon of Section~\ref{sec:gauss}, where the leakage is driven by $\langle\beta_\perp,a_\perp\rangle$ rather than by $\norm{a_\perp}$.

\paragraph{The variance proxy under sketching.}
Write $B:=\norm{w}_\infty$, $B_Z:=\norm{w_Z}_\infty$, $K_Z^2:=1+\chisqZ$, and $\Lambda(a):=\sup\{\E_Q[w\,h]:0\le h\le1,\ \E_Qh\le a\}$, with $\Lambda_Z$ defined likewise from $(Q_Z,w_Z)$. Suppose $B_Z<\infty$, $\alpha\le\tfrac12$ and $m>\tfrac23\log(1/\delta)$. Then \cite[Thm.\ 1]{pournaderi2026} applied to the surrogate problem $(\Zc,P^T,\Qt,w_Z)$, combined with Proposition~\ref{prop:scoreleak}, gives, with probability $1-\delta$ over the calibration sample, target coverage at least $1-\alpha-\gamma^Z_m(\delta)-\Delta_S$, where $\gamma^Z_m(\delta)$ is the smallest $\gamma>0$ with
\[
\exp\Bigl(-\frac{m\gamma^2}{2\sigma^2_{\alpha,Z}+(2-\tfrac{10}{3}\alpha)B_Z\gamma+\tfrac23\gamma^2}\Bigr)=\delta,
\qquad \sigma^2_{\alpha,Z}:=\alpha^2K_Z^2+(1-2\alpha)\Lambda_Z(\alpha),
\]
the bound (9) of \cite{pournaderi2026} with the lower-order quantile term replaced by its upper bound $B_Z$ and evaluated at $(K_Z^2,\Lambda_Z(\alpha),B_Z)$; the left side decreases to $e^{-3m/2}$ as $\gamma\to\infty$, so the condition $m>\tfrac23\log(1/\delta)$ is exactly what makes a finite root exist, and otherwise one sets $\gamma^Z_m(\delta)=+\infty$ (the statement is then void).

\begin{proposition}[The variance proxy does not increase under sketching]\label{prop:proxy}
Under Assumption~\ref{ass:chi}, $K_Z^2\le K^2$ and $\Lambda_Z(a)\le\Lambda(a)$ for every $a\in(0,1]$; under Assumption~\ref{ass:bdd}, $B_Z\le B$. Consequently, for $\alpha\le\tfrac12$, $\sigma^2_{\alpha,Z}\le\sigma^2_\alpha$, and $\gamma^Z_m(\delta)\le\gamma_m(\delta)$ whenever both are finite.
\end{proposition}
\begin{proof}
$K_Z^2\le K^2$ and $B_Z\le B$ are Theorem~\ref{prop:main}(ii). For the tail average,
\begin{align*}
\Lambda_Z(a)=\sup_{h=h(Z)}\E_{Q_Z}[w_Z h]&=\sup_{h=h(Z)}\E_P\bigl[(\E_P[w\mid Z])^2h\bigr]\\
&\le\ \sup_{h=h(Z)}\E_P[w^2h]=\sup_{h=h(Z)}\E_Q[w\,h]\ \le\ \Lambda(a),
\end{align*}
where the suprema are over $\sigma(Z)$-measurable $h$ with $0\le h\le1$ and $\E_{Q}h\le a$ (note $\E_{Q_Z}h(Z)=\E_Qh$), the inequality is conditional Jensen, and the last step restricts the supremum defining $\Lambda(a)$ to $\sigma(Z)$-measurable $h$. The coefficient $1-2\alpha$ is nonnegative for $\alpha\le\tfrac12$, which gives $\sigma^2_{\alpha,Z}\le\sigma^2_\alpha$; the left side of the defining equation of $\gamma^Z_m$ is then pointwise at most that of $\gamma_m$, so the root is no larger.
\end{proof}

The restriction to $\alpha\le\tfrac12$ is needed because the coefficient $1-2\alpha$ changes sign. The monotonicity is of the proxy, not of every quantity in the sharpest bound: the target-law quantile $\rho_\alpha$ of the ratio that enters the lower-order term of \cite[eq.\ (9)]{pournaderi2026} can increase under sketching (two atoms with $\PX=(\tfrac1{50},\tfrac{49}{50})$, $w=(2,\tfrac{48}{49})$ and a constant sketch give $\rho_{0.1}=\tfrac{48}{49}<1=\rho_{0.1,Z}$), which is why $\gamma^Z_m$ above uses $B_Z$ in its place.

\section{Proof of the lower bound (Theorem~\ref{thm:lb})}\label{app:lb}

\begin{proof}[Proof of Theorem~\ref{thm:lb}]
\emph{Construction.} Let $\mathcal X=\{x_0,x_0',x_D\}$ and $T(x_0)=T(x_0')=z_0\ne z_D=T(x_D)$. (For a linear sketch $\Pi\in\R^{k\times d}$ with $k<d$ take $x_0'=x_0+v$ with $0\ne v\in\ker\Pi$.) For some $p\in(0,1]$ put
\[
\QX=\delta_{x_0},\qquad \PX=\tfrac p2\delta_{x_0}+\tfrac p2\delta_{x_0'}+(1-p)\delta_{x_D}.
\]
Let $A$ be the uniform law on $[-1,1]$ and $S^\star(x,y)=|y|-(1-\alpha)$, so that the score bands
\[
B_0=\{S^\star\le0\}=\{|y|\le1-\alpha\},\qquad B_1=\{0<S^\star\le\alpha\}=\{1-\alpha<|y|\le1\}
\]
carry mass $1-\alpha$ and $\alpha$ under $A$. Let $B=B^{(\xi)}$ be the law on $[-1,1]$ with density
\[
b(y)=\tfrac12\Bigl(1-\frac{\xi}{1-\alpha}\Bigr)\ \text{on } |y|\le1-\alpha,\qquad
b(y)=\tfrac12\Bigl(1+\frac{\xi}{\alpha}\Bigr)\ \text{on } 1-\alpha<|y|\le1,
\]
i.e.\ $A$ with mass $\xi$ moved from $B_0$ to $B_1$, uniformly within each band, so that
\[
B(B_0)=1-\alpha-\xi,\qquad B(B_1)=\alpha+\xi,\qquad \TV(A,B)=\xi .
\]
Define
\[
P^0_{Y|X}:\ (x_0,x_0',x_D)\mapsto(A,B,A),\qquad P^1_{Y|X}:\ (x_0,x_0',x_D)\mapsto(B,A,A).
\]
Lemma~\ref{lem:A1check} verifies Assumption~\ref{ass:A1} for both kernels with the stated constants.

\emph{(a)} Under either kernel the sketched source law puts mass $p$ on $z_0$ with response law $\tfrac12(A+B)$ and mass $1-p$ on $z_D$ with law $A$; $Q_Z=\delta_{z_0}$ and $w_Z(z_0)=1/p$. The target conditional at $z_0$ is $A$ under $P^0$ and $B$ under $P^1$, while $P_{Y|Z=z_0}=\tfrac12(A+B)$ in both cases, so
\[
\Delta_T=\TV\bigl(A,\tfrac12(A+B)\bigr)=\tfrac12\TV(A,B)=\xi/2 .
\]

\emph{(b)} Since the sketched data have the same law under both kernels and $\widehat q$ is a fixed function of $(\mathcal D,U)$, the law of $\widehat q$ is the same under $P^0$ and $P^1$. Under $P^0$ the target coverage of $\{S^\star\le q\}$ is $A\{|y|\le1-\alpha+q\}$ and under $P^1$ it is $B\{|y|\le1-\alpha+q\}$. Thresholds in $\Rbar$ are allowed, including $q=+\infty$ (the set $\R$) and $q=-\infty$ (the empty set). Clipping $q$ to $[0,\alpha]$ does not increase the absolute coverage error under either kernel: for $q\in[-(1-\alpha),0)$ the errors are
\[
D_A(q)=q,\qquad D_B(q)=q\Bigl(1-\frac{\xi}{1-\alpha}\Bigr)-\xi,
\]
so $|D_A(q)|\ge0=|D_A(0)|$ and $|D_B(q)|\ge\xi=|D_B(0)|$; for $q<-(1-\alpha)$, including $q=-\infty$, both coverages are $0$ and the errors are larger still; for $q>\alpha$, including $q=+\infty$, both coverages equal $1$, as at $q=\alpha$. Write the clipped threshold as $\widetilde q=\alpha\,a(\mathcal D,U)$ with $a\in[0,1]$, and $\bar a:=\E[a]$. Then
\[
D_A(a)=\alpha a\ge0,\qquad D_B(a)=(\alpha+\xi)a-\xi ,
\]
so, by Jensen's inequality,
\[
\E|D_A(a)|=\alpha\bar a,\qquad \E|D_B(a)|\ge|\E D_B(a)|=|(\alpha+\xi)\bar a-\xi| .
\]
Hence
\begin{align*}
\max_j\E\abs{D_j(\widehat q)}\ &\ge\ \max_j\E\abs{D_j(\widetilde q)}\ \ge\ \max\bigl\{\alpha\bar a,\ \abs{(\alpha+\xi)\bar a-\xi}\bigr\}\\
&\ge\ \min_{t\in[0,1]}\max\bigl\{\alpha t,\ \abs{(\alpha+\xi)t-\xi}\bigr\}=\frac{\alpha\xi}{2\alpha+\xi},
\end{align*}
the minimum being attained at $t=\xi/(2\alpha+\xi)$, where the two branches are equal. The deterministic rule $\widehat q\equiv\alpha\xi/(2\alpha+\xi)$ has $|D_A|=|D_B|=\alpha\xi/(2\alpha+\xi)$, so the bound is attained.
\end{proof}

\begin{lemma}[Regularity of the lower-bound kernels]\label{lem:A1check}
For $0<\xi\le\tfrac12\min\{\alpha,1-\alpha\}$, both $A$ and $B^{(\xi)}$ satisfy the conditions of Assumption~\ref{ass:A1} with $\mu_{\mathrm{low}}=\tfrac14$, $\mu_{\mathrm{up}}=\tfrac34$, $r_0=\tfrac{2\alpha}3$.
\end{lemma}
\begin{proof}
$A$ has density $\tfrac12$ on $[-1,1]$. Under the restriction on $\xi$, the two values of the density $b$ of $B^{(\xi)}$ satisfy
\[
\tfrac12\Bigl(1-\frac{\xi}{1-\alpha}\Bigr)\ge\tfrac14,\qquad \tfrac12\Bigl(1+\frac\xi\alpha\Bigr)\le\tfrac34,
\]
so $\tfrac14\le b\le\tfrac34$ on $[-1,1]$ and the upper mass bound holds with $\mu_{\mathrm{up}}=\tfrac34$. For the lower growth, every quantile $q$ of either law at level $\tau\in\{\alpha/2,1-\alpha/2\}$ lies in $(-1,1)$ with
\[
\min\{q+1,1-q\}\ \ge\ \frac{\alpha/2}{\mu_{\mathrm{up}}}=\frac{2\alpha}{3},
\]
because the mass beyond $q$ on the near side is at least $\alpha/2$ and the density is at most $\tfrac34$. Hence for $h\le r_0=2\alpha/3$ the intervals $[q-h,q]$ and $[q,q+h]$ lie in $[-1,1]$ and carry mass at least $\tfrac14h$.
\end{proof}

\section{Proofs of Propositions~\ref{prop:lip} and \ref{prop:fullgauss}}\label{app:lip}

\begin{proof}[Proof of Proposition~\ref{prop:lip}]
From the proof of Proposition~\ref{prop:gauss}, conditional on $Z=z$ the index
\[
t=\langle\beta,X\rangle=\langle\Pi\beta,z\rangle+\langle\beta_\perp,X_\perp\rangle
\]
has law $\Lambda_P:=N(c,\norm{\beta_\perp}^2)$ under $P$ and $N(c+m,\norm{\beta_\perp}^2)$ under $Q$, with $c=\langle\Pi\beta,z\rangle$ and $m=\mu\langle\beta_\perp,a_\perp\rangle$. Hence
\[
P_{Y|Z=z}=\int K(t,\cdot)\,\mathrm d\Lambda_P(t),\qquad Q_{Y|Z=z}=\int K(t+m,\cdot)\,\mathrm d\Lambda_P(t),
\]
and for every measurable $A$,
\[
|Q_{Y|Z=z}(A)-P_{Y|Z=z}(A)|\le\int|K(t+m,A)-K(t,A)|\,\mathrm d\Lambda_P(t)\le L_K|m| .
\]
For $\hat\beta\in\operatorname{row}(\Pi)$, as in Corollary~\ref{cor:informed}(b),
\[
|\langle\beta_\perp,a_\perp\rangle|\le\norm{\beta_\perp}\le\norm\beta\sin\angle(\beta,\hat\beta).
\]
For the Gaussian kernel, Lemma~\ref{lem:gausstv} gives
\[
\TV\bigl(N(t,\sigma^2),N(t',\sigma^2)\bigr)=2\Phi\Bigl(\frac{|t-t'|}{2\sigma}\Bigr)-1\le\frac{|t-t'|}{\sqrt{2\pi}\,\sigma}.
\]
For the heteroscedastic kernel, write $f_{\mu,\sigma}$ for the $N(\mu,\sigma^2)$ density and $U\sim N(0,1)$. Then
\[
\partial_\mu f_{\mu,\sigma}=f_{\mu,\sigma}\,\frac{y-\mu}{\sigma^2},\qquad
\tfrac12\norm{\partial_\mu f_{\mu,\sigma}}_{L^1}=\frac{\E|U|}{2\sigma}=\frac{1}{\sqrt{2\pi}\,\sigma},
\]
and
\[
\partial_\sigma f_{\mu,\sigma}=f_{\mu,\sigma}\,\frac{(y-\mu)^2/\sigma^2-1}{\sigma},\qquad
\tfrac12\norm{\partial_\sigma f_{\mu,\sigma}}_{L^1}=\frac{\E|U^2-1|}{2\sigma}=\frac{2\varphi(1)}{\sigma},
\]
using $\E|U^2-1|=4\varphi(1)$. Moving from $(t,\sigma(t))$ to $(t',\sigma(t'))$ along the path $s\mapsto(s,\sigma(s))$, $s$ between $t$ and $t'$, and integrating the $L^1$ norm of the derivative of $f$ along the path,
\begin{multline*}
\TV\bigl(N(t,\sigma(t)^2),N(t',\sigma(t')^2)\bigr)\le\int\Bigl(\frac{1}{\sqrt{2\pi}\,\sigma(s)}+\frac{2\varphi(1)}{\sigma(s)}\,|\sigma'(s)|\Bigr)\mathrm ds\\
\le\frac{1/\sqrt{2\pi}+2\varphi(1)L}{\sigma_{\min}}\,|t-t'| .
\end{multline*}
(A bound with an additional $\sigma_{\max}/\sigma_{\min}^2$ factor follows alternatively from \cite[Thm.\ 1.3]{devroye2018}; the path argument removes the need for an upper bound on $\sigma$.)
\end{proof}

\begin{proof}[Proof of Proposition~\ref{prop:fullgauss}]
Conditional on $Z=z$, write $X=\Pi^{\!\top}z+X_\perp$ as in the proof of Proposition~\ref{prop:gauss}, with $X_\perp\sim N(0,I-\Pi^{\!\top}\Pi)$ under $P$ and $X_\perp\sim N(\mu a_\perp,I-\Pi^{\!\top}\Pi)$ under $Q$. The residual is
\[
R:=Y-\theta^{\!\top}X-\theta_0=\langle\gamma,X\rangle-\theta_0+\sigma\xi=\langle\Pi\gamma,z\rangle-\theta_0+\langle\gamma_\perp,X_\perp\rangle+\sigma\xi,
\]
so that, with $c_z=\langle\Pi\gamma,z\rangle-\theta_0$,
\[
R\mid Z=z\ \sim\ N(c_z,s_\theta^2)\ \text{ under }P,\qquad
R\mid Z=z\ \sim\ N\bigl(c_z+\mu\langle\gamma_\perp,a_\perp\rangle,s_\theta^2\bigr)\ \text{ under }Q.
\]
The score $S=|R|$ is a function of $R$, so by Lemma~\ref{lem:dp} and Lemma~\ref{lem:gausstv}, for every $z$,
\[
\delta_{\Pi,S}(z)\le2\Phi\Bigl(\frac{\mu|\langle\gamma_\perp,a_\perp\rangle|}{2s_\theta}\Bigr)-1;
\]
the two further inequalities are $2\Phi(t/2)-1\le t/\sqrt{2\pi}$, Cauchy--Schwarz, and $\norm{\gamma_\perp}\le\norm\gamma$, $s_\theta\ge\sigma$. In the heteroscedastic case
\[
R=\langle\gamma,X\rangle-\theta_0+\sigma(\langle\beta,X\rangle)\,\xi,
\]
and conditional on $Z=z$ the pair $(t_1,t_2):=(\langle\gamma,X\rangle,\langle\beta,X\rangle)$ is Gaussian with the same covariance under $P$ and $Q$ and means differing by
\[
(m_1,m_2):=\mu\bigl(\langle\gamma_\perp,a_\perp\rangle,\langle\beta_\perp,a_\perp\rangle\bigr).
\]
Hence, with $K(t_1,t_2)=N(t_1-\theta_0,\sigma(t_2)^2)$ and $\Lambda$ the $P$-law of $(t_1,t_2)$ given $z$,
\[
P_{R|Z=z}=\int K(t_1,t_2)\,\mathrm d\Lambda(t),\qquad Q_{R|Z=z}=\int K(t_1+m_1,t_2+m_2)\,\mathrm d\Lambda(t),
\]
so
\[
\delta_{\Pi,S}(z)\le\sup_{t}\TV\bigl(K(t_1+m_1,t_2+m_2),K(t_1,t_2)\bigr).
\]
Along the path $u\mapsto(t_1+um_1,t_2+um_2)$, $u\in[0,1]$, half the $L^1$ norm of the derivative of the Gaussian density is at most
\[
\frac{|m_1|}{\sqrt{2\pi}\,\sigma(\cdot)}+\frac{2\varphi(1)}{\sigma(\cdot)}\,|\sigma'(\cdot)|\,|m_2|
\]
by the two derivative computations in the proof of Proposition~\ref{prop:lip} (which bound half the $L^1$ norms of $\partial_\mu f$ and $\partial_\sigma f$), and integrating over $u$ with $\sigma\ge\sigma_{\min}$, $|\sigma'|\le L$ bounds the total variation as claimed.
\end{proof}

\section{Proof of Theorem~\ref{thm:lognormal}}\label{app:lognormal}

\paragraph{Closed forms.} Under $P$, $G:=\langle a,X\rangle\sim N(0,1)$ and $w=e^{\mu G-\mu^2/2}$, so $w>M$ iff $G>c$. Completing the square,
\[
e^{\mu g-\mu^2/2}\varphi(g)=\varphi(g-\mu),\qquad e^{2\mu g-\mu^2}\varphi(g)=e^{\mu^2}\varphi(g-2\mu);
\]
hence
\begin{gather*}
\E[w\ind\{w>M\}]=\bar\Phi(c-\mu),\qquad
\E[(w-M)_+]=\bar\Phi(c-\mu)-M\bar\Phi(c)=e_M,\\
\E[w^2\ind\{w\le M\}]=e^{\mu^2}\Phi(c-2\mu),\qquad
\E[w^2\ind\{w>M\}]=e^{\mu^2}\bar\Phi(c-2\mu),\\
\E[v_M^2]=e^{\mu^2}\Phi(c-2\mu)+M^2\bar\Phi(c)=c_M^2K_M^2 .
\end{gather*}
Also, with $c'=(\log M')/\mu+\mu/2$,
\[
\QX\{w>M'\}=\E_P[w\ind\{w>M'\}]=\bar\Phi(c'-\mu).
\]

\paragraph{Clipped rule.} By the proof of Proposition~\ref{prop:clip}, the rule run with $v_M$ coincides with the rule run with the density ratio $r_M=v_M/c_M$, and $\tfrac12\E_P|w-r_M|\le e_M$. The ratio $r_M$ is bounded by $B_M$ and has $\E_P[r_M^2]=K_M^2$, and the score is continuous under $r_M\PX\otimes P_{Y|X}$ because $r_M\PX$ and $\QX$ are mutually absolutely continuous. Theorem~\ref{thm:price}(a) applied to the surrogate problem $(\PX,\,r_M\PX,\,r_M)$ gives, with probability at least $1-\delta$,
\[
\bigl|\Prob_{r_M\PX\otimes P_{Y|X}}\{Y\in\Cset^{v_M}(X)\mid\mathcal D\}-(1-\alpha)\bigr|\le2\eta_m(\delta;K_M^2,B_M)+B_M/m,
\]
and the total-variation step of Proposition~\ref{prop:clip}, which holds conditionally on $\mathcal D$ because the coverage event is then a fixed set, adds $e_M$. For the sketched rule,
\[
w_Z=e^{\mu\langle\Pi a,Z\rangle-\mu^2\norm{a_\parallel}^2/2},\qquad \langle\Pi a,Z\rangle\sim N(0,\norm{a_\parallel}^2)\ \text{ under }P_Z,
\]
so $w_Z$ is lognormal with parameter $\mu\norm{a_\parallel}$, and the argument applies in the sketched space with the additional total-variation step $\Delta_\Pi$ of Theorem~\ref{prop:main}.

\paragraph{Unclipped rule.} Write $w=v_M+(w-M)_+$ and, in the notation of Appendix~\ref{app:price},
\[
\widehat F^w(t)=\widehat F^{v_M}(t)+\widehat F^+(t),\qquad
\widehat F^+(t):=\frac1m\sum_i(w(X_i)-M)_+\ind\{S_i\le t\}\in\bigl[0,\widehat F^+(+\infty)\bigr],
\]
and $F^{v_M}(t):=\E[v_M(X)\ind\{S\le t\}]$, so that $0\le F^Q(t)-F^{v_M}(t)\le e_M$. Steps 1--2 of the proof of Theorem~\ref{thm:price}(a), applied to the summands $v_M(X_i)\ind\{S_i\le t\}\in[0,M]$ with variance at most $\E[v_M^2]=c_M^2K_M^2$, give, with probability at least $1-\delta$,
\[
\sup_t\bigl|\widehat F^{v_M}(t)-F^{v_M}(t)\bigr|\le\eta_m(\delta;c_M^2K_M^2,M).
\]
By Chebyshev's inequality, with probability at least $1-\delta'$,
\[
\widehat F^+(+\infty)\le e_M+\sqrt{\frac{\operatorname{Var}((w-M)_+)}{m\delta'}}\le e_M+\sqrt{\frac{e^{\mu^2}\bar\Phi(c-2\mu)}{m\delta'}}.
\]
On the intersection of the two events, including at $t=+\infty$,
\[
\sup_{t\in\Rbar}\bigl|\widehat F^w(t)-F^Q(t)\bigr|\le\eta':=\eta_m(\delta;c_M^2K_M^2,M)+2e_M+\sqrt{\frac{e^{\mu^2}\bar\Phi(c-2\mu)}{m\delta'}}.
\]
Steps 3--4 of that proof now go through with $\eta'$ in place of $\eta$ for every test covariate $x$ with $w(x)\le M'$, giving
\[
u_{\mathrm{lo}}:=(1-\alpha)(1-\eta')-\eta'\ \le\ F^Q\bigl(\widehat Q^w(x)\bigr)\ \le\ u_{\mathrm{hi}}:=(1-\alpha)(1+\eta'+M'/m)+\eta'
\]
and hence
\[
F^Q_x(q_{\mathrm{lo}})\le F^Q_x\bigl(\widehat Q^w(x)\bigr)\le F^Q_x(q_{\mathrm{hi}})
\]
for such $x$. Integrating over $\QX$ and bounding the contribution of $\{w>M'\}$ by its $\QX$-mass,
\[
u_{\mathrm{lo}}-\QX\{w>M'\}\ \le\ \Prob_Q\{Y\in\Cset^w(X)\mid\mathcal D\}\ \le\ u_{\mathrm{hi}}+\QX\{w>M'\},
\]
and
\begin{gather*}
(1-\alpha)-u_{\mathrm{lo}}\le2\eta',\qquad u_{\mathrm{hi}}-(1-\alpha)\le2\eta'+M'/m,\\
\QX\{w>M'\}=\bar\Phi\bigl((\log M')/\mu-\mu/2\bigr),
\end{gather*}
which is the claim. The two events depend on the fixed $M$ but not on $M'$, so the conclusion holds simultaneously for all $M'>0$.\qed

\section{Auxiliary lemmas}\label{app:aux}

\begin{lemma}[Total variation of mixtures with a common marginal]\label{lem:tvmix}
Let $\mu=\nu\otimes K$ and $\mu'=\nu\otimes K'$ for a probability $\nu$ on $\Zc$ and kernels $K,K'$ from $\Zc$ to $\R$. Then $\TV(\mu,\mu')=\int\TV(K(z,\cdot),K'(z,\cdot))\,\mathrm d\nu(z)$.
\end{lemma}
\begin{proof}
For each $z$ let $\rho_z:=\tfrac12(K(z,\cdot)+K'(z,\cdot))$, which dominates both $K(z,\cdot)$ and $K'(z,\cdot)$, and let $k(z,\cdot),k'(z,\cdot)$ be the densities; $z\mapsto\rho_z$ is a kernel and the densities can be chosen jointly measurable. Then $\mu$ and $\mu'$ have densities $k$ and $k'$ with respect to $\nu\otimes\rho$, so by Scheff\'e
\[
\TV(\mu,\mu')=\tfrac12\iint|k-k'|\,\mathrm d\rho_z\,\mathrm d\nu(z)=\int\Bigl(\tfrac12\int|k(z,\cdot)-k'(z,\cdot)|\,\mathrm d\rho_z\Bigr)\mathrm d\nu(z),
\]
and the inner integral is $\TV(K(z,\cdot),K'(z,\cdot))$.
\end{proof}

\begin{lemma}[Data processing for total variation]\label{lem:dp}
Let $\mu,\mu'$ be probability measures on a space $\mathcal U$ and $K$ a Markov kernel from $\mathcal U$ to $\mathcal V$. Then $\TV(\mu K,\mu'K)\le\TV(\mu,\mu')$, where $(\mu K)(B)=\int K(u,B)\,\mathrm d\mu(u)$.
\end{lemma}
\begin{proof}
For measurable $B\subseteq\mathcal V$, $u\mapsto K(u,B)$ takes values in $[0,1]$, so
\[
|(\mu K)(B)-(\mu'K)(B)|=\Bigl|\int K(\cdot,B)\,\mathrm d(\mu-\mu')\Bigr|\le\sup_{0\le g\le1}\Bigl|\int g\,\mathrm d(\mu-\mu')\Bigr|=\TV(\mu,\mu').
\]
\end{proof}

\begin{lemma}[Chain rule for a deterministic map]\label{lem:chain}
Let $\QX\ll\PX$ on $\mathcal X$, $T:\mathcal X\to\Zc$ measurable, $P_Z=T_\#\PX$, $Q_Z=T_\#\QX$, with fibre laws $P_{X|Z},Q_{X|Z}$. Then $\KL(\QX\|\PX)=\KL(Q_Z\|P_Z)+\E_{Q_Z}\KL(Q_{X|Z}\|P_{X|Z})$ (both sides possibly $+\infty$).
\end{lemma}
\begin{proof}
Let $w=\mathrm d\QX/\mathrm d\PX$ and $w_Z=\mathrm dQ_Z/\mathrm dP_Z$. By Theorem~\ref{prop:main}(i), $w_Z\circ T=\E_P[w\mid T]$, and $w_Z>0$ $Q_Z$-a.s. For bounded measurable $h$ and $B\subseteq\Zc$,
\begin{align*}
\E_Q[h(X)\ind_B(Z)]&=\E_P[w\,h(X)\ind_B(Z)]=\int_B\E_P[w\,h(X)\mid Z=z]\,\mathrm dP_Z(z)\\
&=\int_B\E_P\Bigl[\frac{w}{w_Z(z)}\,h(X)\,\Big|\,Z=z\Bigr]\,\mathrm dQ_Z(z),
\end{align*}
so for $Q_Z$-a.e.\ $z$ the fibre ratio is
\[
\frac{\mathrm dQ_{X|Z=z}}{\mathrm dP_{X|Z=z}}=\frac{w}{w_Z(z)}\qquad\text{on the fibre}.
\]
Hence
\[
\log w=\log w_Z\circ T+\log\bigl(w/w_Z\circ T\bigr),
\]
and integrating against $\QX$ gives the identity, the first term being $\KL(Q_Z\|P_Z)$ and the second $\int\KL(Q_{X|Z=z}\|P_{X|Z=z})\,\mathrm dQ_Z(z)$.
\end{proof}

\begin{lemma}[Gaussian total variation]\label{lem:gausstv}
For $s>0$ and $\theta\in\R$, $\TV\bigl(N(\theta,s^2),N(0,s^2)\bigr)=2\Phi\bigl(|\theta|/(2s)\bigr)-1$.
\end{lemma}
\begin{proof}
By scaling take $s=1$, $\theta\ge0$. The densities cross at $y=\theta/2$ and $N(\theta,1)$ dominates on $\{y>\theta/2\}$, so
\[
\TV=\Prob\{N(\theta,1)>\theta/2\}-\Prob\{N(0,1)>\theta/2\}=\Phi(\theta/2)-\bigl(1-\Phi(\theta/2)\bigr).
\]
\end{proof}

\begin{lemma}[Haar-frame moments]\label{lem:haar}
Let $\Pi\in\R^{k\times d}$ have orthonormal rows and be distributed as the first $k$ rows of a Haar-distributed orthogonal matrix. Then for every fixed $v\in\R^d$, $\E\norm{\Pi^{\!\top}\Pi v}^2=\tfrac kd\norm v^2$.
\end{lemma}
\begin{proof}
With $\pi_i$ the rows of $\Pi$,
\[
\norm{\Pi^{\!\top}\Pi v}^2=\sum_{i=1}^k\langle\pi_i,v\rangle^2 .
\]
Each $\pi_i$ is uniform on the unit sphere, so
\[
\E\langle\pi_i,v\rangle^2=\norm v^2/d
\]
by symmetry (the $d$ coordinates of a uniform unit vector have equal second moments summing to $1$).
\end{proof}

\begin{lemma}[Johnson--Lindenstrauss; for context]\label{lem:jl}
Let $G\in\R^{k\times d}$ have i.i.d.\ $N(0,1/k)$ entries, let $u_1,\dots,u_n\in\R^d$, $\varepsilon\in(0,1)$ and $\delta\in(0,1)$. If $k\ge12\,\varepsilon^{-2}\log\bigl(n(n-1)/\delta\bigr)$ then, with probability at least $1-\delta$, $(1-\varepsilon)\norm{u_i-u_j}^2\le\norm{G(u_i-u_j)}^2\le(1+\varepsilon)\norm{u_i-u_j}^2$ for all $i\ne j$.
\end{lemma}
\begin{proof}
For a fixed $v\ne0$, $k\norm{Gv}^2/\norm v^2\sim\chi^2_k$, and the Dasgupta--Gupta tail bounds \cite{dasgupta2003} give, for $\varepsilon\le1$,
\begin{align*}
\Prob\bigl\{\norm{Gv}^2\le(1-\varepsilon)\norm v^2\bigr\}&\le\exp\bigl(\tfrac k2(\varepsilon+\log(1-\varepsilon))\bigr)\le e^{-k\varepsilon^2/4},\\
\Prob\bigl\{\norm{Gv}^2\ge(1+\varepsilon)\norm v^2\bigr\}&\le\exp\bigl(\tfrac k2(\log(1+\varepsilon)-\varepsilon)\bigr)\le e^{-k(\varepsilon^2/2-\varepsilon^3/3)/2}\le e^{-k\varepsilon^2/12}.
\end{align*}
A union bound over the $n(n-1)/2$ pairs gives failure probability at most $n(n-1)e^{-k\varepsilon^2/12}\le\delta$.
\end{proof}
The lemma is about the $\binom n2$ difference vectors of a given point set; it says nothing about a fixed direction $v$ not determined by that set, for which $\norm{Gv}^2$ concentrates around $\norm v^2$ but the retained energy after orthonormalisation, $\norm{\Pi^{\!\top}\Pi v}^2$ with $\Pi$ the orthonormalised rows of $G$, is only $\tfrac kd\norm v^2$ on average by Lemma~\ref{lem:haar}. This is the distinction behind Section~\ref{sec:gauss}.

\section{Experimental details}\label{app:exp}

\paragraph{Data.} Covariates are $X\sim N(0,I_{40})$ (source) and $X\sim N(\mu a,I_{40})$ (target). The response is $Y=\langle\beta,X\rangle+\sigma(X)\xi$ with $\xi\sim N(0,1)$, $\beta$ a fixed random vector scaled to $\norm\beta=3$ (seed $0$), and $\sigma(X)=0.5+0.5\,(\langle\beta,X\rangle/\norm\beta)_+$ (E1, E2) or $\sigma(X)\equiv1$ (E3, E4, E5). The shift direction $a$ is $0.7\,\beta/\norm\beta+\sqrt{0.51}\,g$ with $g$ a random unit vector orthogonal to $\beta$ (E1, E2, with different seeds for $g$), a random unit vector drawn with a seed independent of $\beta$ (E3, E5), or $\beta/\norm\beta$ (E4).

\paragraph{Sketches.} Random sketches are Gaussian matrices orthonormalised by QR. The informed sketches of E1, E2 and E5 are $\Pi_1=\hat\beta^{\!\top}/\norm{\hat\beta}$ and $\Pi_8$, which stacks $\hat\beta/\norm{\hat\beta}$ above $7$ Gaussian rows and orthonormalises; $\hat\beta$ is the ordinary-least-squares coefficient on the training fold, and both sketches are redrawn in every replication together with the training fold, so they are independent of the calibration fold as required. The shift-aware sketch of E3 appends $a_\perp/\norm{a_\perp}$ to $\Pi$ and re-orthonormalises.

\paragraph{Weights.} Known weights are the closed forms of Proposition~\ref{prop:gauss}: $w(x)=\exp(\mu\langle a,x\rangle-\mu^2/2)$ for full weighting and $w_Z(z)=\exp(\mu\langle\Pi a,z\rangle-\mu^2\norm{\Pi a}^2/2)$ for the sketched rules. The clipped rule uses $\min\{w(x),20\}$ for both calibration and test points. In E2, weights are $\exp(\hat f(z))$ where $\hat f$ is the fitted logit of an (essentially unregularised, $C=10^4$) logistic regression of the label ``target'' on $z$, trained on $2000$ unlabeled source and $2000$ unlabeled target covariates drawn independently of all other data; with equal class sizes the odds estimate $w_Z$ up to a constant, which the algorithm ignores. The reported estimation error is $\tfrac12\E_{P_Z}|\bar w_Z-w_Z|$ with $\bar w_Z=\hat w_Z/\E_{P_Z}\hat w_Z$, computed exactly as $2\Phi(\norm{\hat b-\mu\Pi a}/2)-1$ from the fitted coefficient vector $\hat b$ (Section~\ref{sec:exp}); no Monte Carlo evaluation is involved.

\paragraph{Procedure.} The training fold fits $\hat\mu(z)=\hat\theta^{\!\top}z+\hat\theta_0$ by least squares on $(\Pi X_i,Y_i)$; the score is $|y-\hat\mu(z)|$. The threshold is computed by sorting calibration scores and returning the smallest score at which the cumulative weight reaches $(1-\alpha)\bigl(\sum_jw(Z_j)+w(z)\bigr)$, or $+\infty$ if none does; this is Algorithm~\ref{alg:wcp}.

\paragraph{Reporting.} Coverage is the fraction of $2\cdot10^4$ independent target test points covered in a replication (an infinite threshold counts as covering), averaged over replications ($200$ for E1, E2 and E5, $300$ for E3, $150$ for E4); SD is the standard deviation of this fraction across replications. Length is the within-replication median over test points of $2\widehat Q^w(z)$, then the median across replications; infinite thresholds are included in the within-replication median, which is finite whenever fewer than half the thresholds are infinite and is reported as $\infty$ otherwise (Table~\ref{tab:nonlinear}). Infinite sets is the mean percentage of infinite thresholds. Leakage and $\chisqZ$ values are the closed forms of Proposition~\ref{prop:gauss} evaluated at the realised $\Pi$, averaged over replications where $\Pi$ varies. Figure~\ref{fig:hetero} is E1 on the grid $\mu\in\{0.5,1,\dots,3\}$ with the random-number generator re-seeded at each $\mu$, so its values at $\mu\in\{1,1.5,2.5\}$ coincide with Table~\ref{tab:hetero}. Random seeds are fixed per experiment and recorded in the code. All scripts, the JSON results of the reported runs, the raw console logs and a script downloading the two public data sets are available at \url{https://github.com/mehrdad-pournaderi/Sketched-Calibration-for-Conformal-Prediction-under-Covariate-Shift}.

\paragraph{E6--E8.} E6 uses the data of E1 ($\mu\in\{1.5,2.5\}$, seed $11$) and of E3 (seed $5$); the full-score rules use the least-squares fit on all covariates, and the residual sketch is the least-squares direction of the cross-fitted absolute residuals (two halves of the training fold) on $X$. E7 uses $d=40$, orthonormal $b_1,b_2,b_3$ and $g$ drawn once (seed $0$), $\mu\in\{1.5,2.5\}$, $n\in\{1000,4000,16000\}$, $m\in\{200,500\}$, $100$ replications (seed $31$), $2\cdot10^4$ target test points, $2000+2000$ unlabelled covariates for the shift direction and $n_t=100$ labelled target points for selection; the predictor is \texttt{HistGradientBoostingRegressor} (scikit-learn, $300$ iterations, learning rate $0.1$), with cross-fitted residuals from two $200$-iteration fits on the halves; sliced inverse regression uses $10$ slices; sketched ratios are the Gaussian closed forms for the realised sketches. E8 uses the California housing data (scikit-learn) and the UCI protein tertiary structure data (CASP); covariates are mapped to normal scores by a quantile transform ($1000$ quantiles), winsorised at $\pm3$, and standardised; each replication (seed $5$, $100$ replications) draws $3000$ training and $m$ calibration points uniformly from the pool, splits the remainder in two, uses one half as the unlabelled pool for estimating $w_Z$ by gradient-boosting regression of the known ratio on the sketch ($200$ iterations), and draws $2000$ test and $100$ selection points from the other half with probabilities proportional to $e^{\gamma\langle a,x\rangle}$ (with replacement); the reported $\hat\chi^2_Z$ is the empirical dispersion of the estimated $w_Z$ over the unlabelled pool; on a finite pool the exact $\E[w\mid Z]$ of a generic linear sketch equals $w$, so this is a property of the smoothed estimate, not of a population divergence (Section~\ref{sec:exp}). In E7 the selected rule minimises $|\widehat{\mathrm{Cov}}_j-0.9|$ over nine candidates: full weighting, clipping, the unweighted rule, the residual sketches with $k=1$ and $k=2$, the mean-plus-residual sketch, the estimated shift direction, the random sketch, and the two-dimensional residual-plus-shift sketch omitted from Table~\ref{tab:nonlinear}; in E8 the candidates are the seven rules listed in Tables~\ref{tab:california}--\ref{tab:protein} other than the selected and hindsight rows, together with the tilt-direction sketch (identical to full weighting). ``Best candidate in hindsight'' minimises the same quantity over the same family with the realised test coverage.

\bibliographystyle{plain}
\bibliography{references}

\begin{thebibliography}{10}

\bibitem{barber2021}
Rina~Foygel Barber, Emmanuel~J. Cand{\`e}s, Aaditya Ramdas, and Ryan~J.
  Tibshirani.
\newblock The limits of distribution-free conditional predictive inference.
\newblock {\em Information and Inference: A Journal of the IMA},
  10(2):455--482, 2021.

\bibitem{brillinger1983}
David~R. Brillinger.
\newblock A generalized linear model with ``{G}aussian'' regressor variables.
\newblock In P.~J. Bickel, K.~A. Doksum, and J.~L. Hodges, editors, {\em A
  Festschrift for {E}rich {L}. {L}ehmann}, pages 97--114. Wadsworth, Belmont,
  CA, 1983.

\bibitem{cannings2017}
Timothy~I. Cannings and Richard~J. Samworth.
\newblock Random-projection ensemble classification.
\newblock {\em Journal of the Royal Statistical Society: Series B},
  79(4):959--1035, 2017.

\bibitem{cook1998}
R.~Dennis Cook.
\newblock {\em Regression Graphics: Ideas for Studying Regressions Through
  Graphics}.
\newblock Wiley, New York, 1998.

\bibitem{dasgupta2003}
Sanjoy Dasgupta and Anupam Gupta.
\newblock An elementary proof of a theorem of {J}ohnson and {L}indenstrauss.
\newblock {\em Random Structures \& Algorithms}, 22(1):60--65, 2003.

\bibitem{devroye2018}
Luc Devroye, Abbas Mehrabian, and Tommy Reddad.
\newblock The total variation distance between high-dimensional {G}aussians
  with the same mean, 2018.
\newblock arXiv:1810.08693.

\bibitem{gebelein1941}
Hans Gebelein.
\newblock Das statistische {P}roblem der {K}orrelation als {V}ariations- und
  {E}igenwertproblem und sein {Z}usammenhang mit der {A}usgleichsrechnung.
\newblock {\em Zeitschrift f{\"u}r Angewandte Mathematik und Mechanik},
  21(6):364--379, 1941.

\bibitem{gibbs2023}
Isaac Gibbs, John~J. Cherian, and Emmanuel~J. Cand{\`e}s.
\newblock Conformal prediction with conditional guarantees, 2023.
\newblock arXiv:2305.12616.

\bibitem{gretton2009}
Arthur Gretton, Alex Smola, Jiayuan Huang, Marcel Schmittfull, Karsten
  Borgwardt, and Bernhard Sch{\"o}lkopf.
\newblock Covariate shift by kernel mean matching.
\newblock In {\em Dataset Shift in Machine Learning}. MIT Press, Cambridge, MA,
  2009.

\bibitem{guan2023}
Leying Guan.
\newblock Localized conformal prediction: a generalized inference framework for
  conformal prediction.
\newblock {\em Biometrika}, 110(1):33--50, 2023.

\bibitem{isaev2026}
R.~Isaev, A.~Conrad, Denis Belomestny, Eric Moulines, and Sergey Samsonov.
\newblock Conformalized quantile regression and minimax limits of fixed-score
  calibration under known covariate shift, 2026.
\newblock arXiv:2609.24929.

\bibitem{jl1984}
William~B. Johnson and Joram Lindenstrauss.
\newblock Extensions of {L}ipschitz mappings into a {H}ilbert space.
\newblock {\em Contemporary Mathematics}, 26:189--206, 1984.

\bibitem{joshi2025}
Sunay Joshi, Shayan Kiyani, George Pappas, Edgar Dobriban, and Hamed Hassani.
\newblock Conformal inference under high-dimensional covariate shifts via
  likelihood-ratio regularization.
\newblock In {\em Advances in Neural Information Processing Systems},
  volume~38, 2025.
\newblock arXiv:2502.13030.

\bibitem{lancaster1957}
H.~O. Lancaster.
\newblock Some properties of the bivariate normal distribution considered in
  the form of a contingency table.
\newblock {\em Biometrika}, 44(1/2):289--292, 1957.

\bibitem{lei2021}
Lihua Lei and Emmanuel~J. Cand{\`e}s.
\newblock Conformal inference of counterfactuals and individual treatment
  effects.
\newblock {\em Journal of the Royal Statistical Society: Series B},
  83(5):911--938, 2021.

\bibitem{li1991}
Ker-Chau Li.
\newblock Sliced inverse regression for dimension reduction.
\newblock {\em Journal of the American Statistical Association},
  86(414):316--327, 1991.

\bibitem{maillard2012}
Odalric-Ambrym Maillard and R{\'e}mi Munos.
\newblock Linear regression with random projections.
\newblock {\em Journal of Machine Learning Research}, 13:2735--2772, 2012.

\bibitem{papadopoulos2002}
Harris Papadopoulos, Kostas Proedrou, Volodya Vovk, and Alex Gammerman.
\newblock Inductive confidence machines for regression.
\newblock In {\em Machine Learning: ECML 2002}, pages 345--356. Springer, 2002.

\bibitem{polo2022}
Felipe~Maia Polo and Renato Vicente.
\newblock Effective sample size, dimensionality, and generalization in
  covariate shift adaptation, 2022.
\newblock arXiv:2010.01184v5.

\bibitem{pournaderi2026}
Mehrdad Pournaderi.
\newblock Sharp training-conditional coverage for conformal prediction under
  covariate shift, 2026.
\newblock arXiv:2609.33456.

\bibitem{rpconformal2024}
X.~Qian, J.~Wu, L.~Wei, and Y.~Lin.
\newblock Random projection ensemble conformal prediction for high-dimensional
  classification.
\newblock {\em Chemometrics and Intelligent Laboratory Systems}, 2024.
\newblock doi:10.1016/j.chemolab.2024.105225.

\bibitem{raskutti2016}
Garvesh Raskutti and Michael~W. Mahoney.
\newblock A statistical perspective on randomized sketching for ordinary
  least-squares.
\newblock {\em Journal of Machine Learning Research}, 17(213):1--31, 2016.

\bibitem{renyi1959}
Alfr{\'e}d R{\'e}nyi.
\newblock On measures of dependence.
\newblock {\em Acta Mathematica Academiae Scientiarum Hungaricae}, 10:441--451,
  1959.

\bibitem{romano2019}
Yaniv Romano, Evan Patterson, and Emmanuel~J. Cand{\`e}s.
\newblock Conformalized quantile regression.
\newblock In {\em Advances in Neural Information Processing Systems},
  volume~32, 2019.

\bibitem{sesia2022}
Matteo Sesia and Stefano Favaro.
\newblock Conformal frequency estimation with sketched data.
\newblock In {\em Advances in Neural Information Processing Systems},
  volume~35, 2022.

\bibitem{stojanov2019}
Petar Stojanov, Mingming Gong, Jaime~G. Carbonell, and Kun Zhang.
\newblock Low-dimensional density ratio estimation for covariate shift
  correction.
\newblock In {\em Proceedings of the 22nd International Conference on
  Artificial Intelligence and Statistics (AISTATS)}, volume~89 of {\em
  Proceedings of Machine Learning Research}, Naha, Okinawa, Japan, 2019.

\bibitem{sugiyama2012}
Masashi Sugiyama, Taiji Suzuki, and Takafumi Kanamori.
\newblock {\em Density Ratio Estimation in Machine Learning}.
\newblock Cambridge University Press, Cambridge, 2012.

\bibitem{tibshirani2019}
Ryan~J. Tibshirani, Rina~Foygel Barber, Emmanuel~J. Cand{\`e}s, and Aaditya
  Ramdas.
\newblock Conformal prediction under covariate shift.
\newblock In {\em Advances in Neural Information Processing Systems},
  volume~32, 2019.

\bibitem{vovk2012}
Vladimir Vovk.
\newblock Conditional validity of inductive conformal predictors.
\newblock In {\em Proceedings of the Asian Conference on Machine Learning
  (ACML)}, volume~25 of {\em Proceedings of Machine Learning Research}, pages
  475--490, 2012.

\bibitem{vovk2005}
Vladimir Vovk, Alex Gammerman, and Glenn Shafer.
\newblock {\em Algorithmic Learning in a Random World}.
\newblock Springer, New York, 2005.

\bibitem{wang2026}
J.~Wang and S.~Goel.
\newblock Weight clipping for robust conformal inference under unbounded
  covariate shifts, 2026.
\newblock arXiv:2605.02072.

\bibitem{xu2025}
Rui Xu, Chao Chen, Yue Sun, Parvathinathan Venkitasubramaniam, and Sihong Xie.
\newblock {W}asserstein-regularized conformal prediction under general
  distribution shift.
\newblock In {\em International Conference on Learning Representations (ICLR)},
  2025.
\newblock arXiv:2501.13430.

\bibitem{yang2024}
Yachong Yang, Arun~Kumar Kuchibhotla, and Eric Tchetgen~Tchetgen.
\newblock Doubly robust calibration of prediction sets under covariate shift.
\newblock {\em Journal of the Royal Statistical Society: Series B},
  86(4):943--965, 2024.

\end{thebibliography}

\end{document}